\documentclass[unsortedaddress,superscriptaddress,pra,notitlepage]{revtex4}

\usepackage{amsmath,amssymb}
\usepackage{braket}
\usepackage{theorem}
\usepackage{color}
\usepackage{amsfonts}
\usepackage[mathscr]{eucal}
\usepackage[hidelinks]{hyperref}
\usepackage{graphicx}
\usepackage{xcolor}
\usepackage{tikz}
\usepackage{scalerel}
\usepackage{bm}
\usepackage[new]{old-arrows}

\usetikzlibrary{calc,decorations.pathreplacing}
\usetikzlibrary{arrows,shapes}

\newtheorem{definition}{Definition}[section]
\newtheorem{defin}[definition]{Definition}
\newtheorem{theorem}[definition]{Theorem}
\newtheorem{thm}[definition]{Theorem}
\newtheorem{prop}[definition]{Proposition}
\newtheorem{proposition}[definition]{Proposition}
\newtheorem{lemma}[definition]{Lemma}
\newtheorem{remark}[definition]{Remark}
\newtheorem{rem}[definition]{Remark}
\newtheorem{corollary}[definition]{Corollary}
\newtheorem{cor}[definition]{Corollary}

\newtheorem{example}[definition]{Example}

\newtheorem{ex}[definition]{Example}

\newtheorem{cond}[definition]{Condition}

\newenvironment{proof}[1][Proof]{\begin{trivlist}
\item[\hskip \labelsep {\bfseries #1}]}{\hfill$\Box$\end{trivlist}}

\newcommand{\nn}{\nonumber \\}

\def\vfi{\varphi}
\def\theta{\vartheta}
\def\hil{{\mathcal H}}
\def\kil{{\mathcal K}}

\def\A{{\mathcal A}}

\def\B{{\mathcal B}}

\def\E{{\mathcal E}}
\def\F{{\mathcal F}}

\def\M{\mathcal{M}}

\def\S{{\mathcal S}}

\def\half{\frac{1}{2}}
\def\imp{\implies}

\def\ep{\varepsilon}

\def\bN{\mathbb{N}}
\def\bC{\mathbb{C}}

\def\bR{\mathbb{R}}
\def\bZ{\mathbb{Z}}

\def\bP{\mathbb{P}}
\def\bT{\mathbb{T}}

\def\bz{\left(}
\def\jz{\right)}

\def\inv{^{-1}}

\def\egy{\mathbf 1}
\def\map{\Phi}

\def\sa{\mathrm{sa}}

\def\N{\mathcal{N}}

\def\rho{\varrho}
\def\povm{\mathrm{POVM}}

\def\nn{\nonumber}

\def\ptni{\mathrm{PTNI}}

\def\meas{\mathrm{meas}}

\def\p{_{\ge 0}}
\def\pne{_{\gneq 0}}

\def\valt{\cdot}

\def\petz{\mathrm{P}}
\def\ccf{(\bR^{(0,+\infty)})_{\conv}\cup(\bR^{(0,+\infty)})_{\conc}}
\def\sing{\mathrm{sing}}

\def\tv{\mathrm{TV}}
\def\hs{\mathrm{hs}}
\def\tr{\mathrm{tr}}
\def\fM{\mathfrak{M}}
\def\proj{\mathrm{proj}}

\newcommand{\ki}[1]{\emph{#1}}

\newcommand{\s}{\mbox{ }}
\newcommand{\ds}{\mbox{ }\mbox{ }}

\newcommand{\norm}[1]{\left\| #1\right\|}

\newcommand{\inner}[2]{\left\langle #1 , #2\right\rangle}

\newcommand{\abs}[1]{\left| #1 \right|}

\newcommand{\diad}[2]{\left|#1\right\rangle\!\left\langle #2\right|}

\newcommand{\pr}[1]{\diad{#1}{#1}}

\newcommand{\derleft}[1]{\partial^{-} #1}
\newcommand{\derright}[1]{\partial^{+} #1}

\newcommand{\wtilde}[1]{\widetilde{#1}}

\newcommand{\persp}[1]{\mathcal{P}_{#1}}
\newcommand{\hsp}[2]{D_{(\id-#1)_+}^{#2}}
\newcommand{\hsn}[2]{D_{(\id-#1)_-}^{#2}}
\newcommand{\hsht}[2]{E_{#1}^{#2}}
\newcommand{\tn}[2]{D_{|\id-#1|}^{#2}}
\newcommand{\hsd}[2]{D_{#1}^{#2,\mathrm{hs}}}
\newcommand{\hsq}[2]{Q_{#1}^{#2,\mathrm{hs}}}
\newcommand{\hsr}[2]{D_{#1}^{#2,\mathrm{hs}}}
\newcommand{\fdiv}[2]{D_{#1}^{#2}}
\newcommand{\nsd}[2]{\nu_{#1|#2}}
\newcommand{\nsdd}[2]{\nu_{#2|#1}}
\newcommand{\erri}[4]{\beta_{0,#1}^{#2}(#3\|#4)}
\newcommand{\errii}[4]{\beta_{1,#1}^{#2}(#3\|#4)}
\newcommand{\errm}[4]{\beta_{\mathrm{mix},#1}^{#2}(#3\|#4)}
\newcommand{\np}[4]{N_{#1}^{#2}(#3\|#4)}
\newcommand{\rts}[2]{\mathrm{RT}(#1\|#2)}
\newcommand{\test}[1]{\mathbb{T}(#1)}
\newcommand{\hahnp}[2]{#2^+_{#1}}
\newcommand{\hahnm}[2]{#2^-_{#1}}
\newcommand{\hahnpt}[2]{\wtilde{#2}^+_{#1}}
\newcommand{\hahnmt}[2]{\wtilde{#2}^-_{#1}}

\newcommand{\dirac}[1]{\delta_{#1}}

\makeatletter
\renewcommand{\p@enumii}{}
\makeatother

\DeclareMathOperator{\id}{id}
\DeclareMathOperator{\Tr}{Tr}

\DeclareMathOperator{\supp}{supp}

\DeclareMathOperator{\ran}{ran}

\DeclareMathOperator{\spec}{spec}

\DeclareMathOperator{\rt}{\Gamma}

\DeclareMathOperator{\conv}{conv}
\DeclareMathOperator{\conc}{conc}

\DeclareMathOperator{\symdiff}{\bigtriangleup}

\DeclareMathOperator{\invv}{inv}

\begin{document}

\title{Hockey stick $f$-divergences}

\author{Fumio Hiai}
\email{hiai.fumio@gmail.com}

\affiliation{
Graduate School of Information Sciences, Tohoku University, \\
Aoba-ku, Sendai 980-8579, Japan
}

\affiliation{Department of Analysis and Operations Research, Institute of Mathematics,
Budapest University of Technology and Economics,
M\H uegyetem rkp.~3., H-1111 Budapest, Hungary}

\author{Mil\'an Mosonyi}
\email{milan.mosonyi@gmail.com}

\affiliation{Department of Analysis and Operations Research, Institute of Mathematics,
Budapest University of Technology and Economics,
M\H uegyetem rkp.~3., H-1111 Budapest, Hungary}

\author{Marco Tomamichel}
\email{marco.tomamichel@nus.edu.sg}

\affiliation{Center for Quantum Technologies, National University of Singapore
Singapore 117543}

\affiliation{Department of Electrical and Computer Engineering, National University of Singapore, Singapore 117583}

\begin{abstract}
In this paper we give a systematic and unified treatment and extensions of various results on a new notion of quantum  $f$-divergences defined from quantum hockey stick divergences, the theory of which has been developed recently in 
\cite{BHT_fdiv,HircheTomamichel_integral,LiuHircheCheng2025}.
In particular, we consider non-normalized states and hockey stick $f$-divergences defined from more general notions of quantum hockey stick divergences, as well as a somewhat more general form of the integral representation defined in terms of an additional real parameter. We also consider the extension of the theory to general von Neumann algebras,
and extend various results from \cite{HircheTomamichel_integral,LiuHircheCheng2025} to this setting. Our main results here are the representation of the hockey stick $f$-divergences in terms of 
Neyman-Pearson error probabilities, 
which was given in the finite-dimensional case in \cite{LiuHircheCheng2025}, an extension of Jen\v cov\'a's result \cite{Jencova2023} on the 
detection of reversibility of a quantum channel on a pair of states in terms of the hockey stick divergences, and an extension of a result in 
\cite{HircheTomamichel_integral} showing that the regularized hockey stick R\'enyi $\alpha$-divergences coincide with the Petz-type R\'enyi divergences for $\alpha\in(0,1)$ and with the
sandwiched R\'enyi divergences for $\alpha>1$. Moreover, we give some partial results on the characterization of when different notions of quantum $f$-divergences
give the same value on a pair of quantum states.
\end{abstract}

\maketitle
\tableofcontents

\section{Introduction}

Various quantifiers of the difference of states of a physical system play a fundamental role in physics and in information theory, both as
operationally motivated measures of statistical distinguishability of states (divergences), as well as parent quantities for 
measures of information and correlations in states (entropies, channel capacities, etc.). In classical information theory, a general theory of a large class of divergences has been developed by Csisz\'ar \cite{Csiszar-fdiv}, and Ali and Silvey \cite{AliSilvey}.
In this theory, a divergence is defined for every convex function $f$ on the positive reals, called the (classical) $f$-divergence. Special cases of these divergences include 
the variational distance, the Kullback-Leibler divergence (relative entropy), and (an exponential function of) the R\'enyi divergences \cite{Renyi}. 
Due to the non-commuting nature of density operators describing the states of a quantum system, the classical $f$-divergences have different extensions to distinguishability
measures of quantum states, i.e., quantum $f$-divergences. Apart from the operationally most relevant version due to Petz \cite{P85,P86}, the measured (minimal) and the maximal $f$-divergences
\cite{Matsumoto_newfdiv} have been the most intensively studied ones in the literature. On top of these general notions of quantum $f$-divergences, a plethora of quantum extensions 
of $f$-divergences corresponding to specific $f$ functions have been studied in quantum information theory, e.g., various notions of quantum relative entropies
\cite{BS,mosonyi2022geometric} and quantum R\'enyi divergences \cite{AD,FawziFawzi2021,mosonyi2022geometric,Renyi_new,WWY}.

Recently, Frenkel \cite{Frenkel_integral} gave a remarkable integral decomposition of the Umegaki relative entropy \cite{Umegaki}, which, after some 
simple rewriting \cite{Jencova2023,HircheTomamichel_integral}, turns out to be a quantum analogue of the integral decomposition of the classical relative entropy 
into so-called hockey stick divergences. This was then the basis of defining new quantum R\'enyi divergences \cite{Frenkel_integral,HircheTomamichel_integral}, and 
more general quantum $f$-divergences \cite{HircheTomamichel_integral}. These are defined by choosing a suitable notion of quantum hockey stick divergences (the measured ones, as we will see below)
and mimicking the known decomposition of the classical $f$-divergences into an integral of hockey stick divergences. These give new quantum $f$-divergences in the sense that they do not coincide with 
any one of the quantum $f$-divergences studied in the literature so far. Further properties of these new divergences and alternative representations were then given in 
\cite{BHT_fdiv,LiuHircheCheng2025}.

In this paper we give a systematic and unified treatment and extensions of various results from \cite{BHT_fdiv,HircheTomamichel_integral,LiuHircheCheng2025}.
In particular, we consider non-normalized states and hockey stick $f$-divergences defined from more general notions of quantum hockey stick divergences, as well as a somewhat more general form of the integral representation defined in terms of an additional real parameter (Section \ref{sec:finitedim}.) We also consider the extension of the theory to general von Neumann algebras
(Section \ref{sec:V}), 
and extend various results from \cite{HircheTomamichel_integral,LiuHircheCheng2025} to this setting. Our main results here are the representation of the hockey stick $f$-divergences in terms of 
Neyman-Pearson error probabilities (Sections \ref{sec:V.A}, \ref{sec:V.B}), 
which was given in the finite-dimensional case in \cite{LiuHircheCheng2025}, an extension of Jen\v cov\'a's result \cite{Jencova2023} on the 
detection of reversibility of a quantum channel on a pair of states in terms of the hockey stick divergences (Section \ref{sec:V.E}), and an extension of a result in 
\cite{HircheTomamichel_integral} showing that the regularized hockey stick R\'enyi $\alpha$-divergences coincide with the Petz-type R\'enyi divergences for $\alpha\in(0,1)$ and with the
sandwiched R\'enyi divergences for $\alpha>1$ (Sections \ref{sec:VI.B}--\ref{sec:VI.D}). Moreover, we give some partial results on the characterization of when different notions of quantum $f$-divergences
give the same value on a pair of quantum states (Sections \ref{sec:VII.A}--\ref{sec:VII.C}). While the general von Neumann algebra case contains the finite-dimensional quantum case as well as the classical case, we discuss these separately for didactic reasons.

Finally, we note that hockey stick $f$-divergences have also been studied in the independent and concurrent work \cite{daSilva2026} very recently, while 
Frenkel's integral representation has also been considered in the von Neumann algebra setting in the papers \cite{vLuWi2026,Kossmann}.

\section{Preliminaries}

We will use the notations 
\begin{align*}
(\bR^{(0,+\infty)})_{\conv}&:=\left\{f:\,(0,+\infty)\to\bR\text{ convex}\right\},\\
(\bR^{(0,+\infty)})_{\conc}&:=\left\{f:\,(0,+\infty)\to\bR\text{ concave}\right\}.
\end{align*}
For a function $f\in\ccf$, $\partial^-f(x)$ and $\partial^+f(x)$ will denote the left and the right 
derivatives, respectively, at some $x\in(0,+\infty)$. For any function $g\in\bR^{(0,+\infty)}$,
$g(x^-)$ will denote the left limit of $g$ at some $x\in(0,+\infty)$, provided that it exists, and similarly, 
$g(x^+)$ will denote the right limit at some $x\in[0,+\infty)$, if it exists.

By $\log$ we will denote the natural logarithm, with its extension to $[0,+\infty]$ as 
$\log 0:=-\infty$, $\log+\infty:=+\infty$.
For a natural number $n\in\bN=\{1,2,\ldots\}$, we will use the notation
$[n]:=\{1,2,\ldots,n\}$.

By a Hilbert space we always mean a complex separable Hilbert space.
We will denote the inner product on a Hilbert space by $\inner{\valt}{\valt}$ and follow the convention that 
it is linear in its second and conjugate linear in its first variable. We will also use the Dirac notation:
for any vectors $x,y$ in a Hilbert space $\hil$, the operator $\diad{y}{x}$ is defined by $\diad{y}{x}z:=\inner{x}{z}y$, $z\in\hil$.

For a linear operator $A$ on a Hilbert space $\hil$, we will use the notations 
$\norm{A}:=\norm{A}_{\infty}:=\sup\{\norm{A\psi}:\,\psi\in\hil,\,\norm{\psi}\le 1\}$ for the operator norm, and 
$\B(\hil):=\{A:\,\hil\to\hil\text{ linear },\norm{A}_{\infty}<+\infty\}$ will denote the set of all bounded linear operators on $\hil$. 
We will use the notation
$\B(\hil)_{\sa}$ for the set of self-adjoint operators on $\hil$.
For an interval $J\subseteq\bR$, 
$\B(\hil)_{J}:=\{A\in\B(\hil)_{\sa}:\,\spec(A)\subseteq J\}$, i.e., 
it is the set of self-adjoint operators on $\hil$ with their spectra in $J$.
We will use the shorthand notations $\B(\hil)_{\ge 0}:=\B(\hil)_{[0,+\infty)}$
for the set of positive semi-definite (PSD) operators on $\hil$, and 
$\B(\hil)_{>0}:=\B(\hil)_{(0,+\infty)}$ for the set of 
positive definite operators. An inequality $A\le B$ between operators $A,B\in\B(\hil)$ is always interpreted in the L\"owner (or PSD) order, meaning $B-A\in\B(\hil)\p$. Elements of the set
\begin{align*}
\bT(\hil):=\B(\hil)_{[0,1]}:=\{T\in\B(\hil)_{\sa}:\,0\le T\le I\}
\end{align*}
are called \ki{tests} on $\hil$.
The set of (orthogonal) projections on $\hil$ will be denoted by 
$\bP(\hil):=\B(\hil)_{\{0,1\}}=\{P\in\B(\hil)_{\sa}:\,P^2=P\}$.
For a positive semi-definite operator $A\in\B(\hil)\p$, we will use the notation
\begin{align}\label{eq:supppr}
A^0:=\lim_{t\searrow 0}A^t
\end{align}
for the projection onto $\supp A:=(\ker A)^{\perp}$. 

The set of \ki{states} (density operators) on a finite-dimensional Hilbert space $\hil$ is
$\S(\hil):=\{\rho\in\B(\hil)\p:\,\Tr\rho=1\}$.
For $r\in\bN$, the set of $r$-outcome positive operator valued measures (POVMs) on $\hil$ are defined as
\begin{align*}
\povm(\hil,[r]):=\left\{(M_i)_{i=1}^r\in\B(\hil)\p^r:\,\sum_{i=1}^r M_i=I\right\}.
\end{align*}

\section{Classical $f$-divergences}

\label{sec:classical}


\subsection{The perspective function}

For any function $f:(0,+\infty)\to\bR$, 
its \ki{perspective} $P_f:\,(0,+\infty)\times (0,+\infty)\to\bR$ is defined by
\begin{align}\label{persp def}
\persp{f}(x,y):=yf\bz\frac{x}{y}\jz,\ds\ds\ds x,y\in (0,+\infty).
\end{align}
By definition, $f(x)=\persp{f}(x,1)$ for all $x\in(0,+\infty)$, and
the \ki{transpose} $\tilde f$ of $f$ is defined as
\begin{align*}
\tilde f(y):=\persp{f}(1,y)=y f\bz\frac{1}{y}\jz,\ds\ds\ds y\in (0,+\infty).
\end{align*}
Thus, $f$ and $\tilde f$ can be considered as marginals of the two-variable function $\persp{f}$.
Moreover, we have
\begin{align}\label{perpective symm}
\persp{\tilde f}(x,y)=\persp{f}(y,x),\ds\ds\ds x,y\in(0,+\infty).
\end{align}

Note that $\persp{f}$ is \ki{homogeneous of degree $1$}, i.e.,
\begin{align*}
\persp{f}(tx,ty)=t\persp{f}(x,y),\ds\ds\ds x,y,t\in(0,+\infty).
\end{align*}
Moreover, every $1$-homogeneous function $P$ on $(0,+\infty)^2$ is the perspective of its first marginal
$f(x):=P(x,1)$, as for this $f$,
\begin{align*}
\persp{f}(x,y)=yf\bz\frac{x}{y}\jz=y P\bz\frac{x}{y},1\jz=P(x,y).
\end{align*}

%

We will be exclusively interested in functions that satisfy the following conditions:

\begin{cond}\label{cond:fdiv}
$f:\,(0,+\infty)\to\bR$ is a continuous function such that 
the limits 
\begin{align}\label{f limits}
f(0^+):=\lim_{x\searrow 0}f(x)\ds\ds\text{and}\ds\ds
\tilde f(0^+):=\lim_{y\searrow 0}\tilde f(y)=\lim_{x\to+\infty}\frac{f(x)}{x}
\end{align}
exist in $\bR\cup\{\pm\infty\}$, and they are not both infinity with opposite signs. 
\end{cond}

If $f$ satisfies Condition \ref{cond:fdiv}, we can extend $\persp{f}$ to $[0,+\infty)\times[0,+\infty)$ by
\begin{align}\label{persp extension}
\persp{f}(x,y):=\lim_{\ep\searrow 0}(y+\ep)f\bz\frac{x+\ep}{y+\ep}\jz
=
\begin{cases}
yf\bz\frac{x}{y}\jz, & \text{if $x,y>0$}, \\
yf(0^+), & \text{if $x=0$}, \\
x\tilde f(0^+), & \text{if $y=0$},
\end{cases}
\end{align}
with the convention $0\cdot\infty:=0$. In particular, this extension is consistent with the original definition in 
the sense that $\persp{f}(x,y)=yf(x/y)$ 
when $x,y\in(0,+\infty)$, 
according to both \eqref{persp def} and \eqref{persp extension}.

\begin{lemma}\label{lemma:persp additive}
Let $f_1,f_2$ satisfy Condition \ref{cond:fdiv}. For any 
$c_1,c_2\in\bR$, and any
$x,y\in[0,+\infty)$, 
\begin{align*}
\persp{c_1f_1+c_2f_2}(x,y)=c_1\persp{f_1}(x,y)+c_2\persp{f_2}(x,y),
\end{align*}
provided that the RHS makes sense (i.e., it is not a sum of infinities of opposite signs).
\end{lemma}
\begin{proof}
By definition,  
\begin{align*}
\persp{c_1f_1+c_2f_2}(x+\ep,y+\ep)
=
c_1\persp{f_1}(x+\ep,y+\ep)
+
c_2\persp{f_2}(x+\ep,y+\ep),
\end{align*}
for any $\ep>0$,
and taking the limit $\ep\searrow 0$ yields the assertion.
\end{proof}

We will mainly be interested in the case when $f$ is convex or concave.
The following is well known and easy to verify:

\begin{lemma}\label{lemma:persp convex}
Let $f:\,(0,+\infty)\to\bR$. The following are equivalent:
\begin{enumerate}
\item
$f$ is convex (concave);
\item
$\tilde f$ is convex (concave);
\item
$\persp{f}$ is convex (concave).
\end{enumerate}
\end{lemma}

If $f\in\ccf$ then $f$ is known to be continuous on $(0,+\infty)$, and 
by Lemma \ref{lemma:convex f tilde limit} below, it satisfies also the other points in Condition \ref{cond:fdiv}. 

As a preparation for this and other subsequent proofs, we recall that if $f:\,(0,+\infty)\to\bR$ is convex 
then $\derleft{f}(x)$, $\derright{f}(x)$ exist at every $x\in(0,+\infty)$, the functions 
$\derleft{f}$, $\derright{f}$ are monotone increasing and left (resp.~right) continuous, and 
at any $x\in(0,+\infty)$,
\begin{align}\label{eq:convexf partial der order2}
\derleft{f}(x^-)=\derright{f}(x^-)=\derleft{f}(x)\le\derright{f}(x)=
\derleft{f}(x^+)=\derright{f}(x^+).
\end{align}
Thus, for any $x\in(0,+\infty)$,
\begin{align*}
\exists\s f'(x)
\ds\iff\ds 
\derleft{f}(x)=\derright{f}(x)
\ds\iff\ds
\derleft{f}\text{ is continuous at }x
\ds\iff\ds
\derright{f}\text{ is continuous at }x.
\end{align*}
We may define
\begin{align}\label{eq:fder def}
f'(x):=\half(\derleft{f}(x)+\derright{f}(x)),\ds\ds\ds x\in(0,+\infty),
\end{align}
which coincides with the usual derivative of $f$ at every point where the latter exists.  
Then 
\begin{align*}
f'(x^-)=\derleft{f}(x),\ds\ds\ds
f'(x^+)=\derright{f}(x),\ds\ds\ds
x\in(0,+\infty).
\end{align*}
Since $f'$ is monotone increasing, it has an associated Lebesgue-Stieltjes measure $df'$, which is independent 
of how we extend $f'$ to the discontinuity points of $\derleft{f}$, and is uniquely specified by 
\begin{align*}
df'([a,b))=f'(b^-)-f'(a^-),\ds\ds\ds 0<a<b<+\infty.
\end{align*}
Moreover, since $f'$ is monotone, the second derivative $f''$ of $f$ exists almost everywhere, and 
\begin{align}\label{eq:Alex}
df'=f''\,d\lambda+\underbrace{\sum_{x\in\bR}\left[\partial^+f(x)-\partial^-f(x)\right]\,\dirac{x}}_{=:(df')_{\sing}},
\end{align}
where $\lambda$ is the Lebesgue measure, and $\dirac{x}$ is the Dirac measure concentrated at $\{x\}$. 
It is well known that 
\begin{align}\label{eq:Taylor1 remainder}
f(x)=f(a)+\int_a^x f'(t)\,dt,\ds\ds\ds a,x\in(0,+\infty).
\end{align}
All these hold without alteration for a concave $f$, with the exception that in this case
$\derleft{f}$ and $\derright{f}$ are monotone decreasing, 
and the inequality in \eqref{eq:convexf partial der order2} holds in the opposite direction.

\begin{lemma}\label{lemma:convex f tilde limit}
If $f\in(\bR^{(0,+\infty)})_{\conv}\cup(\bR^{(0,+\infty)})_{\conc}$  then 
the limits in \eqref{f limits} exist in $(-\infty,+\infty]$, and 
\begin{align*}
\tilde f(0^+)
=
\lim_{x\to+\infty}\partial^-f(x)
=
\lim_{x\to+\infty}\partial^+f(x)
=
\lim_{x\to+\infty}f'(x).
\end{align*}
\end{lemma}
\begin{proof}
It is sufficient to prove the assertion for the case when $f$ is convex.
If $f$ is monotone increasing then clearly $f(0^+)=\inf_{x\in(0,+\infty)}f(x)$. Otherwise there exists a point $x_0\in(0,+\infty)$ such that $f$ is monotone decreasing on 
$(0,x_0]$ and monotone increasing on $(x_0,+\infty)$, whence
$f(0^+)=\sup_{x\in(0,x_0]}f(x)$. 

Since $x\mapsto \partial^-f(x)$ and $x\mapsto \partial^+f(x)$ are increasing, their limits exist at 
$+\infty$, and, moreover, 
$\lim_{x\to+\infty}\partial^-f(x)=
\lim_{x\to+\infty}\partial^+f(x)=\lim_{x\to+\infty}f'(x)$, according to \eqref{eq:convexf partial der order2}. By \eqref{eq:Taylor1 remainder},
\begin{align*}
f(x)=f(x_0)+\int_{x_0}^x f'(t)\,dt
\begin{cases}
\ge f(x_0)+(x-x_0) f'(x_0),\\
\le f(x_0)+(x-x_0) f'(x),
\end{cases}
\ds\ds\ds x\in(0,+\infty).
\end{align*}
From the lower bound we get 
\begin{align*}
\liminf_{x\to+\infty}\frac{f(x)}{x}
\ge 
\underbrace{\lim_{x\to+\infty}\frac{f(x_0)}{x}}_{=0}
+
f'(x_0)\underbrace{\lim_{x\to+\infty}\frac{x-x_0}{x}}_{=1}.
\end{align*}
Taking then the limit $x_0\to+\infty$, we get 
$\liminf_{x\to+\infty}\frac{f(x)}{x}\ge \lim_{x\to+\infty}f'(x)$.
Similarly, the upper bound above yields
$\limsup_{x\to+\infty}\frac{f(x)}{x}\le \lim_{x\to+\infty}f'(x)$.
Hence, $\tilde f(0^+)$ exists, and is equal to 
$\lim_{x\to+\infty}f'(x)$ as stated.
\end{proof}

For any real number $x\in\bR$, let
\begin{align*}
x_+:=\frac{|x|+x}{2}=\begin{cases}
x,&x\ge 0,\\
0,&x<0,
\end{cases}
\ds\ds\ds\ds\ds
x_-:=\frac{|x|-x}{2}=\left\{\begin{array}{ll}
-x,&x\le 0,\\
0,&x>0,
\end{array}\right\}
=x_+-x.
\end{align*}

\begin{lemma}\label{lemma:convexf repr}
Let $f\in\ccf$. For every $a,x\in(0,+\infty)$, 
\begin{align}
f(x)&=f(a)+f'(a^-)(x-a)+\int_{(0,a)}(x-t)_-\,df'(t)+\int_{[a,+\infty)}(x-t)_+\,df'(t)
\label{eq:convexf repr3}\\
&=f(a)+f'(a^+)(x-a)+\int_{(0,a]}(x-t)_-\,df'(t)+\int_{(a,+\infty)}(x-t)_+\,df'(t).
\label{eq:convexf repr4}
\end{align}
\end{lemma}
\begin{proof}
Fix $a\in(0,+\infty)$. If $x=a$, both equalities hold trivially.

Assume next that $x>a$. 
In this case the first integral terms in both 
\eqref{eq:convexf repr3} and \eqref{eq:convexf repr4} are equal to $0$.  
By \eqref{eq:Taylor1 remainder}, 
\begin{align*}
f(x)-f(a)
&=
\int_{[a,x]}f'(u)\,du
=
\int_{[a,x]}f'(u^+)\,du
=
\int_{[a,x]}\bz f'(a^-)+df'([a,u])\jz\,du\\
&=
f'(a^-)(x-a)+\int_{[a,x]}\int_{[a,x]}\egy_{[a,u]}(t)\,df'(t)\,du.
\end{align*}
The double integral above can be evaluated using the Fubini-Tonelli theorem as 
\begin{align*}
\int_{[a,x]}\int_{[t,x]}\,du\,df'(t)=
\int_{[a,x]}(x-t)\,df'(t)=\int_{[a,+\infty)}(x-t)_+\,df'(t).
\end{align*} 
This proves \eqref{eq:convexf repr3}, and \eqref{eq:convexf repr4} follows from it by noting that 
\begin{align*}
\int_{[a,+\infty)}(x-t)_+\,df'(t)&=
\int_{(a,+\infty)}(x-t)_+\,df'(t)+(x-a)df'(\{a\})
\end{align*}
and that $df'(\{a\})=f'(a^+)-f'(a^-)$.

Assume finally that $x<a$. 
In this case the second integral terms in both 
\eqref{eq:convexf repr3} and \eqref{eq:convexf repr4} are equal to $0$.  
By \eqref{eq:Taylor1 remainder}, 
\begin{align*}
f(a)-f(x)
&=
\int_{[a,x]}f'(u)\,du
=
\int_{[x,a]}f'(u^-)\,du
=
\int_{[x,a]}\bz f'(a^+)-df'([u,a])\jz\,du\\
&=
f'(a^+)(a-x)-\int_{[x,a]}\int_{[x,a]}\egy_{[u,a]}(t)\,df'(t)\,du.
\end{align*}
The double integral above can be evaluated using the Fubini-Tonelli theorem as 
\begin{align*}
\int_{[x,a]}\int_{[x,t]}\,du\,df'(t)=
\int_{[x,a]}(t-x)\,df'(t)=\int_{(0,a]}(x-t)_-\,df'(t).
\end{align*} 
This proves \eqref{eq:convexf repr3}, and \eqref{eq:convexf repr4} follows from it by noting that 
\begin{align*}
\int_{(0,a]}(x-t)_-\,df'(t)&=
\int_{(0,a)}(x-t)_+\,df'(t)+(a-x)df'(\{a\})
\end{align*}
and that $df'(\{a\})=f'(a^+)-f'(a^-)$.
\end{proof}

\begin{rem}\label{rem:convexf repr}
If $f:\,(0,+\infty)\to\bR$ is a twice continuously differentiable function, then  
for any $a\in(0,+\infty)$, and any $x\in(0,+\infty)$,
\begin{align}
f(x)&=
f(a)+f'(a)(x-a)+\int_{0}^af''(t)(x-t)_-\,dt+\int_a^{+\infty}f''(t)(x-t)_+\,dt
\label{eq:convexf repr1}\\
&=f(a)+f'(a)(x-a)+
\begin{cases}
\int_{0}^af''(t)(x-t)_-\,dt,&x\in(0,a],\\
\int_a^{+\infty}f''(t)(x-t)_+\,dt,&x\in[a,+\infty).
\end{cases}
\label{eq:convexf repr2}
\end{align}
This is a special case of Lemma \ref{lemma:convexf repr} when $f$ is convex or concave, but it also holds without assuming these properties. To see why \eqref{eq:convexf repr1} holds, 
note that $(x-t)_-=0$ for $t\le x$ and $(x-t)_+=0$ for $t\ge x$, whence both integrals 
in \eqref{eq:convexf repr1} are the integral of a continuous function on a compact interval.
Hence, the RHS of \eqref{eq:convexf repr1} is differentiable with respect to $a$, and its derivative is
\begin{align*}
f'(a)-f'(a)+f''(a)(x-a)+f''(a)(x-a)_--f''(a)(x-a)_+=0.
\end{align*}
Thus, the RHS of \eqref{eq:convexf repr1} is independent of $a$, and choosing $a:=x$ makes it equal to $f(x)$, proving the equality in \eqref{eq:convexf repr1}.
The equality to \eqref{eq:convexf repr2} is trivial.
\end{rem}

\begin{rem}
Note that the integral terms in \eqref{eq:convexf repr1} give an expression for the remainder term of the degree one Taylor approximation of $f$ around $a$. A similar-looking alternative integral representation is given by 
\begin{align}\label{eq:Taylor remainder2}
f(x)=
f(a)+f'(a)(x-a)+\int_{a}^xf''(t)(x-t)\,dt,
\end{align}
which one can easily verify using a simple integration by parts.
\end{rem}

%

\begin{lemma}\label{lemma:persp integral repr}
Let $f\in\ccf$. Then 
\begin{align}
\hspace{-.8cm}\persp{f}(x,y)&=
yf(a)+f'(a^-)(x-ay)+\int_{(0,a)}(x-ty)_-\,df'(t)+\int_{[a,+\infty)}(x-ty)_+\,df'(t)
\label{eq:persp integral repr1}\\
&=yf(a)+f'(a^+)(x-ay)+\int_{(0,a]}(x-ty)_-\,df'(t)+\int_{(a,+\infty)}(x-ty)_+\,df'(t)
\label{eq:persp integral repr2}
\end{align}
for any $a\in(0,+\infty)$, and any $x,y\in[0,+\infty)$,
\end{lemma}
\begin{proof}
We only prove \eqref{eq:persp integral repr1}, since its equality to 
\eqref{eq:persp integral repr2} follows easily, in the same way as in the proof of 
Lemma \ref{lemma:convexf repr}.

When $xy>0$, the equality in \eqref{eq:persp integral repr1} follows immediately 
from Lemma \ref{lemma:convexf repr}.

If $x=0=y$ then $\persp{f}(0,0)=0$, and the RHS of \eqref{eq:persp integral repr1} is easily seen to be equal to $0$ as well, so for the rest we may assume that $x+y>0$.

Assume that $x>0$ and $y=0$. Then $(x-ty)_-=0$ for every $t\in(0,+\infty)$, whence the RHS of 
\eqref{eq:persp integral repr1} becomes
\begin{align*}
xf'(a^-)+x\int_{[a,+\infty)}\,df'(t)
=
xf'(a^-)+x\underbrace{\lim_{b\to+\infty}f'(b)}_{=\tilde f(0^+)}-xf'(a^-)
=
x\tilde f(0^+)=\persp{f}(x,0),
\end{align*}
where we used Lemma \ref{lemma:convex f tilde limit}.

Assume finally that $x=0$ and $y>0$. Then $(x-ty)_+=0$ for all $t>0$, whence
the second integral in \eqref{eq:persp integral repr1} is equal to $0$.
Thus, the RHS of \eqref{eq:persp integral repr1} becomes
\begin{align*}
yf(a)-yaf'(a^-)+y\int_{(0,a)}t\,df'(t),
\end{align*}
and the integral can be evaluated using integration by parts as
\begin{align*}
\int_{(0,a)}t\,df'(t)
=
af'(a^-)-0f'(0^+)-\int_{(0,a)}f'(t^-)\,dt
&=
af'(a^-)-\int_{(0,a)}f'(t)\,dt
=
af'(a^-)-f(a)+f(0^+).
\end{align*}
Thus, the RHS of \eqref{eq:persp integral repr1} is equal to 
$yf(0^+)=\persp{f}(0,y)$,
proving \eqref{eq:persp integral repr1}.
\end{proof}

\subsection{Classical $f$-divergences}

The following has been shown in \cite{Csiszar-fdiv}:

\begin{lemma}\label{lemma:cl fdiv def}
Let $(\Omega,\F)$ be a measurable space and $P,Q$ be finite positive measures on it. 
Let $f\in\ccf$. 
For any $\sigma$-finite positive measure $\mu$ on $(\Omega,\F)$ dominating $P$ and $Q$, the integral 
\begin{align}\label{eq:cl fdiv def}
\fdiv{f}{}(P\|Q):=\int_{\Omega}\persp{f}\bz\frac{dP}{d\mu}(\omega),\frac{dQ}{d\mu}(\omega)\jz\,d\mu(\omega)
\end{align}
is well-defined, and its value is independent of the dominating measure $\mu$.
\end{lemma}

\begin{proof}
(Sketch). We may assume without loss of generality that $f$ is convex. 
The statement is trivial for linear functions, and by subtracting a linear function from it, 
we may assume that $f$ is non-negative, so $\persp{f}\bz\frac{dP}{d\mu}(\omega),\frac{dQ}{d\mu}(\omega)\jz\ge 0$
for every $\omega$, whence the integral exists (measurability is obvious). Let 
$\nu:=P+Q$, so that $P,Q\ll\nu\ll\mu$. Using the positive homogeneity of the perspective function, we get 
\begin{align*}
\int_{\Omega}\persp{f}\bz\frac{dP}{d\mu}(\omega),\frac{dQ}{d\mu}(\omega)\jz\,d\mu(\omega)
&=
\int_{\Omega}\persp{f}\bz\frac{dP}{d\nu}(\omega)\frac{d\nu}{d\mu}(\omega),\frac{dQ}{d\nu}(\omega)\frac{d\nu}{d\mu}(\omega)\jz\,d\mu(\omega)\\
&=
\int_{\Omega}\persp{f}\bz\frac{dP}{d\nu}(\omega),\frac{dQ}{d\nu}(\omega)\jz\frac{d\nu}{d\mu}(\omega)\,d\mu(\omega)\\
&=
\int_{\Omega}\persp{f}\bz\frac{dP}{d\nu}(\omega),\frac{dQ}{d\nu}(\omega)\jz\,d\nu(\omega),
\end{align*}
where the second equality follows by the positive homogeneity of the perspective function.
This shows the independence of \eqref{eq:cl fdiv def} of the dominating measure $\mu$.
\end{proof}

\begin{defin}
In the setting of Lemma \ref{lemma:cl fdiv def}, $\fdiv{f}{}(P\|Q)$ is called the 
\ki{(classical) $f$-divergence} of $P$ and $Q$.
\end{defin}

\begin{lemma}\label{lemma:cl fdiv def2}
In the setting of Lemma \ref{lemma:cl fdiv def}, let 
\begin{align*}
\Omega_{\mu}^{++}:=\left\{\omega\in\Omega:\,\frac{dP}{d\mu}(\omega)\frac{dQ}{d\mu}(\omega)>0\right\}.
\end{align*}
Then $P^{++}:=P\vert_{\Omega_{\mu}^{++}}$ and
$Q^{++}:=Q\vert_{\Omega_{\mu}^{++}}$ are mutually absolutely continuous to each other, and 
\begin{align}\label{eq:cl fdiv def2}
\fdiv{f}{}(P\|Q)=
f(0^+)Q\bz\left\{\frac{dP}{d\mu}=0\right\}\jz+
\tilde f(0^+)P\bz\left\{\frac{dQ}{d\mu}=0\right\}\jz+
\fdiv{f}{++}(P\|Q),
\end{align}
where 
\begin{align}\label{eq:fdiv++ def}
\fdiv{f}{++}(P\|Q)&:=
\int_{\Omega_{\mu}^{++}}f\bz\frac{dP/d\mu}{dQ/d\mu}\jz\,\frac{dQ}{d\mu}\,d\mu
=\int_{\Omega_{\mu}^{++}}f\bz\frac{dP^{++}}{dQ^{++}}\jz\,dQ^{++}
=\fdiv{f}{}(P^{++}\|Q^{++}).
\end{align}
If, moreover, $P\ll Q$ and $f(0):=f(0^+)\in\bR$ then 
\begin{align}
\fdiv{f}{}(P\|Q)=\int_{\Omega}f\bz\frac{dP}{dQ}\jz\,dQ.
\end{align}
\end{lemma}
\begin{proof}
By \eqref{eq:cl fdiv def} and \eqref{persp extension},
\begin{align*}
\fdiv{f}{}(P\|Q)
=&
\int_{\Omega}\persp{f}\bz\frac{dP}{d\mu}(\omega),\frac{dQ}{d\mu}(\omega)\jz\,d\mu(\omega)\\
=&
\underbrace{\int_{\left\{\frac{dP}{d\mu}=0\right\}}f(0^+)\frac{dQ}{d\mu}(\omega)\,d\mu(\omega)}_{
=f(0^+)Q\bz\left\{\frac{dP}{d\mu}=0\right\}\jz}
+\int_{\left\{\frac{dQ}{d\mu}=0\right\}\setminus\left\{\frac{dP}{d\mu}=0\right\}}\tilde f(0^+)\frac{dP}{d\mu}(\omega)\,d\mu(\omega)\\
&+\int_{\Omega_{\mu}^{++}}f\bz\frac{dP/d\mu(\omega)}{dQ/d\mu(\omega)}\jz\,\frac{dQ}{d\mu}(\omega)\,d\mu(\omega).
\end{align*}
The second term above is equal to 
\begin{align*}
\tilde f(0^+)P\bz\left\{\frac{dQ}{d\mu}=0\right\}\setminus\left\{\frac{dP}{d\mu}=0\right\}\jz
=
\tilde f(0^+)P\bz\left\{\frac{dQ}{d\mu}=0\right\}\jz,
\end{align*}
according to Lemma \ref{lemma:zero RN}.
This proves \eqref{eq:cl fdiv def2}.

Note that for any $A\in\F$, $A\subseteq\Omega_{\mu}^{++}$, 
\begin{align*}
P(A)=0\ds\iff\ds \mu(A)=0\ds\iff\ds Q(A)=0.
\end{align*}
Indeed, $\mu(A)=0\imp P(A)=0$ follows from $P\ll \mu$, while 
$P(A)=0$ implies
\begin{align*}
0=P(A)\ge \int_{A\cap\{dP/d\mu\ge 1/n\}}\frac{dP}{d\mu}\,d\mu\ge (1/n)\mu\bz A\cap\left\{\frac{dP}{d\mu}\ge \frac{1}{n}\right\}\jz
\end{align*}
for every $n\in\bN$, whence
\begin{align*}
\mu(A)=
\mu\bz A\cap\left\{\frac{dP}{d\mu}>0\right\}\jz=
\lim_{n\to+\infty}\mu\bz A\cap\left\{\frac{dP}{d\mu}\ge \frac{1}{n}\right\}\jz=0.
\end{align*}
The equivalence $\mu(A)=0\iff Q(A)=0$ follows the same way. Hence,
$Q^{++}\ll\mu^{++}\ll P^{++}\ll \mu^{++}\ll Q^{++}$, where $\mu^{++}:=\mu\vert_{\Omega_{\mu}^{++}}$.
Moreover, for any $A\in\F$, $A\subseteq\Omega_{\mu}^{++}$,
\begin{align*}
P^{++}(A)=P(A)=
\int_A\frac{dP}{d\mu}\,d\mu
=
\int_A\frac{dP}{d\mu}\,d\mu^{++},
\end{align*}
whence $\frac{dP}{d\mu}\vert_{\Omega_{\mu}^{++}}$ is a realization of the Radon-Nikodym derivative
$\frac{dP^{++}}{d\mu^{++}}$, and likewise,
$\frac{dQ}{d\mu}\vert_{\Omega_{\mu}^{++}}$ is a realization of the Radon-Nikodym derivative
$\frac{dQ^{++}}{d\mu^{++}}$. Moreover, by a well-known property of the Radon-Nikodym derivative,
\begin{align*}
\bz\frac{dP}{d\mu}\Big/\frac{dQ}{d\mu}\jz\Bigg\vert_{\Omega_{\mu}^{++}}
=
\bz\frac{dP}{d\mu}\Bigg\vert_{\Omega_{\mu}^{++}}\jz
\Big/
\bz\frac{dQ}{d\mu}\Bigg\vert_{\Omega_{\mu}^{++}}\jz
=
\frac{dP^{++}}{d\mu^{++}}\Big/\frac{dQ^{++}}{d\mu^{++}}
\end{align*}
is a realization of the Radon-Nikodym derivative $dP^{++}/dQ^{++}$. Hence,
\begin{align*}
\int_{\Omega_{\mu}^{++}}f\bz\frac{dP/d\mu}{dQ/d\mu}\jz\,\frac{dQ}{d\mu}\,d\mu
=\int_{\Omega_{\mu}^{++}}f\bz\frac{dP^{++}}{dQ^{++}}\jz\frac{dQ^{++}}{d\mu^{++}}\,d\mu^{++}
=\int_{\Omega_{\mu}^{++}}f\bz\frac{dP^{++}}{dQ^{++}}\jz\,dQ^{++},
\end{align*}
proving the first equality in \eqref{eq:fdiv++ def}. The second equality in 
\eqref{eq:fdiv++ def} then follows by applying 
\eqref{eq:cl fdiv def2} to $P^{++}$ and $Q^{++}$ in place of $P$ and $Q$, respectively,
with $\mu:=Q^{++}$.

Assume now that $P\ll Q$ and $f(0):=f(0^+)\in\bR$. Then we may choose $\mu:=Q$ and 
$dQ/d\mu=dQ/dQ\equiv 1$, 
whence $\Omega_Q^{++}=\{dP/dQ>0\}$,
and
by \eqref{eq:cl fdiv def2},
\begin{align}
\fdiv{f}{}(P\|Q)
&=
f(0^+)\underbrace{Q\bz\left\{\frac{dP}{dQ}=0\right\}\jz}_{=\int_{\{dP/dQ=0\}}\,dQ}+
\tilde f(0^+)\underbrace{P\bz\left\{\frac{dQ}{dQ}=0\right\}\jz}_{=0}+
\int_{\Omega_{Q}^{++}}f\bz\frac{dP/dQ}{dQ/dQ}\jz\,\frac{dQ}{dQ}\,dQ\\
&=
\int_{\{dP/dQ=0\}}f\bz\frac{dP}{dQ}\jz\,dQ
+
\int_{\{dP/dQ>0\}}f\bz\frac{dP}{dQ}\jz\,dQ
=
\int_{\Omega}f\bz\frac{dP}{dQ}\jz\,dQ.
\end{align}
\end{proof}

\begin{lemma}
In the setting of Lemma \ref{lemma:cl fdiv def2},
\begin{align*}
&Q\bz\left\{\frac{dP}{d\mu}=0\right\}\jz=
Q\bz\left\{\frac{dP}{d(P+Q)}=0\right\}\jz,
\ds\ds\ds
P\bz\left\{\frac{dQ}{d\mu}=0\right\}\jz=
P\bz\left\{\frac{dQ}{d(P+Q)}=0\right\}\jz,\\
&\int_{\Omega_{\mu}^{++}}f\bz\frac{dP/d\mu}{dQ/d\mu}\jz\,\frac{dQ}{d\mu}\,d\mu
=
\int_{\Omega_{P+Q}^{++}}f\bz\frac{dP/d(P+Q)}{dQ/d(P+Q)}\jz\,\frac{dQ}{d(P+Q)}\,d(P+Q).
\end{align*}
In particular, each term in \eqref{eq:cl fdiv def2} is separately independent of the dominating measure $\mu$.
\end{lemma}
\begin{proof}
Follows immediately from Lemma \ref{lemma:error welldef}, with suitable choices of $F,G$ and $J$.
\end{proof}

\subsection{Classical $f$-divergences from hockey stick divergences}

\begin{ex}
For every $t\in\bR$, consider the functions
\begin{align*}
(\id-t)_+:\,x\mapsto (x-t)_+,\ds\ds\ds
(\id-t)_-:\,x\mapsto (x-t)_-\,,\ds\ds\ds x\in\bR.
\end{align*}
These are convex on $\bR$, with corresponding $f$-divergences
\begin{align}
\hsp{t}{}(P\|Q)
&=
\int_{\Omega}\bz \frac{dP}{d\mu}(\omega)-t \frac{dQ}{d\mu}(\omega)\jz_+\,d\mu(\omega)\nn\\
&=
\int_{\Omega}\bz \frac{dP}{d\mu}(\omega)-t \frac{dQ}{d\mu}(\omega)\jz\egy_{\{\frac{dP}{d\mu}-t \frac{dQ}{d\mu}>0\}}(\omega)\,d\mu(\omega)\nn\\
&=
P\bz\left\{\frac{dP}{d\mu}(\omega)-t \frac{dQ}{d\mu}(\omega)>0\right\}\jz
-tQ\bz\left\{\frac{dP}{d\mu}(\omega)-t \frac{dQ}{d\mu}(\omega)>0\right\}\jz,\label{eq:clhspgen}\\
\hsn{t}{}(P\|Q)
&=
\int_{\Omega}\bz \frac{dP}{d\mu}(\omega)-t \frac{dQ}{d\mu}(\omega)\jz_-\,d\mu(\omega)\nn\\
&=
-\int_{\Omega}\bz \frac{dP}{d\mu}(\omega)-t \frac{dQ}{d\mu}(\omega)\jz\egy_{\{\frac{dP}{d\mu}-t \frac{dQ}{d\mu}\le0\}}(\omega)\,d\mu(\omega)\nn\\
&=
-P\bz\left\{\frac{dP}{d\mu}(\omega)-t \frac{dQ}{d\mu}(\omega)\le 0\right\}\jz
+tQ\bz\left\{\frac{dP}{d\mu}(\omega)-t \frac{dQ}{d\mu}(\omega)\le 0\right\}\jz.
\label{eq:clhsngen}
\end{align}
It is straightforward to verify that 
\begin{align*}
\hsp{t}{}(P\|Q)-\hsn{t}{}(P\|Q)&=D_{(\id-t)}(P\|Q)=P(\Omega)-tQ(\Omega),\\
\hsp{t}{}(P\|Q)+\hsn{t}{}(P\|Q)&=D_{|\id-t|}(P\|Q)=
\int_{\Omega}\abs{\frac{dP}{d\mu}(\omega)-t \frac{dQ}{d\mu}(\omega)}\,d\mu(\omega)=:\norm{P-tQ}_{\tv},
\end{align*}
where the \ki{total variation norm} $\norm{P-tQ}_{\tv}$ is independent of the dominating measure $\mu$.
\end{ex}

\begin{defin}
For any $P,Q$ as above, and any $t\in\bR$, 
$\hsp{t}{}(P\|Q)$ is the \ki{(classical) hockey stick divergence} of $P$ and $Q$ with parameter $t$.
\end{defin}

Lemma \ref{lemma:persp integral repr} yields the following:

\begin{cor}\label{cor:classical fdiv repr0}
Let $f\in\ccf$.
For any $a\in(0,+\infty)$ and any probability measures $P,Q$ on a measurable space $(\Omega,\F)$, 
their classical 
$f$-divergence can be decomposed in terms of hockey stick divergences as
\begin{align*}
\fdiv{f}{}(P\|Q)
=&
f(a)Q(\Omega)+f'(a^+)(P(\Omega)-aQ(\Omega))\\
&+
\int_{(0,a]}\hsn{t}{}(P\|Q)\,df'(t)+\int_{(a,+\infty)}\hsp{t}{}(P\|Q)\,df'(t).
\end{align*}
\end{cor}
\begin{proof}
Using Lemma \ref{lemma:persp integral repr},
\begin{align*}
\fdiv{f}{}(P\|Q)
&=
\int_{\Omega}\persp{f}\bz\frac{dP}{d\mu}(\omega),\frac{dQ}{d\mu}(\omega)\jz\,d\mu(\omega)\\
&=
\int_{\Omega}\left[\frac{dQ}{d\mu}(\omega)f(a)+f'(a^+)\bz\frac{dP}{d\mu}(\omega)-a\frac{dQ}{d\mu}(\omega)\jz
+\int_{(0,a]}\bz\frac{dP}{d\mu}(\omega)-t\frac{dQ}{d\mu}(\omega)\jz_-\,df'(t)\right.\\
& \hspace{7.1cm}
+\left.\int_{(a,+\infty)}\bz\frac{dP}{d\mu}(\omega)-t\frac{dQ}{d\mu}(\omega)\jz_+\,df'(t)\right]\,d\mu(\omega)\\
&=
f(a)Q(\Omega)+f'(a^+)(P(\Omega)-aQ(\Omega))+
\int_{(0,a]}\hsn{t}{}(P\|Q)\,df'(t)+\int_{(a,+\infty)}\hsp{t}{}(P\|Q)\,df'(t).
\end{align*}
Interchanging the order of integrations is justified again by assuming that $f$ is convex, subtracting a linear function if necessary to get a non-negative function, and then use the Fubini-Tonelli theorem.
\end{proof}

When $\Omega$ is finite, we will always use the counting measure as the dominating measure, and identify 
$P$ and $Q$ with their corresponding density functions $\rho$ and $\sigma$. Moreover, 
we further identify $\rho$ and $\sigma$ with commuting matrices
$\sum_{\omega\in\Omega}\rho(\omega)\pr{\omega}$ and 
$\sum_{\omega\in\Omega}\sigma(\omega)\pr{\omega}$, and use the notations $\rho$ and $\sigma$ for these matrices
as well. Vice versa, a pair of commuting PSD operators $\rho,\sigma$ on a finite-dimensional Hilbert space can be identified with a pair of non-negative functions on $[\dim\hil]$ via
$\rho\mapsto(\inner{e_i}{\rho e_i})_{i=1}^{\dim\hil}$,
$\sigma\mapsto(\inner{e_i}{\sigma e_i})_{i=1}^{\dim\hil}$,
where $(e_i)_{i=1}^{\dim\hil}$ is any ONB basis jointly diagonalizing both.
With these identifications, Corollary \ref{cor:classical fdiv repr0} can be reformulated as follows: 

\begin{cor}\label{cor:classical fdiv repr}
Let $f\in\ccf$.
For any $a\in(0,+\infty)$ and any commuting pair $\rho,\sigma\in\B(\hil)\p$, their classical 
$f$-divergence can be decomposed in terms of hockey stick divergences as
\begin{align}
\fdiv{f}{}(\rho\|\sigma)
&=
f(a)\Tr\sigma+f'(a^+)\Tr(\rho-a\sigma)+\int_{(0,a]}\Tr(\rho-t\sigma)_-\,df'(t)
+\int_{(a,+\infty)}\Tr(\rho-t\sigma)_+\,df'(t)
\label{eq:classical fdiv repr1}\\
&=
f(a)\Tr\sigma+f'(a^+)\Tr(\rho-a\sigma)+\int_{(0,a]}\hsn{t}{}(\rho\|\sigma)\,df'(t)
+\int_{(a,+\infty)}\hsp{t}{}(\rho\|\sigma)\,df'(t).
\label{eq:classical fdiv repr2}
\end{align}
\end{cor}

\begin{rem}
In \cite[Proposition 3]{Sason_Verdu_fdiv2016},
the following integral representation was given for 
$\fdiv{f}{}(\rho\|\sigma)$ with normalized $\rho$ and $\sigma$ and under the assumption 
that $f$ is twice continuously differentiable and $f(1)=0$:
\begin{align}
\fdiv{f}{}(\rho\|\sigma)
&=
\int_{1}^{+\infty}\frac{1}{t^3}f''\bz\frac{1}{t}\jz\Tr(\sigma-t\rho)_+\,dt
+\int_1^{+\infty}f''(t)\Tr(\rho-t\sigma)_+\,dt.
\end{align}
This follows immediately from \eqref{eq:classical fdiv repr1} by choosing $a:=1$ and rewriting the integral 
$\int_{0}^af''(t)\Tr(\rho-t\sigma)_-\,dt$ using the change of variables 
$u:=1/t$ (see also Section \ref{sec:symmetry} for more detailes on the latter).
The proof in \cite{Sason_Verdu_fdiv2016} is not very direct and it is considerably more complicated than the one 
based on Lemma \ref{lemma:convexf repr}.
A direct proof of \eqref{eq:classical fdiv repr1} similar to the one based on Lemma \ref{lemma:convexf repr} was given in 
\cite[Appendix A]{HircheTomamichel_integral}, using \eqref{eq:Taylor remainder2} for the remainder term of the first-order Taylor expansion instead 
of \eqref{eq:convexf repr1}. Note that $\rho^0=\sigma^0$ was explicitly assumed in 
\cite{Sason_Verdu_fdiv2016}, and implicitly in \cite{HircheTomamichel_integral}; the non-trivial part of the proof of Lemma \ref{lemma:persp integral repr} is proving 
\eqref{eq:persp integral repr1} when $xy=0$, which then allows to obtain 
\eqref{eq:classical fdiv repr1} without any restriction on the supports.
\end{rem}

\begin{rem}\label{rem:cl HT formula}
Note that since 
\begin{align*}
\hsn{t}{}(\rho\|\sigma)=\Tr(\rho-t\sigma)_-=
\Tr(\rho-t\sigma)_+-\Tr(\rho-t\sigma),
\end{align*}
\eqref{eq:classical fdiv repr1}--\eqref{eq:classical fdiv repr2} can also be written as 
\begin{align}
\fdiv{f}{}(\rho\|\sigma)
=&
f(a)\Tr\sigma+f'(a^+)\Tr(\rho-a\sigma)+\int_{(0,a]}\left[\Tr(\rho-t\sigma)_+-\Tr(\rho-t\sigma)\right]\,df'(t)\nn\\
&\hspace{5cm}+\int_{(a,+\infty)}\Tr(\rho-t\sigma)_+\,df'(t)
\nn\\
=&
f\bz\frac{\Tr\rho}{\Tr\sigma}\jz\Tr\sigma
+f'\bz\bz\frac{\Tr\rho}{\Tr\sigma}\jz^+\jz\underbrace{\Tr\bz\rho-\bz\frac{\Tr\rho}{\Tr\sigma}\jz\sigma\jz}_{=0}
\nn\\
&+\int_{(0,\Tr\rho/\Tr\sigma]}\left[\Tr(\rho-t\sigma)_+-\Tr(\rho-t\sigma)\right]\,df'(t)
+\int_{(\Tr\rho/\Tr\sigma,+\infty)}\Tr(\rho-t\sigma)_+\,df'(t)
\nn\\
=&
f\bz\frac{\Tr\rho}{\Tr\sigma}\jz\Tr\sigma+\int_{(0,+\infty)}
\hsht{t}{}(\rho\|\sigma)\,df'(t),
\label{eq:et decomposition}
\end{align}
where we chose $a:=\Tr\rho/\Tr\sigma$, and 
\begin{align*}
\hsht{t}{}(\rho\|\sigma):=\hsp{t}{}(\rho\|\sigma)-\bz\Tr(\rho-t\sigma)\jz_+
=
\begin{cases}
\hsp{t}{}(\rho\|\sigma),&t\ge\frac{\Tr\rho}{\Tr\sigma},\\
\hsp{t}{}(\rho\|\sigma)-\Tr(\rho-t\sigma)=\hsn{t}{}(\rho\|\sigma),&t\le\frac{\Tr\rho}{\Tr\sigma}.
\end{cases}
\end{align*}
This expression appeared in \cite[Lemma 2.5]{HircheTomamichel_integral} in the more general quantum setting, albeit restricted to the case where $\Tr\rho=1=\Tr\sigma$ and $f$ is twice differentiable.
Note, however, that $\hsht{t}{}$ is not an $f$-divergence, i.e., there exists no $f\in\bR^{(0,+\infty)}$
such that 
$\hsht{t}{}(\rho\|\sigma)=\fdiv{f}{}(\rho\|\sigma)$ would hold for every commuting 
$\rho,\sigma$, and hence \eqref{eq:et decomposition} does not give a decomposition 
of an $f$-divergence in terms of a family of primitive $f$-divergences, unlike 
\eqref{eq:classical fdiv repr1}--\eqref{eq:classical fdiv repr2}.
\end{rem}

\subsection{Classical $f$-divergences from Neyman-Pearson error probabilities}
\label{sec:clfdiv from NP}

Lemma \ref{lemma:error welldef} yields immediately the following:

\begin{cor}\label{cor:classical NP}
In the setting of Lemma \ref{lemma:error welldef},
\begin{align*}
\erri{t}{}{P}{Q}&:=P\bz\left\{\frac{dP}{d\mu}-t\frac{dQ}{d\mu}\le 0\right\}\jz
=P\bz\left\{\frac{dP}{d(P+Q)}-t\frac{dQ}{d(P+Q)}\le 0\right\}\jz,\\
\errii{t}{}{P}{Q}&:=Q\bz\left\{\frac{dP}{d\mu}-t\frac{dQ}{d\mu}> 0\right\}\jz=
Q\bz\left\{\frac{dP}{d(P+Q)}-t\frac{dQ}{d(P+Q)}> 0\right\}\jz,
\end{align*}
are well-defined in the sense that they are independent of the choice of the dominating measure $\mu$.
\end{cor}

\begin{defin}
In the setting of Lemma \ref{lemma:error welldef}, the set 
\begin{align*}
\np{t}{}{P}{Q}:=\left\{\frac{dP}{d(P+Q)}-t\frac{dQ}{d(P+Q)}> 0\right\}
\end{align*}
is called the \ki{Neyman-Pearson test} with parameter $t$.
When $P,Q$ are probability meaures, 
$\erri{t}{}{P}{Q}$ and 
$\errii{t}{}{P}{Q}$ above are the 
\ki{type I and type II error probabilities}, respectively, of 
discriminating between $P$ and $Q$ using the
Neyman-Pearson test $\np{t}{}{P}{Q}$. 
\end{defin}

\begin{rem}
According to Corollary \ref{cor:classical NP}, the error probabilities do not change if 
instead of $\np{t}{}{P}{Q}$, the test 
\begin{align*}
\left\{\frac{dP}{d\mu}-t\frac{dQ}{d\mu}> 0\right\}
\end{align*}
is used, where $\mu$ is any positive measure dominating both $P$ and $Q$.
\end{rem}

\begin{rem}
According to the Hahn decomposition theorem, for every $t\in(0,+\infty)$, there exist 
$\hahnm{t}{\Omega},\hahnp{t}{\Omega}\in\F$ such that 
$\hahnm{t}{\Omega}\cap \hahnp{t}{\Omega}=\emptyset$, 
$\hahnm{t}{\Omega}\cup \hahnp{t}{\Omega}=\Omega$, and 
the restriction of $P-tQ$ onto $\hahnm{t}{\Omega}$ is negative, while 
the restriction of $P-tQ$ onto $\hahnp{t}{\Omega}$ is positive.
Moreover, if $\hahnmt{t}{\Omega}$ 
and $\hahnpt{t}{\Omega}$ have the same properties then 
$|P-tQ|(\hahnm{t}{\Omega}\symdiff \hahnmt{t}{\Omega})=0=|P-tQ|(\hahnp{t}{\Omega}\symdiff\hahnpt{t}{\Omega})$.

It might be tempting to think that for any such $\hahnm{t}{\Omega},\hahnp{t}{\Omega}$, the  
Neyman-Pearson error probabilities can be expressed as
$\erri{t}{}{P}{Q}=P(\hahnm{t}{\Omega})$, 
$\errii{t}{}{P}{Q}=Q(\hahnp{t}{\Omega})$. This, however, is not true, and in fact the value of 
$P(\hahnm{t}{\Omega})$ (resp.~$Q(\hahnp{t}{\Omega})$)
might depend on the specific Hahn decomposition. As a simple example, consider 
$\Omega=[d]$ for some $d\ge 3$, and 
\begin{align*}
P:=\sum_{k=1}^{d-1}\frac{1}{d+1}\egy_{\{k\}}+\frac{2}{d+1}\egy_{\{d\}},\ds\ds\ds
Q:=\frac{2}{d+1}\egy_{\{1\}}+\sum_{k=2}^{d}\frac{1}{d+1}\egy_{\{k\}},
\end{align*}
so that 
\begin{align*}
P-Q=-\frac{1}{d+1}\egy_{\{1\}}+\frac{1}{d+1}\egy_{\{d\}}.
\end{align*}
Then $\hahnm{1}{\Omega}:=\{1\}$, $\hahnp{1}{\Omega}:=\{2,\ldots,d\}$ and 
$\hahnmt{t}{\Omega}:=\{1,\ldots,d-1\}$, $\hahnpt{1}{\Omega}:=\{d\}$ 
are two Hahn decompositions, with
\begin{align*}
P(\hahnm{1}{\Omega})=\frac{1}{d+1}<\frac{d-1}{d+1}=P(\hahnmt{1}{\Omega}),\ds\ds\ds
Q(\hahnp{1}{\Omega})=\frac{d-1}{d+1}>\frac{1}{d+1}=Q(\hahnpt{1}{\Omega}).
\end{align*}
\end{rem}
\bigskip

\begin{lemma}\label{lemma:FG properties}
In the setting of Lemma \ref{lemma:error welldef}, define
\begin{align*}
F(t)&:=\erri{t}{}{P}{Q}=
P\bz\left\{\frac{dP}{d\mu}-t\frac{dQ}{d\mu}\le 0\right\}\jz,\\
G(t)&:=Q(\Omega)-\errii{t}{}{P}{Q}
=Q\bz\left\{\frac{dP}{d\mu}-t\frac{dQ}{d\mu}\le 0\right\}\jz.
\end{align*}
Then $F$ and $G$ are monotone increasing right continuous non-negative functions on $\bR$ with  
\begin{align}\label{eq:FG limits}
&F(t)=0,\ds t\in(-\infty,0], & & 
G(t)=0,\ds t\in(-\infty,0), \ds\ds G(0)=Q\bz\left\{\frac{dP}{d\mu}=0\right\}\jz,\\
&F(+\infty):= \lim_{t\to+\infty}F(t)=P\bz\left\{\frac{dQ}{d\mu}>0\right\}\jz,
& & 
G(+\infty):= \lim_{t\to+\infty}G(t)=Q\bz\left\{\frac{dQ}{d\mu}>0\right\}\jz=Q(\Omega).
\end{align}
\end{lemma}
\begin{proof}
Non-negativity and monotonicity are obvious from the definitions, and right continuity follows from the monotone continuity of measures. We have 
\begin{align*}
F(0)=
P\bz\left\{\frac{dP}{d\mu}=0\right\}\jz=0
\end{align*} 
by Lemma \ref{lemma:zero RN}. Hence, $0\le F(t)\le 0$ for all $t\in(-\infty,0)$ due to non-negativity and monotonicity. For $t\in(-\infty,0)$, we have
\begin{align*}
0\le G(t)
=
Q\bz\left\{\frac{dP}{d\mu}\le t\frac{dQ}{d\mu}\right\}\jz
\le
Q\bz\left\{\frac{dQ}{d\mu}=0\right\}\jz=0,
\end{align*}
according to Lemma \ref{lemma:zero RN}.
The expression for $G(0)$ is trivial by definition.

Note that $t\mapsto\left\{\frac{dP}{d\mu}-t\frac{dQ}{d\mu}\le 0\right\}$ is monotone increasing with 
\begin{align*}
\lim_{t\to+\infty}\left\{\frac{dP}{d\mu}-t\frac{dQ}{d\mu}\le 0\right\}
=
\bigcup_{n\in\bN}\left\{\frac{dP}{d\mu}\le n\frac{dQ}{d\mu}\right\}
=
\left\{\frac{dP}{d\mu}=0\right\}\bigcup \left\{\frac{dQ}{d\mu}>0\right\}.
\end{align*}
From this, $F(+\infty)=P(\{dQ/d\mu>0\})$ follows by the monotone continuity of measures and 
Lemma \ref{lemma:zero RN}, and we also get
\begin{align*}
Q\bz\left\{\frac{dQ}{d\mu}>0\right\}\jz\le G(+\infty)=\lim_{t\to+\infty}
Q\bz\left\{\frac{dP}{d\mu}\le t\frac{dQ}{d\mu}\right\}\jz
\le Q(\Omega)=Q\bz\left\{\frac{dQ}{d\mu}>0\right\}\jz,
\end{align*} 
where the last equality is again due to Lemma \ref{lemma:zero RN}.
\end{proof}

\begin{lemma}\label{lemma:dF dG}
In the setting of Lemma \ref{lemma:FG properties}, let 
\begin{align}\label{eq:canonical RV}
X(\omega):=\left\{\begin{array}{ll}
\frac{\frac{dP}{d\mu}(\omega)}{\frac{dQ}{d\mu}(\omega)},&\frac{dP}{d\mu}(\omega)\frac{dQ}{d\mu}(\omega)>0,\\
0,&\text{otherwise},
\end{array}\right\}
=\begin{cases}
\frac{\frac{dP}{d\mu}(\omega)}{\frac{dQ}{d\mu}(\omega)},&\frac{dQ}{d\mu}(\omega)>0,\\
0,&\text{otherwise}.
\end{cases}
\end{align}
Then 
\begin{align}\label{eq:dF}
dF+\egy_{\{0\}}P\bz\left\{\frac{dQ}{d\mu}=0\right\}\jz=P\circ X\inv,
\ds\ds\ds\ds\ds
dG=Q\circ X\inv.
\end{align}
\end{lemma}
\begin{proof}
For any $-\infty<t_1<t_2<+\infty$, we have 
\begin{align*}
dF\bz(t_1,t_2]\jz
&=
F(t_2^+)-F(t_1^+)=F(t_2)-F(t_1)
=
P\Bigg(\underbrace{\left\{t_1\frac{dQ}{d\mu}<\frac{dP}{d\mu}\le t_2\frac{dQ}{d\mu}\right\}}_{=:A_{t_1,t_2}}\Bigg)
=
P\Bigg( A_{t_1,t_2}\bigcap\left\{\frac{dP}{d\mu}>0\right\}\Bigg)\\
&=
P\Bigg( A_{t_1,t_2}\bigcap\left\{\frac{dP}{d\mu}>0\right\}\bigcap\left\{\frac{dQ}{d\mu}>0\right\}\Bigg)
+
P\Bigg( \underbrace{A_{t_1,t_2}\bigcap\left\{\frac{dP}{d\mu}>0\right\}\bigcap\left\{\frac{dQ}{d\mu}=0\right\}}_{=\emptyset}\Bigg),
\end{align*}
where the first equality is by definition, the second one follows from the rigt continuity of $F$, 
the third equality is again by definition, the fourth one follows from 
Lemma \ref{lemma:zero RN}, and the fifth one is trivial. On the other hand.
\begin{align*}
\{X\in(t_1,t_2]\}=
\bz A_{t_1,t_2}\bigcap\left\{\frac{dP}{d\mu}>0\right\}\bigcap\left\{\frac{dQ}{d\mu}>0\right\}\jz
\bigcup\begin{cases}
\emptyset,&0\notin(t_1,t_2],\\
\left\{\frac{dP}{d\mu}=0\right\}\bigcup\left\{\frac{dQ}{d\mu}=0\right\},&0\in(t_1,t_2].
\end{cases}
\end{align*}
Hence,
\begin{align*}
(P\circ X\inv)((t_1,t_2])
&=
P\Bigg( A_{t_1,t_2}\bigcap\left\{\frac{dP}{d\mu}>0\right\}\bigcap\left\{\frac{dQ}{d\mu}>0\right\}\Bigg)
+\egy_{(t_1,t_2]}(0)P\bz\left\{\frac{dP}{d\mu}=0\right\}\bigcup\left\{\frac{dQ}{d\mu}=0\right\}\jz\\
&=
dF((t_1,t_2])+\underbrace{\egy_{(t_1,t_2]}(0)}_{=\delta_{\{0\}}((t_1,t_2])}P\bz\left\{\frac{dQ}{d\mu}=0\right\}\jz,
\end{align*}
where the second equality is due to Lemma \ref{lemma:zero RN}.
This proves the first equality in \eqref{eq:dF}.

The second equality in \eqref{eq:dF} follows as for any $t\in\bR$,
\begin{align*}
dG((-\infty,t])&=G(t)
=
Q\bz\left\{\frac{dP}{d\mu}\le t\frac{dQ}{d\mu}\right\}\jz
=
Q\bz\left\{\frac{dP}{d\mu}\le t\frac{dQ}{d\mu}\right\}\bigcap
\left\{\frac{dQ}{d\mu}>0\right\}\jz\\
&=
Q\bz\left\{X\le t\right\}\bigcap
\left\{\frac{dQ}{d\mu}>0\right\}\jz
=
Q\bz\left\{X\le t\right\}\jz
=
(Q\circ X\inv)((-\infty,t]).
\end{align*}
Here, the first, the second, and the last equalities follow by definition, the third and the fifth 
equalities are due to Lemma \ref{lemma:zero RN}, and the fourth equality is obvious from the definition of $X$.
\end{proof}

\begin{cor}\label{cor:dF simple}
In the setting of Lemma \ref{lemma:FG properties},
the following are equivalent:
\begin{enumerate}
\item\label{dF simple1}
With the random variable $X$ given in \eqref{eq:canonical RV},
\begin{align}\label{eq:dF simple}
dF=P\circ X\inv,\ds\ds\ds\ds
dG=Q\circ X\inv.
\end{align}
\item\label{dF simple2}
There exists a random variable $X:\,\Omega\to\bR$ such that \eqref{eq:dF simple} holds.
\item\label{dF simple3}
$P\ll Q$.
\end{enumerate}
\end{cor}
\begin{proof}
\ref{dF simple1}$\imp$\ref{dF simple2} is trivial. Assume \ref{dF simple2}. Then 
\begin{align*}
P(\Omega)=(P\circ X\inv)(\bR)=dF(\bR)=F(+\infty)-F(-\infty)=P\bz\left\{\frac{dQ}{d\mu}>0\right\}\jz,
\end{align*}
where the last equality follows from \eqref{eq:FG limits}.
Thus, for any $A\in\F$, we have 
\begin{align}\label{eq:dF proof1}
P(A)=P\bz A\cap\left\{\frac{dQ}{d\mu}>0\right\}\jz=\int_{A\cap\left\{\frac{dQ}{d\mu}>0\right\}}\frac{dP}{d\mu}\,d\mu.
\end{align}
In particular, if $0=Q(A)=\int_A(dQ/d\mu)\,d\mu$ then $\mu(A\cap\{dQ/d\mu>0\})=0$, and hence
$P(A)=0$ follows from \eqref{eq:dF proof1}. This proves \ref{dF simple2}$\imp$\ref{dF simple3}.
Finally, assume \ref{dF simple3}, i.e., that $P\ll Q$. 
Since $Q\bz\left\{\frac{dQ}{d\mu}=0\right\}\jz=0$ by Lemma \ref{lemma:zero RN}, 
we also have $P\bz\left\{\frac{dQ}{d\mu}=0\right\}\jz=0$, which in turn gives 
\ref{dF simple3} due to Lemma \ref{lemma:dF dG}.
\end{proof}

\begin{lemma}\label{lemma:dF/dG}
In the setting of Lemma \ref{lemma:FG properties},
\begin{align}\label{eq:dF/dG1}
dF=\id\egy_{[0,+\infty)}\cdot dG,
\end{align}
which can be equivalently stated as
\begin{align}\label{eq:dF/dG2}
dF\ll dG\ds\ds\ds\text{and}\ds\ds\ds\frac{dF}{dG}(x)=x\egy_{[0,+\infty)}(x),\ds\ds x\in\bR,
\end{align}
or that for any positive measure $\eta$ on $\B(\bR)$
such that $dF,dG\ll \eta$,
\begin{align}\label{eq:dF/dG3}
\eta\bz\left\{x\in\bR:\,\frac{dF}{d\eta}(x)\ne x\frac{dG}{d\eta}(x)\right\}\jz=0.
\end{align}
\end{lemma}
\begin{proof}
For any $A\in\B(\bR)$, we have 
\begin{align*}
dF(A)=\underbrace{dF(A\cap(-\infty,0))}_{=0}+
\underbrace{dF(A\cap\{0\})}_{=dF(\{0\})=0}+dF(A\cap(0,+\infty)),
\end{align*}
and
\begin{align*}
(\id\egy_{[0,+\infty)}\cdot dG)(A)=
\underbrace{\int_{A\cap\{0\}}\id\,dG}_{=0\cdot dG(\{0\})=0}
+
\int_{A\cap(0,\infty)}\id\,dG,
\end{align*}
where we used that $dF((-\infty,0])=0=dG((-\infty,0))$, according to Lemma \ref{lemma:FG properties},
which is equivalent to 
$\frac{dF}{d\eta}(x)=0= x\frac{dG}{d\eta}(x)$ for $\eta$-a.e. $x\in(-\infty,0]$.
On the other hand, if $A\subseteq(0,+\infty)$ then 
$dF(A)=(P\circ X\inv)(A)$, $dG(A)=(Q\circ X\inv)(A)$, 
according to Lemma \ref{lemma:dF dG}, and we have
\begin{align*}
\int_A\frac{d(P\circ X\inv)}{d\eta}(x)\,d\eta(x)&=
(P\circ X\inv)(A)=P\bz X\inv(A)\jz\\
&=
\int_{X\inv(A)}\frac{dP}{d\mu}(\omega)\,d\mu(\omega)
=
\int_{X\inv(A)}X(\omega)\frac{dQ}{d\mu}(\omega)\,d\mu(\omega)
\\
&=
\int_{\Omega}\underbrace{\egy_{X\inv(A)}}_{=(\egy_A\circ X)(\omega)}X(\omega)\frac{dQ}{d\mu}(\omega)\,d\mu(\omega)
=
\int_{\Omega}((\id\egy_A)\circ X)(\omega)\,dQ(\omega)\\
&=
\int_{\bR}(\id\egy_A)(x)\,d(Q\circ X\inv)(x)
=
\int_A x\,d(Q\circ X\inv)(x)\\
&=
\int_{A}x\frac{d(Q\circ X\inv)}{d\eta}(x)\,d\eta(x).
\end{align*}
In the second line we used that for every $\omega\in X\inv(A)\subseteq X\inv((0,+\infty))$, 
$X(\omega)=((dP/d\mu)(\omega))/((dQ/d\mu)(\omega))$ by definition, and the rest of the steps are obvious.
This completes the proof.
\end{proof}

\begin{cor}\label{lemma:dG/dF}
In the setting of Lemma \ref{lemma:dF/dG},
\begin{align}\label{eq:dG/dF1}
dG\big\vert_{(0,+\infty)}\ll dF\big\vert_{(0,+\infty)}\ds\ds\text{with}\ds\ds
\frac{dG\big\vert_{(0,+\infty)}}{dF\big\vert_{(0,+\infty)}}(x)=\frac{1}{x},\ds\ds x\in(0,+\infty).
\end{align}
Moreover, 
\begin{align}\label{eq:dG/dF2}
Q\ll P\ds\imp\ds dG\ll dF\ds\ds\text{with}\ds\ds \frac{dG}{dF}(x)=\frac{1}{x}\egy_{(0,+\infty)}(x),\ds x\in\bR.
\end{align}
\end{cor}
\begin{proof}
For any $A\in\B(\bR)$, $A\subseteq(0,+\infty)$,
\begin{align*}
\int_A \frac{1}{x}\,dF(x)=\int_A \frac{1}{x}x\,dG(x)=\int_A\,dG(x),
\end{align*}
where the first equality follows from \eqref{eq:dF/dG1}. This proves
\eqref{eq:dG/dF1}. Alternatively, \eqref{eq:dG/dF1} follows immediately from 
\eqref{eq:dF/dG3}. If $Q\ll P$ then $dG(\{0\})=0=dF(\{0\})$, according to 
\eqref{eq:FG limits}, whence $dG((-\infty,0])=0=dF((-\infty,0])$, and thus 
\eqref{eq:dG/dF2} follows from \eqref{eq:dG/dF1}.
\end{proof}

\begin{lemma}\label{lemma:cl fdiv alt}
In the setting of Lemma \ref{lemma:FG properties}, 
we have 
\begin{align}
\fdiv{f}{}(P\|Q)&=D_f^{(0,+\infty)}+
f(0^+)Q\bz\left\{\frac{dP}{d\mu}=0\right\}\jz+
\tilde f(0^+)P\bz\left\{\frac{dQ}{d\mu}=0\right\}\jz,
\label{eq:cl fdiv alt1}\\
\fdiv{f}{}(dF\|dG)
&=
D_f^{(0,+\infty)}
+f(0^+)Q\bz\left\{\frac{dP}{d\mu}=0\right\}\jz,
\label{eq:cl fdiv alt2}\\
\fdiv{f}{}(P\circ X\inv\|Q\circ X\inv)
&=
D_f^{(0,+\infty)}
+\persp{f}\bz P\bz\left\{\frac{dQ}{d\mu}=0\right\}\jz,Q\bz\left\{\frac{dP}{d\mu}=0\right\}\jz\jz,
\label{eq:cl fdiv alt3}
\end{align}
where 
\begin{align}
D_f^{(0,+\infty)}&:=\int_{\Omega_{\mu}^{++}}f\bz\frac{dP/d\mu}{dQ/d\mu}\jz\,\frac{dQ}{d\mu}\,d\mu
\label{eq:cl fdiv alt4}\\
&=
\fdiv{f}{}\bz dF\vert_{(0,+\infty)}\|dG\vert_{(0,+\infty)}\jz
=
\int_{(0,+\infty)}f\bz x\jz\,dG(x)
\label{eq:cl fdiv alt5}\\
&=\fdiv{f}{}\bz (P\circ X\inv)\vert_{(0,+\infty)}\|(Q\circ X\inv)\vert_{(0,+\infty)}\jz
=\int_{(0,+\infty)}f\bz x\jz\,d(Q\circ X\inv)(x).
\label{eq:cl fdiv alt6}
\end{align}
\end{lemma}
\begin{proof}
The expressions in \eqref{eq:cl fdiv alt1} and \eqref{eq:cl fdiv alt4} are from Lemma \ref{lemma:cl fdiv def2}.
By Lemma \ref{lemma:dF dG},
$dF\vert_{(0,+\infty)}=(P\circ X\inv)\vert_{(0,+\infty)}$,
$dG\vert_{(0,+\infty)}=(Q\circ X\inv)\vert_{(0,+\infty)}$,
showing the equality of the first terms in 
\eqref{eq:cl fdiv alt5}--\eqref{eq:cl fdiv alt6},
as well as the equality of the second terms.
By lemma \ref{lemma:dF/dG}, $\mu:=dG$ is a dominating measure for both $dF$ and $dG$ with 
$dF/d\mu=\id$, $dG/d\mu=1$, so that $\bR_{dG}^{++}=(0,+\infty)$, 
$dF^{++}=dF\vert_{(0,+\infty)}$,
$dG^{++}=dG\vert_{(0,+\infty)}$,
$dF^{++}=\id dG^{++}$.
Thus, the equality of the two terms in \eqref{eq:cl fdiv alt5} follows by applying 
\eqref{eq:fdiv++ def} to $dF$ and $dG$ in place of $P$ and $Q$, respectively, and 
$dG$ in place of $\mu$, and we also get 
\eqref{eq:cl fdiv alt2} from \eqref{eq:cl fdiv def2} with these replacements.
Finally, \eqref{eq:cl fdiv alt3}
follows as
\begin{align*}
\fdiv{f}{}(P\circ X\inv\|Q\circ X\inv)
&=
\underbrace{\fdiv{f}{}((P\circ X\inv)\vert_{\{0\}}\|(Q\circ X\inv)\vert_{\{0\}})}_{=
\persp{f}\bz P\bz\left\{\frac{dQ}{d\mu}=0\right\}\jz,Q\bz\left\{\frac{dP}{d\mu}=0\right\}\jz\jz}
+
\underbrace{\fdiv{f}{}\bz (P\circ X\inv)\vert_{(0,+\infty)}\|(Q\circ X\inv)\vert_{(0,+\infty)}\jz}_{=D_f^{++}},
\end{align*}
where we used the additivity of $\persp{f}{}$ and that 
$(P\circ X\inv)(\{0\})=P(\{dQ/d\mu=0\})$,
$(Q\circ X\inv)(\{0\})=Q(\{dP/d\mu=0\})$.
\end{proof}

\begin{cor}\label{cor:cl fdiv equality}
In the setting of Lemma \ref{lemma:FG properties}, 
assume that one of the following conditions hold:
\begin{enumerate}
\item\label{cl fdiv equality1}
$P\ll Q$; 
\item\label{cl fdiv equality2}
$\tilde f(0^+)=0$;
\item\label{cl fdiv equality3}
$\fdiv{f}{}(dF\|dG)=\pm\infty$.
\end{enumerate}
Then 
\begin{align*}
\fdiv{f}{}(P\|Q)=\fdiv{f}{}(dF\|dG).
\end{align*}
Moreover,
\begin{align*}
P\ll Q\ds\imp\ds\fdiv{f}{}(P\|Q)=\fdiv{f}{}(dF\|dG)=\fdiv{f}{}(P\circ X\inv\|Q\circ X\inv).
\end{align*}
\end{cor}
\begin{proof}
The assertions under the conditions $P\ll Q$ or $\tilde f(0^+)=0$ follow immediately from
Lemma \ref{lemma:cl fdiv alt}.

Assume now \ref{cl fdiv equality3}.
If $f$ is convex then the $f$-divergence of any two probability measures is in $(-\infty,+\infty]$, and hence 
assumption \ref{cl fdiv equality3} means that $\fdiv{f}{}(dF\|dG)=+\infty$. Since convexity of $f$ also implies that 
$\tilde f(0^+)\in(-\infty,+\infty]$, \eqref{eq:cl fdiv alt1} yields 
\begin{align}\label{cl fdiv equality proof1}
\fdiv{f}{}(P\|Q)=\underbrace{\fdiv{f}{}(dF\|dG)}_{=+\infty}+
\underbrace{\tilde f(0^+)P\bz\left\{\frac{dQ}{d\mu}=0\right\}\jz}_{\in(-\infty,+\infty]}=+\infty,
\end{align}
whence
$\fdiv{f}{}(dF\|dG)=+\infty=\fdiv{f}{}(P\|Q)$.
The proof goes the same way when $f$ is concave.
\end{proof}

\begin{cor}\label{cor:cl Renyi from ep}
Let $f_{\alpha}(x)=x^{\alpha}$, $x\in[0,+\infty)$, $\alpha\in(0,1)\cup(1,+\infty)$. 
\begin{align*}
\text{If \ds $\alpha\in(0,1)$ or $P\ll Q$ \ds then }\ds
D_{f_{\alpha}}(P\|Q)=D_{f_{\alpha}}(dF\|dG)=\int_{(0,+\infty)}t^{\alpha}\,dG(t)=\int_{(0,+\infty)}t^{\alpha-1}\,dF(t).
\end{align*}
\end{cor}
\begin{proof}
If $\alpha\in(0,1)$ then $f(0^+)=0=\tilde f(0^+)$, and 
$D_{f_{\alpha}}(P\|Q)=D_{f_{\alpha}}(dF\|dG)=\int_{(0,+\infty)}t^{\alpha}\,dG(t)$ follows from 
\eqref{eq:cl fdiv alt1}--\eqref{eq:cl fdiv alt2} and \eqref{eq:cl fdiv alt5}.
The equality $\int_{(0,+\infty)}t^{\alpha}\,dG(t)=\int_{(0,+\infty)}t^{\alpha-1}\,dF(t)$ follows from 
\eqref{eq:dG/dF1}.
\end{proof}

\begin{cor}\label{cor:cl relentr from ep}
Let $\eta(x)=x\log x$, $x\in[0,+\infty)$. 
\begin{align*}
\text{If $P\ll Q$ \ds then }\ds
D_{\eta}(P\|Q)=D_{\eta}(dF\|dG)=\int_{(0,+\infty)}t\log\,dG(t)=\int_{(0,+\infty)}\log t\,dF(t).
\end{align*}
\end{cor}
\begin{proof}
The proof follows the same way as in Corollary \ref{cor:cl Renyi from ep}.
\end{proof}

\begin{rem}
Let $\Omega$ be finite and $P\not\ll Q$. Then 
\begin{align*}
&D_{f_{\alpha}}(P\|Q)=+\infty>D_{f_{\alpha}}(dF\|dG),\\
&D_{\eta}(P\|Q)=+\infty>D_{\eta}(dF\|dG),
\end{align*}
where the strict inequalities follow from the fact that $dF\ll dG$ independently of whether 
$P\ll Q$ or not, and that $dF$ and $dG$ are both finitely supported if $\Omega$ is finite.
\end{rem}

\subsection{Differentiability}

\begin{lemma}
In the setting considered in Section \ref{sec:clfdiv from NP}, for every $t\in(0,+\infty)$, 
\begin{align}
\exists\s\partial_t^+\hsp{t}{}(P\|Q)&=
-Q\bz\left\{\frac{dP}{d\mu}-t\frac{dQ}{d\mu}> 0\right\}\jz=-\errii{t}{}{P}{Q},
\label{eq:cl hs derivative1}\\
\exists\s\partial_t^-\hsp{t}{}(P\|Q)&=
-Q\bz\left\{\frac{dP}{d\mu}-t\frac{dQ}{d\mu}\ge 0\right\}\jz.
\label{eq:cl hs derivative2}
\end{align}
\end{lemma}
\begin{proof}
For every $\omega\in\Omega$, the function
\begin{align*}
f(\omega,t):=\bz\frac{dP}{d\mu}(\omega)-t\frac{dQ}{d\mu}(\omega)\jz_+
\end{align*}
is convex with 
\begin{align*}
\partial_2^+f(\omega,t)&=\bz-\frac{dQ}{d\mu}(\omega)\jz\egy_{(0,+\infty)}\bz \frac{dP}{d\mu}(\omega)-t\frac{dQ}{d\mu}(\omega)\jz,\\
\partial_2^-f(\omega,t)&=\bz-\frac{dQ}{d\mu}(\omega)\jz\egy_{[0,+\infty)}\bz \frac{dP}{d\mu}(\omega)-t\frac{dQ}{d\mu}(\omega)\jz,
\end{align*}
according to Corollary \ref{cor:hs derivative}.
It is easy to see that the conditions in  
Lemma \ref{lemma:integral partial derivative} are satisfied with $f$ above, and hence
\eqref{eq:cl hs derivative1}--\eqref{eq:cl hs derivative2} follow immediately from 
Lemma \ref{lemma:integral partial derivative}.
\end{proof}

\begin{rem}
Since $\bR\ni t\mapsto \hsp{t}{}(P\|Q)$ is convex, we have 
$\partial_t^+\hsp{t}{}(P\|Q)=\partial_t^-\hsp{t}{}(P\|Q)$, 
or equivalently, by \eqref{eq:cl hs derivative1}--\eqref{eq:cl hs derivative2},
\begin{align*}
Q\bz\left\{\frac{dP}{d\mu}-t\frac{dQ}{d\mu}> 0\right\}\jz=
Q\bz\left\{\frac{dP}{d\mu}-t\frac{dQ}{d\mu}\ge 0\right\}\jz
\end{align*}
for all but at most countably many $t\in\bR$.
\end{rem}

\section{$f$-divergences from quantum hockey stick divergences: finite dimension}
\label{sec:finitedim}

\subsection{Quantum $f$-divergences}

\begin{defin}
Let $f\in\ccf$. 
We say that a function
\begin{align*}
D_f^q:\,\cup_{n\in\bN}\bz\B(\bC^d)\p\times\B(\bC^d)\p\jz\to[-\infty,+\infty]
\end{align*}
is a \ki{quantum $f$-divergence} if it satisfies the following:
\begin{enumerate}
\item
\ki{Isometric invariance:} For any $d,d' \in\bN$ with $d\le d'$ and any isometry $V:\,\bC^d\to\bC^{d'}$,
\begin{align*}
D_f^q(V\rho V^*\|V\sigma V^*)=D_f^q(\rho\|\sigma),\ds\ds\ds\rho,\sigma\in\B(\bC^d)\pne.
\end{align*}
\item
\ki{Classical reduction:} For any $d\in\bN$ and any $\rho,\sigma\in\B(\bC^d)\pne$ that can be diagonalized in the same orthonormal basis 
$(e_i)_{i=1}^d$ as $\rho=\sum_{i=1}^d\rho(i)\pr{e_i}$,  $\sigma=\sum_{i=1}^d\sigma(i)\pr{e_i}$, we have
\begin{align*}
D_f^q(\rho\|\sigma)=D_f\bz (\rho(i))_{i=1}^d\|(\sigma(i))_{i=1}^d\jz,
\end{align*}
where the RHS is the classical $f$-divergence of $(\rho(i))_{i=1}^d$ and $(\sigma(i))_{i=1}^d$.
\end{enumerate}
\end{defin}

Due to the isometric invariance property, $D_f^q$ can be uniquely extended to any pair $\rho,\sigma$ of PSD operators on a finite-dimensional Hilbert space $\hil$ by taking a unitary 
$U:\,\hil\to\bC^{\dim\hil}$ and defining 
\begin{align*}
D_f^q(\rho\|\sigma):=D_f^q\bz U\rho U^*\|U\sigma U^*\jz.
\end{align*}
It is easy to see that the above definition does not depend on the choice of the unitary, and the resulting quantity also satisfies 
(the obvious modification of) the isometric invariance and the classical reduction properties. The most studied quantum $f$-divergences are the Petz-type,
the measured, and the maximal $f$-divergences, which we discuss very briefly below.  

The following is easy to verify:
\begin{lemma}\label{lemma:unique ext}
For any quantum $f$-divergence $D_f^q$ and any $\rho,\sigma\in\B(\hil)\p$ that are diagonal in the same ONB $(e_i)_{i=1}^d$, we have 
\begin{align*}
D_f^q(\rho\|\sigma)=
D_f\bz(\inner{e_i}{\rho e_i})_{i=1}^d\|(\inner{e_i}{\sigma e_i})_{i=1}^d\jz,
\end{align*}
where on the RHS we have the classical $f$-divergence.
\end{lemma}

\begin{rem}
For any $f\in\ccf$, we will denote the corresponding classical $f$-divergence by $D_f$ in the setting of Section \ref{sec:classical}. Moreoover, 
we will also use the notation $D_f(\rho\|\sigma)=D_f^q(\rho\|\sigma)$ for any quantum $f$-divergence and commuting $\rho,\sigma\in\B(\hil)\p$, justified by the above Lemma \ref{lemma:unique ext}.
\end{rem}

\begin{example}\label{ex:Petz fdiv}
\ki{(Petz-type $f$-divergence \cite{P86})} For any $f\in\ccf$, any PSD operators $\rho,\sigma$ on a finite-dimensional Hilbert space $\hil$, their 
Petz-type $f$-divergence is defined as
\begin{align*}
D_f^{\petz}(\rho\|\sigma):=\lim_{\ep\searrow 0}\Tr\sigma^{1/2}f\bz\Delta_{\rho+\ep I,\sigma+\ep I}\jz(\sigma^{1/2}),
\end{align*}
where
\begin{align*}
\Delta_{\rho+\ep I,\sigma+\ep I}:\,X\mapsto (\rho+\ep I)X(\sigma+\ep I)\inv,\ds\ds\ds X\in\B(\hil),
\end{align*}
is the relative modular operator corresponding to $\rho+\ep I$ and $\sigma+\ep I$. 

For $\rho,\sigma\in\B(\hil)\p$, let 
\begin{align*}
\nsd{\rho}{\sigma}(r,s):=r\Tr P^{\rho}_rP^{\sigma}_s,\ds\ds\ds
\nsdd{\rho}{\sigma}(r,s):=s\Tr P^{\rho}_rP^{\sigma}_s,\ds\ds\ds
r\in\spec(\rho),\,s\in\spec(\sigma),
\end{align*}
be the corresponding Nussbaum-Szko\l a distributions \cite{NSz}. It is well known and easy to verify that 
for any $f\in\ccf$, the corresponding Petz-type $f$-divergence of 
$\rho$ and $\sigma$ satisfies
$\fdiv{f}{\petz}(\rho\|\sigma)=\fdiv{f}{}(\nsd{\rho}{\sigma}\|\nsdd{\rho}{\sigma})$. 
\end{example}

\begin{example}\label{ex:measured fdiv}
\ki{(Measured $f$-divergence)} For any $f\in\ccf$, any PSD operators $\rho,\sigma$ on a finite-dimensional Hilbert space $\hil$, their 
measured $f$-divergence is defined as
\begin{align*}
D_f^{\meas}(\rho\|\sigma):=\sup\left\{D_f\bz(\Tr M_i\rho)_{i=1}^d\|(\Tr M_i\sigma)_{i=1}^d\jz:\,M\in\povm(\hil,[r],\,r\in\bN)\right\}.
\end{align*}
\end{example}

\begin{ex}
\ki{(Maximal $f$-divergence \cite{Matsumoto_newfdiv})}
For $\rho,\sigma\in\B(\hil)\p$, the set of 
\ki{reverse tests} for $(\rho,\sigma)$ is  
the set of triples $(\tilde\rho,\tilde\sigma,\rt)$, where
$\tilde\rho=\sum_{i=1}^d\tilde\rho(i)\pr{\egy_{\{i\}}}$, 
$\tilde\sigma=\sum_{i=1}^d\tilde\sigma(i)\pr{\egy_{\{i\}}}$
are non-zero PSD operators that are diagonal in the canonical ONB of $\bC^d$ for some $d\in\bN$, and 
$\rt:\,\B(\bC^d)\to\B(\hil)$ is a completely positive trace-preserving map such that 
$\rt(\tilde\rho)=\rho$, 
$\rt(\tilde\sigma)=\sigma$. For any $f\in\ccf$, any non-zero PSD operators $\rho,\sigma$ on a finite-dimensional Hilbert space $\hil$, their 
maximal $f$-divergence is defined as
\begin{align*}
D_f^{\max}(\rho\|\sigma):=\inf\left\{D_f\bz(\tilde\rho(i))_{i=1}^d\|(\tilde\sigma(i))_{i=1}^d\jz:\,(\tilde\rho,\tilde\sigma,\rt)\in\rts{\rho}{\sigma}\right\}
\end{align*}
\end{ex}

\begin{defin}
We say that a quantum $f$-divergence $D_f^q$ is \ki{monotone}, or \ki{monotone under CPTP maps} if for any 
$\rho,\sigma\in\B(\hil)\p$ and any CPTP map $\map:\,\B(\hil)\to\B(\hil)$, 
\begin{align*}
D_f^q(\map(\rho)\|\map(\sigma))\le D_f(\rho\|\sigma).
\end{align*}
Moreover, $D_f^q$ is called \ki{monotone under PTP maps} if the above holds under the weaker condition that $\map$ is positive and trace-preserving.
\end{defin}

It is well known that the Petz-type $f$-divergences are monotone for any operator convex $f$, while the measured and the maximal $f$-divergences
are PTP-monotone for any convex $f$. Moreover, the measured is the smallest, and the maximal is the largest monotone quantum $f$-divergence for any convex $f$.

\subsection{Hockey stick divergences}\label{sec:IV.A}

Note that for any $t\in\bR$, and any commuting PSD operators $\rho,\sigma\in\B(\hil)\p$,
\begin{align*}
\hsp{t}{}(\rho\|\sigma)-\hsn{t}{}(\rho\|\sigma)&=\fdiv{(\id-t)}{}(\rho\|\sigma),\\
\hsp{t}{}(\rho\|\sigma)+\hsn{t}{}(\rho\|\sigma)&=\tn{t}{}(\rho\|\sigma).
\end{align*}
Note also that for any $t\in\bR$, and any $\rho,\sigma\in\B(\hil)\p$,
\begin{align*}
\fdiv{(\id-t)}{\meas}(\rho\|\sigma)=\Tr(\rho-t\sigma)=
\fdiv{(\id-t)}{\max}(\rho\|\sigma).
\end{align*}
Hence, $\fdiv{(\id-t)}{}$ has a unique monotone quantum extension, which is equal to 
$\fdiv{(\id-t)}{\meas}$, and it is in fact invariant under trace-preserving positive maps. 
For the rest, we will always consider this quantum extension of $\fdiv{(\id-t)}{}$, which is also compatible with the usual definitions for measured, Petz-type, and maximal extensions.

Assume next that quantum extensions $(\tn{t}{q})_{t\in\bR_+}$, or 
$(\hsp{t}{q})_{t\in\bR_+}$, or $(\hsn{t}{q})_{t\in\bR_+}$ are given; then quantum extensions of the other two families can be naturally defined as
\begin{align}
\hsp{t}{q}(\rho\|\sigma)&:=
\half\bz\tn{t}{q}(\rho\|\sigma)+\fdiv{(\id-t)}{}(\rho\|\sigma)\jz,\label{eq:hsp def}\\
\hsn{t}{q}(\rho\|\sigma)&:=
\half\bz\tn{t}{q}(\rho\|\sigma)-\fdiv{(\id-t)}{}(\rho\|\sigma)\jz,\ds\ds\ds\rho,\sigma\in\B(\hil),
\label{eq:hsn def}
\end{align}
in the first case, 
\begin{align}
\hsn{t}{q}(\rho\|\sigma)&:=\hsp{t}{q}(\rho\|\sigma)-\Tr(\rho-t\sigma),\label{eq:hsn def2}\\
\tn{t}{q}(\rho\|\sigma)&:=\hsn{t}{q}(\rho\|\sigma)+\hsp{t}{q}(\rho\|\sigma),\label{eq:tn def2}
\ds\ds\ds\rho,\sigma\in\B(\hil),
\end{align}
in the second case, and 
\begin{align}
\hsp{t}{q}(\rho\|\sigma)&:=\hsn{t}{q}(\rho\|\sigma)+\Tr(\rho-t\sigma),\label{eq:hsp def3}\\
\tn{t}{q}(\rho\|\sigma)&:=\hsn{t}{q}(\rho\|\sigma)+\hsp{t}{q}(\rho\|\sigma),\label{eq:tn def3}
\ds\ds\ds\rho,\sigma\in\B(\hil),
\end{align}
in the third case. Obviously, if the members of the original family of quantum divergences are  
all monotone 
under some class of positive maps then so are the members of the other two families defined according to the above.
Moreover, if any of the three families is continuous in $t$ (meaning continuity for any fixed argument $\rho$, $\sigma$),
then so are the other two defined in the above way. 

To conform most with the existing literature, we will consider $(\hsp{t}{q})_{t\in\bR}$ as the primary 
quantities, and call them \ki{quantum hockey stick divergences}.

\begin{ex}\label{eq:Petz hs}
\ki{(Petz-type hockey stick divergences)}\ds 
According to Example \ref{ex:Petz fdiv}, for any  $\rho,\sigma\in\B(\hil)\p$ and any $t\in\bR$,
\begin{align*}
\hsp{t}{\petz}(\rho\|\sigma)=\sum_{\substack{r\in\spec(\rho)\\ s\in\spec(\sigma)}}(r-ts)_+\Tr P^{\rho}_rP^{\sigma}_s,
\ds\ds\ds
\hsn{t}{\petz}(\rho\|\sigma)=\sum_{\substack{r\in\spec(\rho)\\ s\in\spec(\sigma)}}(r-ts)_-\Tr P^{\rho}_rP^{\sigma}_s\,.
\end{align*}
\end{ex}

\begin{ex}\label{ex:hs meas}
\ki{(Measured hockey stick $f$-divergences)}\ds 
It is easy to see from the general definition in Example \ref{ex:measured fdiv}
that the measured hockey stick divergences can be expressed explicitly as 
\begin{align}
\hsn{t}{\meas}(\rho\|\sigma)&=\Tr(\rho-t\sigma)_-=\max_{T\in\bT(\hil)}\Tr(t\sigma-\rho)T\,,
\label{eq:hsn meas}\\
\hsp{t}{\meas}(\rho\|\sigma)&=\Tr(\rho-t\sigma)_+=\max_{T\in\bT(\hil)}\Tr(\rho-t\sigma)T\,.
\label{eq:hsp meas}
\end{align}
The representations in \eqref{eq:hsn meas}--\eqref{eq:hsp meas} yield immediately that 
$t\mapsto \hsn{t}{\meas}$ is monotone increasing, while 
$t\mapsto \hsp{t}{\meas}$ is monotone decreasing, that both 
$\hsn{t}{\meas}$ and $\hsn{t}{\meas}$ are jointly convex in their arguments (being maxima of affine functions), and that they are both monotone non-increasing under positive trace-preserving maps.
To see this latter property, let $\rho,\sigma\in\B(\hil)\p$, and $\map\in\ptni(\hil,\kil)$. Then 
\begin{align*}
\hsn{t}{\meas}(\map(\rho)\|\map(\sigma))
&=
\max_{T\in\bT(\kil)}\Tr\map(t\sigma-\rho)T
=
\max_{T\in\bT(\kil)}\Tr (t\sigma-\rho)\map^*(T)\\
&\le
\max_{S\in\bT(\hil)}\Tr (t\sigma-\rho) T=
\hsn{t}{\meas}(\rho\|\sigma).
\end{align*}
The proof for $\hsn{t}{\meas}$ goes the same way.
\end{ex}

\begin{ex}\label{ex:maximal hsp}
\ki{(Maximal hockey stick divergences)}
According to \cite[Section 8.1]{Matsumoto_newfdiv},
\begin{align}
\tn{t}{\max}(\rho\|\sigma)&=
\inf\{\tn{t}{}(\tilde\rho\|\tilde\sigma):\,(\tilde\rho,\tilde\sigma,\rt)\in\rts{\rho}{\sigma}\}\nn\\
&=
\Tr(\rho+t\sigma)-2\max\{\Tr A:\,0\le A\le \rho,\,A\le t\sigma\}.
\label{eq:max tn}
\end{align}
It follows immediately from \eqref{eq:max tn} that 
\begin{align}
\hsn{t}{\max}(\rho\|\sigma)&=
\inf\{\hsn{t}{}(\tilde\rho\|\tilde\sigma):\,(\tilde\rho,\tilde\sigma,\rt)\in\rts{\rho}{\sigma}\}
=
t\Tr\sigma-\max\{\Tr A:\,0\le A\le \rho,\,A\le t\sigma\},
\label{eq:hsn max}\\
\hsp{t}{\max}(\rho\|\sigma)&=
\inf\{\hsp{t}{}(\tilde\rho\|\tilde\sigma):\,(\tilde\rho,\tilde\sigma,\rt)\in\rts{\rho}{\sigma}\}
=
\Tr\rho-\max\{\Tr A:\,0\le A\le \rho,\,A\le t\sigma\}.
\label{eq:hsp max}
\end{align} 
These do not have explicit expressions in general.
\end{ex}

\begin{lemma}\label{lemma:mon hsp bounds}
For any monotone quantum hockey stick divergence $\hsp{t}{q}$, we have 
\begin{align}
t\Tr\sigma(I-\rho^0)&\le \hsn{t}{\meas}(\rho\|\sigma)\le\hsn{t}{q}(\rho\|\sigma)\le \hsn{t}{\max}(\rho\|\sigma)\le t\Tr\sigma,
\label{eq:hsn bounds}\\
\Tr\rho(I-\sigma^0)&\le \hsp{t}{\meas}(\rho\|\sigma)\le\hsp{t}{q}(\rho\|\sigma)\le \hsp{t}{\max}(\rho\|\sigma)\le\Tr\rho.
\label{eq:hsp bounds}
\end{align}
Moreover,
\begin{align}
t\le e^{-D_{\max}(\sigma\|\rho)}&\ds\imp\ds \hsn{t}{q}(\rho\|\sigma)=0
\label{eq:hsn zero}\\
t\ge e^{D_{\max}(\rho\|\sigma)}&\ds\imp\ds \hsp{t}{q}(\rho\|\sigma)=0.
\label{eq:hsp zero}
\end{align}
\end{lemma}
\begin{proof}
The second and the third inequalities in \eqref{eq:hsn bounds}--\eqref{eq:hsp bounds} follow immediately from 
the fact that for any monotone classical divergence, 
the measured extension is the smallest and the maximal extension is the largest among all quantum extensions of the given divergence.
The first and the last inequalities in \eqref{eq:hsn bounds}--\eqref{eq:hsp bounds}
are straightforward from \eqref{eq:hsn meas}--\eqref{eq:hsp meas} and 
\eqref{eq:hsp max}--\eqref{eq:hsn max}, respectively.

If $t\le e^{-D_{\max}(\sigma\|\rho)}$ then $\rho-t\sigma\ge 0$ by definition, and choosing 
$A:=t\sigma$ in \eqref{eq:hsn max} yields $0\le \hsn{t}{\max}(\rho\|\sigma)\le\Tr t\sigma-\Tr t\sigma=0$.
This proves \eqref{eq:hsn zero}, and \eqref{lemma:mon hsp bounds} follows by an analogous argument.
\end{proof}

\begin{defin}\label{def:non-neg-family}
We say that a quantum hockey stick divergence $\hsp{t}{q}$ is \ki{non-negative} if 
for any $\rho,\sigma\in\B(\hil)\p$,
\begin{align*}
0\le\hsn{t}{q}(\rho\|\sigma)\ds\ds\text{and}\ds\ds
0\le\hsp{t}{q}(\rho\|\sigma).
\end{align*}
\end{defin}

\begin{rem}
By \eqref{eq:hsn def2}--\eqref{eq:hsp def3}, non-negativity of a quantum hockey stick divergence $\hsp{t}{q}$  is equivalent to 
\begin{align*}
\Tr(\rho-t\sigma)\le\hsp{t}{q}(\rho\|\sigma)\ds\ds\text{and}\ds\ds
\Tr(t\sigma-\rho)\le\hsn{t}{q}(\rho\|\sigma),\ds\ds\ds\rho,\sigma\in\B(\hil)\p.
\end{align*}
By Lemma \ref{lemma:mon hsp bounds}, any monotone quantum hockey stick divergence is non-negative.
\end{rem}

\begin{defin}\label{def:monotone-family}
We say that $(\hsp{t}{q})_{t\in\bR_+}$ is a \ki{monotone family of quantum hockey stick divergences} if 
for any $\rho,\sigma\in\B(\hil)\p$,
\begin{align*}
&\bR_+\ni t\mapsto \hsn{t}{q}(\rho\|\sigma)\ds\ds\text{is monotone increasing, and }\\
&\bR_+\ni t\mapsto\hsp{t}{q}(\rho\|\sigma)\ds\ds\text{is monotone decreasing}.
\end{align*}
\end{defin}

\begin{defin}\label{def:measurable family}
We say that $(\hsp{t}{q})_{t\in\bR_+}$ is a \ki{measurable family of quantum hockey stick divergences} if 
for any $\rho,\sigma\in\B(\hil)\p$,
\begin{align*}
&\bR_+\ni t\mapsto\hsp{t}{q}(\rho\|\sigma)\ds\ds\text{is measurable}.
\end{align*}
\end{defin}

\begin{rem}
Note that if $(\hsp{t}{q})_{t\in\bR_+}$ is a measurable family of quantum hockey stick divergences then  
for any $\rho,\sigma\in\B(\hil)\p$,
\begin{align*}
&\bR_+\ni t\mapsto\hsn{t}{q}(\rho\|\sigma)\ds\ds\text{is measurable, too}.
\end{align*} 
\end{rem}

\begin{defin}
We say that a quantum hockey stick divergence $\hsp{t}{q}$ is \ki{canonically bounded} if 
for any $\rho,\sigma\in\B(\hil)\p$,
\begin{align*}
t\Tr\sigma(I-\rho^0)\le\hsn{t}{q}(\rho\|\sigma)\le t\Tr\sigma,\ds\ds\ds
\Tr\rho(I-\sigma^0)\le\hsp{t}{q}(\rho\|\sigma)\le \Tr\rho.
\end{align*}
We say that $(\hsp{t}{q})_{t\in\bR_+}$ is a \ki{bounded family of quantum hockey stick divergences} if 
for any $\rho,\sigma\in\B(\hil)\p$ and any $0\le a<b<+\infty$,
\begin{align*}
\bz\hsn{t}{q}(\rho\|\sigma)\jz_{t\in[a,b]}\ds\ds\text{and}\ds\ds
\bz\hsp{t}{q}(\rho\|\sigma)\jz_{t\in[a,b]}\ds\ds\text{are bounded}.
\end{align*}
\end{defin}

\begin{rem}
According to Lemma \ref{lemma:mon hsp bounds}, any monotone 
quantum hockey stick divergence $\hsp{t}{q}$ is canonically bounded, and any 
family $(\hsp{t}{q}(\rho\|\sigma))_{t\in\bR_+}$ of monotone 
quantum hockey stick divergences is bounded.
\end{rem}

\subsection{Measured hockey stick divergences and the Neyman-Pearson error probablities}

\begin{defin}
For any $\rho,\sigma\in\B(\hil)\p$ and $t\in(0,+\infty)$, let 
\begin{align}
\erri{t}{}{\rho}{\sigma}
&:=
\erri{t}{\meas}{\rho}{\sigma}
:=
\Tr\rho\{\rho-t\sigma\le 0\}=\Tr\rho(I-\np{t}{}{\rho}{\sigma}),
\label{eq:meas typeI}\\
\errii{t}{}{\rho}{\sigma}
&:=
\errii{t}{\meas}{\rho}{\sigma}
:=
\Tr\sigma\{\rho-t\sigma> 0\}=\Tr\sigma\np{t}{}{\rho}{\sigma},
\label{eq:meas typeII}
\end{align}
where 
\begin{align*}
\np{t}{\meas}{\rho}{\sigma}:=\np{t}{}{\rho}{\sigma}:=\{\rho-t\sigma>0\}
\end{align*}
is the \ki{Neyman-Pearson test} with parameter $t$. 
When $\rho,\sigma$ are quantum states
(i.e., $\Tr\rho=\Tr\sigma=1$, the quantities
$\erri{t}{\meas}{\rho}{\sigma}$ and 
$\errii{t}{\meas}{\rho}{\sigma}$ above are the 
type I and type II error probabilities, respectively, of 
discriminating between $\rho$ and $\sigma$ using the Neyman-Pearson test
$\np{t}{}{\rho}{\sigma}$.
\end{defin}

\begin{rem}
We have 
\begin{align*}
\hsp{t}{\meas}(\rho\|\sigma)=\Tr(\rho-t\sigma)_+&=\Tr\rho(I-\{\rho-t\sigma\le  0\})-
t\Tr\sigma\{\rho-t\sigma>0\}\\
&=
\Tr\rho-\bz \erri{t}{}{\rho}{\sigma}+t\errii{t}{}{\rho}{\sigma}\jz,\\
\hsn{t}{\meas}(\rho\|\sigma)=\Tr(\rho-t\sigma)_-
&=
\hsp{t}{\meas}(\rho\|\sigma)-\Tr(\rho-t\sigma)\\
&=
t\Tr\sigma-\bz \erri{t}{}{\rho}{\sigma}+t\errii{t}{}{\rho}{\sigma}\jz,
\end{align*}
or equivalently,
\begin{align}\label{eq:mixederror}
\errm{t}{}{\rho}{\sigma}:=
\erri{t}{}{\rho}{\sigma}+t\errii{t}{}{\rho}{\sigma}=
\Tr\rho-\Tr(\rho-t\sigma)_+
=
t\Tr\sigma-\Tr(\rho-t\sigma)_-\,,
\end{align}
where $\errm{t}{}{\rho}{\sigma}$ is the \ki{mixed Neyman-Pearson error probability}.
In particular, 
\begin{align}\label{eq:mixederr variational}
\errm{t}{}{\rho}{\sigma}=
\Tr\rho-\max_{T\in\test{\hil}}\Tr(\rho-t\sigma)T=
\min_{T\in\test{\hil}}\left\{\Tr\rho(I-T)+t\Tr\sigma T\right\},
\end{align}
where $T$ attains the minimum if and only if $\{\rho-t\sigma>0\}\le T\le\{\rho-t\sigma\ge 0\}$.
This is referred to as the (quantum) Neyman-Pearson lemma.
\end{rem}

\begin{lemma}
For any $\rho,\sigma\in\B(\hil)\p$ and any $t\in(0,+\infty)$,
\begin{align*}
0\le\errm{t}{}{\rho}{\sigma}\le\min\left\{\Tr\rho\sigma^0,t\Tr\sigma\rho^0\right\}.
\end{align*}
\end{lemma}
\begin{proof}
Immediate from \eqref{eq:mixederror} and \eqref{eq:hsn bounds}--\eqref{eq:hsp bounds}.
\end{proof}

\begin{lemma}\label{lemma:errprob mon}
For any $\rho,\sigma\in\B(\hil)\p$, the functions
\begin{align*}
t\mapsto \hsn{t}{}(\rho\|\sigma),\ds\ds\ds
t\mapsto \errm{t}{}{\rho}{\sigma},\ds\ds\ds
t\mapsto \erri{t}{}{\rho}{\sigma}
\end{align*}
are monotone increasing, and
\begin{align*}
t\mapsto \hsp{t}{}(\rho\|\sigma),\ds\ds\ds
t\mapsto \errii{t}{}{\rho}{\sigma}
\end{align*}
are monotone decreasing on $(0,+\infty)$.
\end{lemma}
\begin{proof}
The assertions about $\hsn{t}{}(\rho\|\sigma)$, $\hsp{t}{}(\rho\|\sigma)$, and 
$\errm{t}{}{\rho}{\sigma}$ follow immediately from the variational expressions
\eqref{eq:hsn meas}--\eqref{eq:hsp meas} and \eqref{eq:mixederr variational}.
For any $t_1,t_2\in(0,+\infty)$, 
\eqref{eq:mixederr variational} yields
\begin{align*}
\erri{t_1}{}{\rho}{\sigma}+t_1\errii{t_1}{}{\rho}{\sigma}
&\le
\Tr\rho(I-\np{t_2}{}{\rho}{\sigma})+t_1\Tr\sigma\np{t_2}{}{\rho}{\sigma}
=
\erri{t_2}{}{\rho}{\sigma}+t_1\errii{t_2}{}{\rho}{\sigma},\\
\erri{t_2}{}{\rho}{\sigma}+t_2\errii{t_2}{}{\rho}{\sigma}
&\le
\Tr\rho(I-\np{t_1}{}{\rho}{\sigma})+t_2\Tr\sigma\np{t_1}{}{\rho}{\sigma}
=
\erri{t_1}{}{\rho}{\sigma}+t_2\errii{t_1}{}{\rho}{\sigma}.
\end{align*}
By rearranging, we get 
\begin{align}\label{eq:errprob mon proof1}
t_2\bz\errii{t_2}{}{\rho}{\sigma}-\errii{t_1}{}{\rho}{\sigma}\jz
\le
\erri{t_1}{}{\rho}{\sigma}-\erri{t_1}{}{\rho}{\sigma}
\le
t_1\bz\errii{t_2}{}{\rho}{\sigma}-\errii{t_1}{}{\rho}{\sigma}\jz.
\end{align}
Thus,
\begin{align*}
0\le (t_1-t_2)\bz\errii{t_2}{}{\rho}{\sigma}-\errii{t_1}{}{\rho}{\sigma}\jz,
\end{align*}
showing that $t\mapsto \errii{t}{}{\rho}{\sigma}$ is monotone decreasing, and
thus the first inequality in \eqref{eq:errprob mon proof1} gives that 
$t\mapsto \erri{t}{}{\rho}{\sigma}$ is monotone increasing.
\end{proof}

\begin{rem}
The above monotonicity properties of $\erri{t}{}{\rho}{\sigma}$ and $\errii{t}{}{\rho}{\sigma}$
are important ingredients of the information spectrum method; the above proof is from 
\cite{NH}.
\end{rem}

The following was shown in \cite[Lemma 2.2]{LiuHircheCheng2025}. 
It also follows as a special case of the differentiability properties proved in Section \ref{sec:V.B}.

\begin{lemma}\label{lemma:measured hsp derivative}
The function $(0,+\infty)\ni t\mapsto \Tr(\rho-t\sigma)_+$ is differentiable at every 
$t\in(0,+\infty)$, except possibly at the points where $\rho-t\sigma$ is singular, and at each $t\in(0,+\infty)$, 
\begin{align*}
\partial_-\Tr(\rho-t\sigma)_+=-\Tr\sigma\{\rho-t\sigma\ge 0\},\ds\ds\ds
\partial_+\Tr(\rho-t\sigma)_+=-\Tr\sigma\{\rho-t\sigma> 0\}.
\end{align*}
\end{lemma}

\begin{cor}\label{cor:measured hsn derivative}
\begin{align*}
\partial_-\Tr(\rho-t\sigma)_-=\Tr\sigma\{\rho-t\sigma< 0\},\ds\ds\ds
\partial_+\Tr(\rho-t\sigma)_-=\Tr\sigma\{\rho-t\sigma\le 0\}.
\end{align*}
\end{cor}
\begin{proof}
Immediate from Lemma \ref{lemma:measured hsp derivative} and 
$\Tr(\rho-t\sigma)_-=\Tr(\rho-t\sigma)_+-\Tr(\rho-t\sigma)$.
\end{proof}

Motivated by Lemma \ref{lemma:measured hsp derivative},
one may consider various quantum extensions of 
the type II and type I error probabilities, respectively, as
\begin{align}\label{eq:q hsp derivative}
\errii{t}{q}{\rho}{\sigma}:=-\frac{d}{dt}\hsp{t}{q}(\rho\|\sigma),\ds\ds\ds\ds\ds
\erri{t}{q}{\rho}{\sigma}:=
\Tr\rho-\hsp{t}{q}(\rho\|\sigma)-t \errii{t}{q}{\rho}{\sigma},
\end{align}
for every $t\in(0,+\infty)$ such that the derivative in \eqref{eq:q hsp derivative} exists.

\begin{ex}
Note that for commuting states $\rho,\sigma$, 
$\hsp{t}{q}(\rho\|\sigma)=\Tr(\rho-t\sigma)_+=\hsp{t}{\meas}(\rho\|\sigma)$
for any quantum hockey stick divergence $\hsp{t}{q}$. 
By Example \ref{eq:Petz hs} and Lemma \ref{lemma:measured hsp derivative},
\begin{align*}
\errii{t}{\petz}{\rho}{\sigma}:=
-\frac{d}{dt}\hsp{t}{\petz}(\rho\|\sigma)=
-\frac{d}{dt}\hsp{t}{}(\nsd{\rho}{\sigma}\|\nsdd{\rho}{\sigma})=
\Tr\nsdd{\rho}{\sigma}\{\nsd{\rho}{\sigma}-t\nsdd{\rho}{\sigma}>0\}
=\errii{t}{}{\nsd{\rho}{\sigma}}{\nsdd{\rho}{\sigma}},
\end{align*}
and
\begin{align*}
\erri{t}{\petz}{\rho}{\sigma}&:=
\Tr\rho-
\hsp{t}{\petz}(\rho\|\sigma)-
t\errii{t}{\petz}{\rho}{\sigma}\\
&=\Tr\nsd{\rho}{\sigma}-
\hsp{t}{}(\nsd{\rho}{\sigma}\|\nsdd{\rho}{\sigma})-
t\errii{t}{}{\nsd{\rho}{\sigma}}{\nsdd{\rho}{\sigma}}
=\erri{t}{}{\nsd{\rho}{\sigma}}{\nsdd{\rho}{\sigma}}.
\end{align*}
\end{ex}

\subsection{Neyman-Pearson tests and error probabilities at $t\to+\infty$}
\label{sec:NP at infty}

\begin{lemma}\label{lemma:pospr limit}
For any $\rho,\sigma\in\B(\hil)\p$,
\begin{align}
\{\rho-(+\infty)\sigma>0\}&:=\lim_{t\to+\infty}\{\rho-t\sigma>0\}=\{(I-\sigma^0)\rho(I-\sigma^0)>0\},
\label{eq:pospr limit}\\
\{\rho-(+\infty)\sigma\ge 0\}&:=\lim_{t\to+\infty}\{\rho-t\sigma\ge 0\}=(I-\sigma^0).
\label{eq:pospr limit2}
\end{align}
\end{lemma}
\begin{proof}
Both statements are obvious when $\rho^0\le\sigma^0$, and hence for the rest we assume that
$\rho^0\nleq\sigma^0$.
Let us consider the block forms of $\rho$ and $\sigma$ according to the decomposition 
$\hil=\ran\sigma\oplus\ker\sigma$, i.e., 
\begin{align}\label{eq:pospr proof1}
\sigma=\begin{bmatrix}\sigma & 0 \\ 0 & 0\end{bmatrix},\ds\ds\ds
\rho=\begin{bmatrix}\rho_{11} & \rho_{12} \\ \rho_{21} & \rho_{22}\end{bmatrix}.
\end{align}
Note that the decreasingly ordered eigenvalues 
$\lambda^{\downarrow}_1(t)\ge \ldots\ge \lambda^{\downarrow}_d(t)$
of $\rho-t\sigma$ are monotone decreasing in $t$, which can be easily seen from their 
Courant-Weyl-Fischer minimax representation. (Here, $d:=\dim\hil$.)
Thus, there exist $1\le r\le m\le d$ and a $t_0\in(0,+\infty)$ such that for all $t\ge t_0$, 
\begin{align*}
\lambda^{\downarrow}_1(t)\ge \ldots\ge \lambda^{\downarrow}_r(t)>0=
\lambda^{\downarrow}_{r+1}(t)= \ldots= \lambda^{\downarrow}_m(t)=0>
\lambda^{\downarrow}_{m+1}(t)\ge \ldots\ge \lambda^{\downarrow}_d(t).
\end{align*}
We can therefore take a contour $\gamma$ in the complex plane enclosing exactly the 
strictly positive eigenvalues of $\rho-t\sigma$ for every $t\ge t_0$, 
as well as the strictly positive eigenvalues of $\rho_{22}$,
so that 
\begin{align*}
\{\rho-t\sigma>0\}=
\frac{1}{2\pi i}\oint_{\gamma}\bz zI-(\rho-t\sigma)\jz\inv\,dz,\ds\ds\ds t\ge t_0.
\end{align*} 
Using now the block decomposition \eqref{eq:pospr proof1}
and a known formula for the inverse of $2\times 2$ block operators (see, e.g., 
\cite[Proposition 2.8.7]{Bernstein_matrix}),
 we have 
\begin{align*}
\bz zI-(\rho-t\sigma)\jz\inv=
\begin{bmatrix} S_{z,t}\inv 
& -S_{z,t}\inv\rho_{12}(zI_2-\rho_{22})\inv \\
-(zI_2-\rho_{22})\inv\rho_{21}S_{z,t}\inv 
& 
(zI_2-\rho_{22})\inv+(zI_2-\rho_{22})\inv\rho_{21}S_{z,t}\inv\rho_{12}(zI_2-\rho_{22})\inv \end{bmatrix},
\end{align*}
where 
\begin{align*}
S_{z,t}:=zI_1-\rho+t\sigma
\end{align*}
is considered as an operator on $\hil_1:=\ran\sigma$, and $zI_2-\rho_{22}$ as an operator on 
$\hil_2:=\ker(\sigma)$. It is straightforward to verify that 
\begin{align*}
\lim_{t\to+\infty}\sup_{z\in\gamma}\norm{S_{z,t}\inv}=0, 
\end{align*}
whence
\begin{align*}
\lim_{t\to+\infty}\bz zI-(\rho-t\sigma)\jz\inv=
\begin{bmatrix} 0 & 0 \\ 0 & (zI_1-\rho_{22})\inv\end{bmatrix}\,.
\end{align*}
Since $\sup_{t\ge t_0}\sup_{z\in\gamma}\norm{\bz zI-(\rho-t\sigma)\jz\inv}<+\infty$, 
the Lebesgue dominated convergence theorem can be applied to obtain
\begin{align*}
\lim_{t\to+\infty}\{\rho-t\sigma>0\}=
\begin{bmatrix} 0 & 0 \\ 0 & \frac{1}{2\pi i}\oint_{\gamma}(zI_1-\rho_{22})\inv\,dz\end{bmatrix}
=
\begin{bmatrix} 0 & 0 \\ 0 & \{\rho_{22}>0\}\end{bmatrix}.
\end{align*} 
This proves \eqref{eq:pospr limit}, and the proof of \eqref{eq:pospr limit2} follows the same way
by applying the above argument with a contour $\gamma$ that also encloses $0$.
\end{proof}

\begin{cor}\label{cor:IV.20}
For any $\rho,\sigma\in\B(\hil)\p$,
\begin{align}
\lim_{t\to+\infty}\Tr\rho\{\rho-t\sigma>0\}&=
\lim_{t\to+\infty}\Tr\rho\{\rho-t\sigma\ge 0\}=\Tr\rho(I-\sigma^0),
\label{eq:pospr limit3}\\
\lim_{t\to+\infty}\erri{t}{}{\rho}{\sigma}&=
\lim_{t\to+\infty}\Tr\rho\{\rho-t\sigma< 0\}=\Tr\rho\sigma^0,
\label{eq:pospr limit3-2}\\
\lim_{t\to+\infty}t\errii{t}{}{\rho}{\sigma}&=
\lim_{t\to+\infty}t\Tr\sigma\{\rho-t\sigma>0\}=\lim_{t\to+\infty}t\Tr\sigma\{\rho-t\sigma\ge 0\}=0.
\label{eq:pospr limit4}\\
\lim_{t\to+\infty}\errm{t}{}{\rho}{\sigma}&=\Tr\rho\sigma^0.
\label{eq:pospr limit6}\\
\lim_{t\to+\infty}\Tr(\rho-t\sigma)_+&=\Tr\rho(I-\sigma^0).
\label{eq:pospr limit5}
\end{align}
\end{cor}

\begin{proof}
The statement about the limits in \eqref{eq:pospr limit3} follows immediately from 
Lemma \ref{lemma:pospr limit} (see also the block decomposition used in its proof),
and \eqref{eq:pospr limit3-2} is a simple reformulation of \eqref{eq:pospr limit3}.
This in turn gives
\begin{align*}
\Tr\rho(I-\sigma^0)
&\le
\liminf_{t\to+\infty}\Tr(\rho-t\sigma)_+
=
\liminf_{t\to+\infty}\bz\Tr\rho\{\rho-t\sigma\ge 0\}-t\Tr\sigma\{\rho-t\sigma\ge 0\}\jz\\
&=
\Tr\rho(I-\sigma^0)-\limsup_{t\to+\infty}t\Tr\sigma\{\rho-t\sigma\ge 0\},
\end{align*}
where the first inequality is due to \eqref{eq:hsp bounds}, and the second equality follows from 
\eqref{eq:pospr limit3}.
Thus,
\begin{align*}
0\le
\liminf_{t\to+\infty}t\Tr\sigma\{\rho-t\sigma> 0\}
\le
\limsup_{t\to+\infty}t\Tr\sigma\{\rho-t\sigma> 0\}
\le
\limsup_{t\to+\infty}t\Tr\sigma\{\rho-t\sigma\ge 0\}
\le
0,
\end{align*}
proving \eqref{eq:pospr limit4}.
Combining \eqref{eq:pospr limit3} and \eqref{eq:pospr limit4} gives
\eqref{eq:pospr limit6}, and 
\eqref{eq:pospr limit5} follows from it due to \eqref{eq:mixederror}.
\end{proof}

\begin{rem}
Note that in general, 
\begin{align*}
\{(I-\sigma^0)\rho(I-\sigma^0)>0\}\ne (I-\sigma^0)\wedge\rho^0.
\end{align*}
As an example, consider $\rho=\pr{+}=\half\begin{bmatrix} 1 & 1 \\ 1 & 1 \end{bmatrix}$ and 
$\sigma=\pr{0}=\begin{bmatrix} 1 & 0 \\0 & 0 \end{bmatrix}$. Then 
\begin{align*}
\{(I-\sigma^0)\rho(I-\sigma^0)>0\}=
\begin{bmatrix} 0 & 0 \\0 & 1 \end{bmatrix}
\ne 0=(I-\sigma^0)\wedge\rho^0.
\end{align*}
\end{rem}

\subsection{Hockey stick $f$-divergences}

\begin{lemma}\label{lemma:hsd def}
Let $f\in\ccf$, let
$(\hsp{t}{q})_{t\in\bR_+}$ be a bounded and measurable family of non-negative quantum hockey stick divergences, and let 
$\rho,\sigma\in\B(\hil)\p$. Then 
\begin{align}\label{eq:hsfdiv def}
\hsd{f,a}{q}(\rho\|\sigma):=f(a)\Tr \sigma+f'(a^+)\Tr(\rho-a\sigma)+
\int_{(0,a]}\hsn{t}{q}(\rho\|\sigma)\,df'(t)+
\int_{(a,+\infty)}\hsp{t}{q}(\rho\|\sigma)\,df'(t)
\end{align} 
gives the same value for any $a\in(0,+\infty)$, which we denote by $\hsd{f}{q}(\rho\|\sigma)$.
Moreover, for any $a\in(0,+\infty)$,
\begin{align}\label{eq:hsfdiv def2}
\hsd{f}{q}(\rho\|\sigma)=
f(a)\Tr \sigma+f'(a^-)\Tr(\rho-a\sigma)+
\int_{(0,a)}\hsn{t}{q}(\rho\|\sigma)\,df'(t)+
\int_{[a,+\infty)}\hsp{t}{q}(\rho\|\sigma)\,df'(t).
\end{align} 
\end{lemma} 
\begin{proof}
We may assume without loss of generality that $f$ is convex, whence the integral terms in 
\eqref{eq:hsfdiv def} are both integrals of non-negative functions with respect to positive measures, and are therefore non-negative. 
For any $0<a<b$, 
\begin{align*}
\int_{(0,b]}\hsn{t}{q}(\rho\|\sigma)\,df'(t)=
\int_{(0,a]}\hsn{t}{q}(\rho\|\sigma)\,df'(t)+
\int_{(a,b]}\hsn{t}{q}(\rho\|\sigma)\,df'(t), 
\end{align*}
where the last term is finite due to the assumption that 
$(\hsp{t}{q})_{t\in\bR_+}$ is a bounded family. Applying a similar argument to the second integral term in 
\eqref{eq:hsfdiv def}, we get that 
$\hsd{f,a}{q}(\rho\|\sigma)=+\infty$ $\iff$ $\hsd{f,b}{q}(\rho\|\sigma)=+\infty$. 
We may therefore assume for the rest that 
$\hsd{f,a}{q}(\rho\|\sigma)$ 
(equivalently, both of the integral terms in \eqref{eq:hsfdiv def}) are finite for every $a\in(0,+\infty)$.

Then, for any $0<a<b<+\infty$, 
\begin{align}
\hsd{f,b}{q}(\rho\|\sigma)-\hsd{f,a}{q}(\rho\|\sigma)
=&
(f(b)-f(a))\Tr \sigma+f'(b^+)\Tr(\rho-b\sigma)-f'(a^+)\Tr(\rho-a\sigma)\nn\\
&+
\int_{(a,b]}\hsn{t}{q}(\rho\|\sigma)\,df'(t)-
\int_{(a,b]}\hsp{t}{q}(\rho\|\sigma)\,df'(t).
\label{eq:hsd def proof1}
\end{align}
The integral terms above give
\begin{align}
&\int_{(a,b]}\underbrace{\left[\hsn{t}{q}(\rho\|\sigma)-\hsp{t}{q}(\rho\|\sigma)\right]}_{=-\Tr(\rho-t\sigma)}\,df'(t)
=
-(\Tr\rho)df'((a,b])+(\Tr\sigma)\int_{(a,b]}t\,df'(t)\nn\\
&
=-(\Tr\rho)\left[f'(b^+)-f'(a^+)\right]
+(\Tr\sigma)\left[bf'(b^+)-af'(a^+)-\int_{(a,b]}f'(t)\,dt\right]\nn\\
&=
-f'(b^+)\Tr(\rho-b\sigma)+f'(a^+)\Tr(\rho-a\sigma)-(\Tr\sigma)(f(b)-f(a)),
\label{eq:hsd def proof2}
\end{align}
where we used integration by parts to evaluate $\int_{(a,b]}t\,df'(t)$.
Substituting \eqref{eq:hsd def proof2} back to \eqref{eq:hsd def proof1}
yields $\hsd{f,b}{q}(\rho\|\sigma)=\hsd{f,a}{q}(\rho\|\sigma)$.

The equality of \eqref{eq:hsfdiv def} and \eqref{eq:hsfdiv def2} for any $a\in(0,+\infty)$ follows by the same argument as the equality of \eqref{eq:persp integral repr1} and \eqref{eq:persp integral repr2}
in the proof of 
Lemma \ref{lemma:persp integral repr}, and hence we omit it.
\end{proof}

\begin{defin}
We will call the quantity $\hsd{f}{q}(\rho\|\sigma)$ given in 
Lemma \ref{lemma:hsd def} the 
\ki{quantum hockey stick $f$-divergence} of $\rho$ and $\sigma$ corresponding to the 
function $f$ and the 
family $(\hsp{t}{q})_{t\in\bR_+}$
of quantum hockey stick divergences.
\end{defin}

The above definition is justified by the following:

\begin{lemma}
In the setting of Lemma \ref{lemma:hsd def}, $\hsd{f}{q}$ is a quantum $f$-divergence.
\end{lemma}
\begin{proof}
Immediate from Corollary \ref{cor:classical fdiv repr}.
\end{proof}

\begin{rem}\label{rem:hsd def}
According to Lemma \ref{lemma:mon hsp bounds}, if $\hsp{t}{q}$ is monotone for every $t\in(0,+\infty)$ then 
$(\hsp{t}{q})_{t\in\bR_+}$ is a bounded and measurable family of non-negative quantum hockey stick divergences, and hence, 
the corresponding quantum hockey stick $f$-divergence is defined for every 
$f\in\ccf$ as in Lemma \ref{lemma:hsd def}.
Also by Lemma \ref{lemma:mon hsp bounds}, the intervals for the integrals may be restricted as
\begin{align}\label{eq:interval restriction}
\hsd{f}{q}(\rho\|\sigma)=&
f(a)\Tr \sigma+f'(a^+)\Tr(\rho-a\sigma)\\
&+
\int_{(e^{-D_{\max}(\sigma\|\rho)},a]}\hsn{t}{q}(\rho\|\sigma)\,df'(t)+
\int_{(a,e^{D_{\max}(\rho\|\sigma)})}\hsp{t}{q}(\rho\|\sigma)\,df'(t).
\end{align}
In particular, 
\begin{align}
\hsd{f}{q}(\rho\|\sigma)=
f(a)\Tr \sigma+f'(a^+)\Tr(\rho-a\sigma)+
\int_{(a,e^{D_{\max}(\rho\|\sigma)})}\hsp{t}{q}(\rho\|\sigma)\,df'(t),
\ds\ds\ds 0<a\le e^{-D_{\max}(\sigma\|\rho)},
\label{eq:interval restriction2}
\end{align}
and
\begin{align}
\hsd{f}{q}(\rho\|\sigma)=
f(a)\Tr \sigma+f'(a^+)\Tr(\rho-a\sigma)+
\int_{(e^{-D_{\max}(\sigma\|\rho)},a]}\hsn{t}{q}(\rho\|\sigma)\,df'(t),
\ds\ds\ds e^{D_{\max}(\rho\|\sigma)}\le a<+\infty .
\label{eq:interval restriction3}
\end{align}
If $\sigma^0\ll\rho^0$ and $f(0^+),f'(0^+)\in\bR$, then one may take the limit $a\searrow 0$ in 
\eqref{eq:interval restriction2} to obtain
\begin{align}
\hsd{f}{q}(\rho\|\sigma)=
f(0^+)\Tr \sigma+f'(0^+)\Tr\rho+
\int_{(0,e^{D_{\max}(\rho\|\sigma)})}\hsp{t}{q}(\rho\|\sigma)\,df'(t).
\label{eq:interval restriction4}
\end{align}
\end{rem}

\begin{ex}\label{ex:hs from hs}
In the setting of Lemma \ref{lemma:hsd def}, let $f=(\id-c)_+$. With the convention \eqref{eq:fder def},
\begin{align*}
f'(x)=\begin{cases}
1,&x>c,\\
0,&x<c,\\
1/2,&x=c,
\end{cases}
\end{align*}
and therefore $df'=\delta_c$, the Dirac measure concentrated at $c$. Choosing $a:=c$ in 
\eqref{eq:hsfdiv def} gives
\begin{align*}
\hsd{(\id-c)_+}{q}(\rho\|\sigma)
&=
\underbrace{f(c)}_{=0}\Tr \sigma+\underbrace{f'(c^+)}_{=1}\Tr(\rho-c\sigma)+
\underbrace{\int_{(0,c]}\hsn{t}{q}(\rho\|\sigma)\,df'(t)}_{=\hsn{c}{q}(\rho\|\sigma)}+
\underbrace{\int_{(c,+\infty)}\hsp{t}{q}(\rho\|\sigma)\,df'(t)}_{=0}\\
&=\Tr(\rho-c\sigma)+\hsn{c}{q}(\rho\|\sigma)\\
&=\hsp{c}{q}(\rho\|\sigma),
\end{align*} 
as one expects.
\end{ex}

\begin{ex}
(Petz-type hockey stick $f$-divergences)\ds 
By Example \ref{eq:Petz hs} and Corollary \ref{cor:classical fdiv repr},
for any $f\in\ccf$ , 
\begin{align*}
\fdiv{f}{\petz}(\rho\|\sigma)
&=
\fdiv{f}{}(\nsd{\rho}{\sigma}\|\nsdd{\rho}{\sigma})\\
&=
f(a)\Tr\sigma+f'(a^+)\Tr(\rho-a\sigma)+\int_{(0,a]}\fdiv{(\id-t)_-}{}(\nsd{\rho}{\sigma}\|\nsdd{\rho}{\sigma})\,df'(t)
+\int_{(a,+\infty)}\fdiv{(\id-t)_+}{}(\nsd{\rho}{\sigma}\|\nsdd{\rho}{\sigma})\,df't\\
&=
f(a)\Tr\sigma+f'(a^+)\Tr(\rho-a\sigma)+\int_{(0,a)}\fdiv{(\id-t)_-}{\petz}(\rho\|\sigma)\,df'(t)
+\int_{(a,+\infty)}\fdiv{(\id-t)_+}{\petz}(\rho\|\sigma)\,df'(t)\\
&=
\hsd{f}{\petz}(\rho\|\sigma).
\end{align*}
Thus, \eqref{eq:hsfdiv def} does not lead to a new family of quantum $f$-divergences in this case,
but it gives a new and non-trivial integral decomposition of the Petz-type $f$-divergences.
\end{ex}

\begin{ex}\label{ex:measured hsd def}
(Measured hockey stick $f$-divergences)\ds 
According to Example \ref{ex:hs meas},
$\hsp{t}{\meas}$ is monotone for every $t\in(0,+\infty)$, and hence, by Remark \ref{rem:hsd def},
for any $f\in\ccf$, the corresponding measured hockey stick $f$-divergence is
given as
\begin{align}
\hsd{f}{\meas}(\rho\|\sigma)
&=f(a)\Tr \sigma+f'(a^+)\Tr(\rho-a\sigma)+
\int_{(0,a]}\Tr(\rho-a\sigma)_-\,df'(t)+
\int_{(a,+\infty)}\Tr(\rho-a\sigma)_+\,df'(t)
\label{eq:hs meas fdiv def}\\
&=f(a)\Tr \sigma+f'(a^-)\Tr(\rho-a\sigma)+
\int_{(0,a)}\Tr(\rho-a\sigma)_-\,df'(t)+
\int_{[a,+\infty)}\Tr(\rho-a\sigma)_+\,df'(t).
\label{eq:hs meas fdiv def2}
\end{align} 
\end{ex}

\begin{rem}\label{rem:meas HT formula}
The measured hockey stick $f$-divergences were first considered in \cite{HircheTomamichel_integral} in the case $\Tr\rho=\Tr\sigma=1$, 
$a=1$, and under the conditions that $f$ is twice continuously differentiable and $f(1)=0$. 
Note that, analogously to the classical case discussed in Remark \ref{rem:cl HT formula},
one may introduce
\begin{align*}
\hsht{t}{\meas}(\rho\|\sigma):=\hsp{t}{\meas}(\rho\|\sigma)-\bz\Tr(\rho-t\sigma)\jz_+
=
\begin{cases}
\hsp{t}{\meas}(\rho\|\sigma),&t\ge\frac{\Tr\rho}{\Tr\sigma},\\
\hsp{t}{\meas}(\rho\|\sigma)-\Tr(\rho-t\sigma)=\hsn{t}{\meas}(\rho\|\sigma),&t\le\frac{\Tr\rho}{\Tr\sigma},
\end{cases}
\end{align*}
and \eqref{eq:hs meas fdiv def} with $a:=\Tr\rho/\Tr\sigma$ may be rewritten as 
\begin{align*}
\hsd{f}{\meas}(\rho\|\sigma)
&=
f\bz\frac{\Tr\rho}{\Tr\sigma}\jz\Tr\sigma+\int_{0}^{+\infty}f''(t)
\hsht{t}{}(\rho\|\sigma)\,dt,
\end{align*}
using the same argument as in Remark \ref{rem:cl HT formula}.
This expression appeared in \cite[Lemma 2.5]{HircheTomamichel_integral}, again for the case
$\Tr\rho=\Tr\sigma=1$ and $f(1)=0$. 
\end{rem}

\begin{ex}\label{ex:max hsd def}
(Maximal hockey stick $f$-divergences)\ds 
According to Example \ref{ex:maximal hsp},
$\hsp{t}{\max}$ is monotone for every $t\in(0,+\infty)$, and hence, by Remark \ref{rem:hsd def},
for any $f\in\ccf$, the corresponding maximal hockey stick $f$-divergence is
given as
\begin{align}
\hsd{f}{\max}(\rho\|\sigma)
&=f(a)\Tr \sigma+f'(a^+)\Tr(\rho-a\sigma)+
\int_{(0,a]}\hsn{t}{\max}(\rho\|\sigma)\,df'(t)+
\int_{(a,+\infty)}\hsp{t}{\max}(\rho\|\sigma)\,df'(t)
\label{eq:hs max fdiv def}\\
&=f(a)\Tr \sigma+f'(a^-)\Tr(\rho-a\sigma)+
\int_{(0,a)}\hsn{t}{\max}(\rho\|\sigma)\,df'(t)+
\int_{[a,+\infty)}\hsp{t}{\max}(\rho\|\sigma)\,df'(t).
\label{eq:hs max fdiv def2}
\end{align} 
\end{ex}

\begin{ex}\label{ex:hs Renyi}
For any $\alpha\in(0,1)\cup(1,+\infty)$, let $f_{\alpha}(t):=t^{\alpha}$, $t\in[0,+\infty)$. 
Then $f_{\alpha}\in\ccf$, and $f_{\alpha}$ is infinitely many times differentiable on $(0,+\infty)$, whence
\begin{align*}
df'(t)=f''(t)\,dt=\alpha(\alpha-1) t^{\alpha-2}\,dt.
\end{align*}
Thus, for any bounded and measurable family
$(\hsp{t}{q})_{t\in\bR_+}$ of non-negative quantum hockey stick divergences, and any 
$\rho,\sigma\in\B(\hil)\p$, Lemma \ref{lemma:hsd def} gives
\begin{align}\label{eq:hs Renyi def}
&\hsq{\alpha}{q}(\rho\|\sigma):=\hsd{f_{\alpha}}{q}(\rho\|\sigma)\nn\\
&=
\alpha a^{\alpha-1}(\Tr\rho-a\sigma)
+\alpha(\alpha-1)\int_0^a t^{\alpha-2}\hsn{t}{q}(\rho\|\sigma)\,dt
+\alpha(\alpha-1)\int_a^{+\infty} t^{\alpha-2}\hsp{t}{q}(\rho\|\sigma)\,dt.
\end{align}
\end{ex}

\begin{defin}
In the setting of Example \ref{ex:hs Renyi},
$\hsq{\alpha}{q}(\rho\|\sigma)$ is called the 
\ki{$q$-$\mathrm{hs}$ R\'enyi $\alpha$-quantity}
of $\rho$ and $\sigma$ corresponding to the given family of quantum hockey stick R\'enyi divergences, and
\begin{align*}
\hsr{\alpha}{q}(\rho\|\sigma):=\frac{1}{\alpha-1}\log \hsq{\alpha}{q}(\rho\|\sigma),
\end{align*}
is the \ki{$q$-$\mathrm{hs}$ R\'enyi $\alpha$-divergence} of $\rho$ and $\sigma$.
\end{defin}

\begin{rem}\label{rem:a=0}
The freedom of choosing $a$ in \eqref{eq:hsfdiv def} can be useful. Consider, for instance, 
$f_{\alpha}(t):=t^{\alpha}$, $t\in[0,+\infty)$  for some $\alpha>1$. Then 
$f_{\alpha}(0)=0=f_{\alpha}'(0)$, and taking the limit $a\searrow 0$ in 
\eqref{eq:hs Renyi def} yields
\begin{align*}
\hsq{\alpha}{q}(\rho\|\sigma)
=
\alpha(\alpha-1)\int_0^{+\infty}t^{\alpha-2}\hsp{t}{q}(\rho\|\sigma)\,dt,\ds\ds\ds\alpha\in(1,+\infty).
\end{align*}
In particular, we get that if each $\hsp{t}{q}$ is monotone then
\begin{align}\label{eq:hsq infty}
0\le\hsq{\alpha}{q}(\rho\|\sigma)
\begin{cases}
<+\infty,&\rho^0\le\sigma^0,\\
=+\infty,&\rho^0\nleq\sigma^0,
\end{cases}
\ds\ds\ds\ds\alpha\in(1,+\infty).
\end{align}
Here, the first case follows from the bounds $0\le \hsp{t}{q}(\rho\|\sigma)\le\Tr\rho$
and that $\hsp{t}{q}(\rho\|\sigma)=0$ for $t\ge e^{D_{\max}(\rho\|\sigma)}$, according to 
Lemma \ref{lemma:mon hsp bounds}.
The second case follows from $\Tr\rho(I-\sigma^0)\le\hsp{t}{q}(\rho\|\sigma)$ and 
$\int_0^{+\infty}t^{\alpha-2}=+\infty$.

See further expressions for $\hsq{\alpha}{q}(\rho\|\sigma)$ in Section \ref{sec:differentiable hsd}.
\end{rem}

\begin{ex}\label{ex:hs relentr}
Let $\eta(t):=t\log t$, $t\in(0,+\infty)$. Then $\eta$ is convex and infinitely many times differentiable,
whence
\begin{align*}
d\eta'(t)=\eta''(t)\,dt=\frac{1}{t}\,dt.
\end{align*}
Thus, for any bounded and measurable family
$(\hsp{t}{q})_{t\in\bR_+}$ of non-negative quantum hockey stick divergences, and any 
$\rho,\sigma\in\B(\hil)\p$, Lemma \ref{lemma:hsd def} gives
\begin{align}\label{eq:hs relentr def}
\hsd{}{q}(\rho\|\sigma)&:=\hsd{\eta}{q}(\rho\|\sigma)\nn\\
&=
(1+\log a)(\Tr\rho-a\sigma)
+\int_0^a \frac{1}{t}\hsn{t}{q}(\rho\|\sigma)\,dt
+\int_a^{+\infty} \frac{1}{t}\hsp{t}{q}(\rho\|\sigma)\,dt.
\end{align}
\end{ex}

\begin{defin}
In the setting of Example \ref{ex:hs relentr},
$\hsd{}{q}(\rho\|\sigma)$ is called the 
\ki{$q$-$\mathrm{hs}$ relative entropy}
of $\rho$ and $\sigma$ corresponding to the given family of quantum hockey stick R\'enyi divergences.
\end{defin}

\subsection{Basic properties and order}

\begin{lemma}
Let $(\hsp{t}{q})_{t\in\bR_+}$ be a family of quantum hockey stick divergences 
as in Lemma \ref{lemma:hsd def}, and let
$f\in(\bR^{(0,+\infty)})_{\conv}\cup(\bR^{(0,+\infty)})_{\conc}$.
If all $\hsp{t}{q}$, $t\in(0,+\infty)$ satisfy one of the properties below then so does 
$\hsd{f}{q}$ as well:
\begin{enumerate}
\item
block additivity;
\item
joint convexity/concavity;
\item
monotonicity under a class of positive trace-preserving maps.
\end{enumerate}
\end{lemma}
\begin{proof}
Straightforward from the definition.
\end{proof}

\begin{lemma}
Let $(\hsp{t}{q_1})_{t\in\bR_+}$ and $(\hsp{t}{q_2})_{t\in\bR_+}$ 
be two families of quantum hockey stick divergences as in Lemma \ref{lemma:hsd def}. 
\begin{align}\label{eq:hsd order from hsp order}
\text{If}\ds\ds\ds\ds\hsp{t}{q_1}\le \hsp{t}{q_2},\ds\ds t\in(0,+\infty),\ds\ds\ds\ds\text{then}\ds\ds\ds\ds
\hsd{f}{q_1}\le \hsd{f}{q_2}
\end{align}
for any 
$f\in(\bR^{(0,+\infty)})_{\conv}$.
If $f$ is concave then the inequality in the conclusion goes in the opposite direction as
$\hsd{f}{q_1}\ge \hsd{f}{q_2}$.
\end{lemma}
\begin{proof}
Straightforward from the definition.
\end{proof}

\begin{cor}\label{cor:hs fdiv order}
Let $(\hsp{t}{q})_{t\in\bR_+}$ be a family of quantum hockey stick divergences
as in Lemma \ref{lemma:hsd def},
 such that 
all $\hsp{t}{q}$ are  monotone under CPTP maps. Then 
\begin{align}\label{eq:hs order}
\fdiv{f}{\meas}\le\hsd{f}{\meas}\le\hsd{f}{q}\le \hsd{f}{\max}\le\fdiv{f}{\max}
\end{align}
for any $f\in(\bR^{(0,+\infty)})_{\conv}$.
If $f$ is concave then the inequalities in \eqref{eq:hs order} hold in the opposite direction.
\end{cor}

\begin{rem}
The first and the last inequalities in \eqref{eq:hs order} are not equalities in general.
\end{rem}

Corollary \ref{cor:hs fdiv order} yields immediately the following:

\begin{cor}
Under the conditions of Corollary \ref{cor:hs fdiv order},
\begin{align*}
\overline{D}_{\alpha}^{\meas}(\rho\|\sigma)
\le
\liminf_{n\to+\infty}\frac{1}{n}D_{\alpha}^{q,\mathrm{hs}}(\rho^{\otimes n}\|\sigma^{\otimes n})
\le
\limsup_{n\to+\infty}\frac{1}{n}D_{\alpha}^{q,\mathrm{hs}}(\rho^{\otimes n}\|\sigma^{\otimes n}),
\ds\ds\ds\ds\ds\alpha\in(0,+\infty).
\end{align*}
\end{cor}

\begin{rem}
It was shown in \cite{BHT_fdiv} that for every $\rho,\sigma\in\S(\hil)$ and every 
$\alpha\in(0,1)$, 
\begin{align*}
\hsq{\alpha}{\meas}(\rho\|\sigma)
\le
\Tr\rho^{\alpha}\sigma^{1-\alpha}=
\hsq{\alpha}{\petz}(\rho\|\sigma).
\end{align*}
This implies that not every $\hsp{t}{\petz}$ can be monotone under CPTP  maps, since then 
the inequality above would hold in the opposite direction, according to 
Corollary \ref{cor:hs fdiv order}. In fact, it can be shown by a simple argument that 
$\hsp{t}{\petz}$ is not monotone under CPTP maps for any $t\in(0,+\infty)$.
\end{rem}

\subsection{Symmetry}
\label{sec:symmetry}

\begin{lemma}\label{lemma:symmetry}
For a $t\in(0,+\infty)$, let $\hsp{t}{q}$ and $\hsp{1/t}{q}$ be quantum hockey stick divergences, with corresponding 
$\hsn{t}{q}$, $\hsn{1/t}{q}$, $\tn{t}{q}$ and $\tn{1/t}{q}$, and let $\rho,\sigma\in\B(\hil)\p$. 
If one of the following holds for $s=t$ and $s=1/t$ then so do all the other:
\begin{enumerate}
\item\label{symmetry1}
$\displaystyle{\hsp{s}{q}(\sigma\|\rho)-\Tr\sigma=s\left[\hsp{1/s}{q}(\rho\|\sigma)-\Tr\rho\right]}$;
\item\label{symmetry2}
$\displaystyle{\hsn{s}{q}(\sigma\|\rho)+\Tr\sigma=s\left[\hsn{1/s}{q}(\rho\|\sigma)+\Tr\rho\right]}$;
\item\label{symmetry3}
$\displaystyle{\hsp{s}{q}(\rho\|\sigma)=s\hsn{1/s}{q}(\sigma\|\rho)}$;
\item\label{symmetry4}
$\displaystyle{\hsn{s}{q}(\rho\|\sigma)=s\hsp{1/s}{q}(\sigma\|\rho)}$.
\end{enumerate}
\end{lemma}
\begin{proof}
Slightly lengthy but completely straightforward computations, which we leave to the reader.
\end{proof}

\begin{defin}
A family of quantum hockey stick divergences $(\hsp{t}{q})_{t\in\bR_+}$ 
is called \ki{symmetric} if any (and hence all) the properties in Lemma \ref{lemma:symmetry}
hold for every $s\in(0,+\infty)$ and every $\rho,\sigma\in\B(\hil)\p$.
\end{defin}

\begin{lemma}\label{lemma:symmetry2}
Let $f\in\ccf$. Then $\tilde f\in\ccf$ and 
for any $a\in(0,+\infty)$ and any $\rho,\sigma\in\B(\hil)\p$,
\begin{align}\label{eq:symmetry2}
f(a)\Tr\sigma+\derright{f}(a)\Tr(\rho-a\sigma)=
\tilde f\bz\frac{1}{a}\jz\Tr\rho+\derleft{\tilde f}\bz\frac{1}{a}\jz\Tr\bz\sigma-\frac{1}{a}\rho\jz.
\end{align}
\end{lemma}
\begin{proof}
Lemma \ref{lemma:persp convex} yields that $\tilde f\in\ccf$, and \eqref{eq:symmetry2} follows by a
straightforward computation.
\end{proof}

\begin{lemma}\label{lemma:dftilde}
Let $f\in\ccf$. Then 
\begin{align}\label{eq:dftilde}
t(d\tilde f')(t)=\frac{1}{t^2}f''(t)\,dt-(df')_{\sing}\bz\frac{1}{t}\jz=(df'\circ\invv)(t),
\end{align}
where $\invv(t):=1/t$, $t\in(0,+\infty)$, and 
$(df'\circ\invv)(A):=(df')(\{x\in(0,+\infty):\,1/x\in A\})$, $A\subseteq(0,+\infty)$ Borel. 
\end{lemma}
\begin{proof}
Since $\tilde f\in\ccf$, it is differentiable almost everywhere, and 
\begin{align*}
\tilde f'(t)=f\bz\frac{1}{t}\jz-\frac{1}{t}f'\bz\frac{1}{t}\jz.
\end{align*}
This is the difference of two functions that both of locally bounded variation, hence the induced 
Lebesgue-Stieltjes measure is the difference of the Lebesgue-Stieltjes measures corresponding to the two functions. The first function is differentiable almost everywhere, and thus its
induced Lebesgue-Stieltjes measure is 
\begin{align*}
\bz\frac{d}{dt}f\bz\frac{1}{t}\jz\jz\,dt=-\frac{1}{t^2}f'\bz\frac{1}{t}\jz\,dt.
\end{align*}
The absolutely continuous part of the Lebesgue-Stieltjes measure induced by the second function is 
\begin{align*}
\bz\frac{d}{dt}\frac{1}{t}f\bz\frac{1}{t}\jz\jz\,dt=
\bz-\frac{1}{t^2}f'\bz\frac{1}{t}\jz-\frac{1}{t^3}f''\bz\frac{1}{t}\jz\jz\,dt.
\end{align*}
(Recall that $f$ is twice differentiable almost everywhere, according to Alexandrov's theorem.)
The singular part of the Lebesgue-Stieltjes measure induced by the second function is 
\begin{align*}
&\sum_{t\in(0,+\infty)}\left[\lim_{s\searrow t}\frac{1}{s}f'\bz\frac{1}{s}\jz-
\lim_{s\nearrow t}\frac{1}{s}f'\bz\frac{1}{s}\jz
\right]\,\dirac{t}
=
\sum_{t\in(0,+\infty)}\frac{1}{t}
\underbrace{\left[\derleft{f}\bz\frac{1}{t}\jz-\derright{f}\bz\frac{1}{t}\jz\right]}_{
=-(df')_{\sing}(\{1/t\})}\,\dirac{t}\\
&\ds=-\invv\cdot ((df')_{\sing}\circ\invv).
\end{align*}
Putting it all together, we get that 
\begin{align}\label{eq:dftilde proof1}
d\tilde f'(t)&=\frac{1}{t^3}f''\bz\frac{1}{t}\jz\,dt+\frac{1}{t}((df')_{\sing}\circ\invv)(t).
\end{align}
On the other hand, by \eqref{eq:Alex},
\begin{align*}
df'=f''\,d\lambda+(df')_{\sing},
\end{align*}
whence, by Lemma \ref{lemma:RN composition}, 
\begin{align}\label{eq:dftilde proof2}
(df'\circ\invv)(t)=\frac{1}{t^2}f''\bz\frac{1}{t}\jz\,dt+((df')_{\sing}\circ\invv)(t).
\end{align}
Finally, \eqref{eq:dftilde proof1} and \eqref{eq:dftilde proof2} together give
\eqref{eq:dftilde}.
\end{proof}



\begin{lemma}\label{lemma:symmetry3}
In the setting of Lemma \ref{lemma:hsd def}, assume that 
$(\hsp{t}{q})_{t\in\bR_+}$ is a symmetric family of quantum hockey stick divergences.
Then for any $f\in\ccf$ and 
$a\in(0,+\infty)$, and for any $\rho,\sigma\in\B(\hil)\p$,
\begin{align}
\int_{(0,a]}\hsn{t}{q}(\rho\|\sigma)\,df'(t)
&=
\int_{[1/a,+\infty)}\hsp{t}{q}(\sigma\|\rho)\,d\tilde f'(t),
\label{eq:symmetry3-1}\\
\int_{(a,+\infty)} \hsp{t}{q}(\rho\|\sigma)\,d f'(t)
&=
\int_{(0,1/a)}\hsn{t}{q}(\sigma\|\rho)\,d\tilde f'(t).
\label{eq:symmetry3-2}
\end{align}
\end{lemma}
\begin{proof}
We have
\begin{align*}
\int_{(0,a]}\hsn{t}{q}(\rho\|\sigma)\,df'(t)
&=
\int_{(0,a]}t\hsp{1/t}{q}(\sigma\|\rho)\,df'(t)\\
&=
\int_{(0,+\infty)}\bz\egy_{[1/a,+\infty)}\circ\invv\jz(t) 
t\hsp{1/t}{q}(\sigma\|\rho)\,df'(t)\\
&=
\int_{(0,+\infty)}\egy_{[1/a,+\infty)}(t)
\frac{1}{t}\hsp{t}{q}(\sigma\|\rho)\,(df'\circ\invv)(t)\\
&=
\int_{[1/a,+\infty)}\frac{1}{t}\hsp{t}{q}(\sigma\|\rho)t\,d\tilde f'(t)\\
&=
\int_{[1/a,+\infty)}\hsp{t}{q}(\sigma\|\rho)\,d\tilde f'(t),
\end{align*}
where the first equality is due to \ref{symmetry4} of Lemma \ref{lemma:symmetry},
the second and the third equalities are straightforward, 
the fourth equality follows by Lemma \ref{lemma:dftilde},
and the fifth equality is again straightforward.
This proves \eqref{eq:symmetry3-1}, and 
\eqref{eq:symmetry3-2} follows by an exactly analogous argument.
\end{proof}

\begin{cor}\label{cor:hsd symmetric}
In the setting of Lemma \ref{lemma:symmetry3},
%
\begin{align}\label{eq:hsd symmetry}
\hsd{f}{q}(\rho\|\sigma)=\hsd{\tilde f}{q}(\sigma\|\rho).
\end{align}
In particular, for any $\alpha\in(0,+\infty)$,
\begin{align}\label{eq:hsq symmetry}
\hsq{\alpha}{q}(\rho\|\sigma)=\hsq{1-\alpha}{q}(\sigma\|\rho).
\end{align}
\end{cor}
\begin{proof}
Immediate from Lemmas \ref{lemma:hsd def}, \ref{lemma:symmetry2}, and \ref{lemma:symmetry3}.
\end{proof}

\begin{ex}\label{ex:measured hsp symmetric}
\ki{(Symmetry of the measured hockey stick $f$-divergences)}
For any $t\in(0,+\infty)$,
\begin{align*}
t\hsn{1/t}{\meas}(\sigma\|\rho)=
t\Tr(\sigma-(1/t)\rho)_-
=
\Tr(t\sigma-\rho)_-
=
\Tr(\rho-t\sigma)_+=\hsp{t}{\meas}(\rho\|\sigma).
\end{align*}
Thus, the measured hockey stick divergences $(\hsp{t}{\meas})_{t\in\bR_+}$ form a symmetric family. By Corollary \ref{cor:hsd symmetric},
\begin{align*}
\hsd{f}{\meas}(\rho\|\sigma)=\hsd{\tilde f}{\meas}(\sigma\|\rho)
\end{align*}
for any $f\in(\bR^{(0,+\infty)})_{\conv}\cup(\bR^{(0,+\infty)})_{\conc}$.
\end{ex}

\begin{ex}\label{ex:maximal hsp symmetric}
\ki{(Symmetry of the maximal hockey stick  $f$-divergences)}
For any $t\in(0,+\infty)$,
\begin{align*}
t\hsn{1/t}{\max}(\sigma\|\rho)
&=
t\left[\Tr\frac{1}{t}\rho-\max\{\Tr A:\,0\le A\le\sigma,\,A\le \frac{1}{t}\rho\}\right]\\
&=
\Tr\rho-\max\{\Tr tA:\,0\le tA\le t\sigma,\,tA\le \rho\}\\
&=
\hsp{t}{\max}(\rho\|\sigma).
\end{align*}
By Corollary \ref{cor:hsd symmetric},
\begin{align*}
\hsd{f}{\max}(\rho\|\sigma)=\hsd{\tilde f}{\max}(\sigma\|\rho)
\end{align*}
for any $f\in(\bR^{(0,+\infty)})_{\conv}\cup(\bR^{(0,+\infty)})_{\conc}$.
\end{ex}

\begin{ex}
\ki{(Symmetry of the Petz-type hockey stick  $f$-divergences)}
According to Example \ref{eq:Petz hs},
\begin{align*}
t\hsn{1/t}{\petz}(\sigma\|\rho)=
t\hsn{1/t}{}(\nsdd{\rho}{\sigma}\|\nsd{\rho}{\sigma})
=
\hsp{t}{}(\nsd{\rho}{\sigma}\|\nsdd{\rho}{\sigma})
=
\hsp{t}{\petz}(\rho\|\sigma),
\end{align*}
where we used the symmetry of the classical hockey stick divergences, which may may also be seen as a special case of
Example \ref{ex:maximal hsp symmetric} or \ref{ex:measured hsp symmetric}.
By Corollary \ref{cor:hsd symmetric},
\begin{align*}
\hsd{f}{\petz}(\rho\|\sigma)=\hsd{\tilde f}{\petz}(\sigma\|\rho)
\end{align*}
for any $f\in(\bR^{(0,+\infty)})_{\conv}\cup(\bR^{(0,+\infty)})_{\conc}$.
This, however, may be easily seen directly also from the definition of the Petz-type $f$-divergences, and hold
more generally without assuming convexity of concavity of $f$.
\end{ex}

\subsection{Hockey stick $f$-divergences and Neyman-Pearson error probabilities}
\label{sec:differentiable hsd}

\begin{prop}\label{prop:differentiable hs}
Let $f\in(\bR^{(0,+\infty)})_{\conv}$ with $f(0^+)$ finite, 
and assume that one of the following holds:
\begin{enumerate}
\item\label{differentiable hs1}
$f'(t)\ge 0$ for all $t\in(0,+\infty)$.
\item\label{differentiable hs2}
$f'(t)\le 0$ for all $t\in(0,+\infty)$, and 
$\lim_{x\searrow 0}f'(x)\hsn{x}{q}(\rho\|\sigma)$
is finite.
\item\label{differentiable hs3}
$\int_{(0,a]}f'(t)\,d\bz\hsp{(\valt)}{q}(\rho\|\sigma)\jz(t)$ is finite for some (and hence for all)
$a\in(0,+\infty)$. 
\end{enumerate}
Let $(\hsp{t}{q})_{t\in\bR_+}$ be a bounded and monotone family of non-negative hockey stick divergences, and let 
$\rho,\sigma\in\B(\hil)\p$.
Then
\begin{align}
\hsd{f}{q}(\rho\|\sigma)
=&f(0^+)\Tr\sigma-
\int_{(0,+\infty)}f'(t)\,d\bz\hsp{(\valt)}{q}(\rho\|\sigma)\jz(t)\nn\\
&+\lim_{x\to+\infty}f'(x)\hsp{x}{q}(\rho\|\sigma)-\lim_{x\searrow 0}f'(x)\hsn{x}{q}(\rho\|\sigma),
\label{eq:hsfdiv and NPerror}
\end{align}
which is well defined in the sense that there is no addition of infinities with different signs.
\end{prop}
\begin{proof}
Let us fix an $a\in(0,+\infty)$ such that 
the sign of $f'$ is constant on $(0,a]$ and on $(a,+\infty)$.
Note that 
by \eqref{eq:hsn def2},
\begin{align}
d\bz\hsn{(\valt)}{q}(\rho\|\sigma)\jz(t)
&=
d\bz\hsp{(\valt)}{q}(\rho\|\sigma)\jz(t)
-d\underbrace{\bz\fdiv{(\id-(\valt))}{}(\rho\|\sigma)\jz}_{=\Tr(\rho-(\valt)\sigma)}(t)\nn\\
&=
d\bz\hsp{(\valt)}{q}(\rho\|\sigma)\jz(t)+(\Tr\sigma)\,dt.
\label{eq:hsfdiv and NPerror proof1}
\end{align}
Note also that $d\bz\hsp{(\valt)}{q}(\rho\|\sigma)\jz(t)$ is a finite negative measure, and
hence $d\mu(t):=-d\bz\hsp{(\valt)}{q}(\rho\|\sigma)\jz(t)$ is a finite positive measure on 
$(0,+\infty)$.

For any $b\in(a,+\infty)$, a simple integration by parts yields
\begin{align}
&\int_{(a,b]}\hsp{t}{q}(\rho\|\sigma)\,df'(t)\nn\\
&\ds=
f'(b^+)\hsp{b}{q}(\rho\|\sigma)
-
f'(a^+)\hsp{a}{q}(\rho\|\sigma)
-\int_{(a,b]}f'(t)\,d\bz\hsp{(\valt)}{q}(\rho\|\sigma)\jz(t),
\label{eq:hsfdiv and NPerror proof3-1}
\end{align}
and taking the limit $b\to+\infty$ gives
\begin{align}
&\int_{(a,+\infty)}\hsp{t}{q}(\rho\|\sigma)\,df'(t)\nn\\
&\ds=
\lim_{x\to+\infty}f'(x)\hsp{x}{q}(\rho\|\sigma)
-
f'(a^+)\hsp{a}{q}(\rho\|\sigma)
+\int_{(a,+\infty)}f'(t)\,d\mu(t).
\label{eq:hsfdiv and NPerror proof3}
\end{align}
If $f'$ is non-positive on $(a,+\infty)$ then (since it is also monotone increasing, and hence bounded), all terms above are finite. 
If $f'$ is non-negative on $(a,+\infty)$ then any term in \eqref{eq:hsfdiv and NPerror proof3} that is not finite is equal to $+\infty$, whence the expression in \eqref{eq:hsfdiv and NPerror proof3} is well defined.

For any $b\in(0,a)$,
\begin{align}
&\int_{(b,a]}\hsn{t}{q}(\rho\|\sigma)\,df'(t)\nn\\
&\ds=
f'(a^+)\hsn{a}{q}(\rho\|\sigma)
-
f'(b^+)\hsn{b}{q}(\rho\|\sigma)
-\int_{(b,a]}f'(t)\,d\bz\hsn{(\valt)}{q}(\rho\|\sigma)\jz(t)\nn\\
&\ds=
f'(a^+)\hsn{a}{q}(\rho\|\sigma)
-
f'(b^+)\hsn{b}{q}(\rho\|\sigma)
-\int_{(b,a]}f'(t)(\Tr\sigma)\,dt
-\int_{(b,a]}f'(t)\,d\bz\hsp{(\valt)}{q}(\rho\|\sigma)\jz(t)\nn\\
&\ds=
f'(a^+)\hsn{a}{q}(\rho\|\sigma)
-
f'(b^+)\hsn{b}{q}(\rho\|\sigma)
-(\Tr\sigma)f(a)+(\Tr\sigma)f(b)
+\int_{(b,a]}f'(t)\,d\mu(t),
\label{eq:hsfdiv and NPerror proof2}
\end{align}
where the first equality follows by using integration by parts, 
in the second equality we used 
\eqref{eq:hsfdiv and NPerror proof1},
and the third equality is obvious.
Clearly, each term in \eqref{eq:hsfdiv and NPerror proof2} has a limit 
as $b\searrow 0$, and they are all finite except possibly for
\begin{align*}
l_1:=\lim_{b\searrow 0}f'(b)\hsn{b}{q}(\rho\|\sigma)
\ds\ds\text{and}\ds\ds
l_2:=\lim_{b\searrow 0}\int_{(b,a]}f'(t)\,d\mu(t)=\int_{(0,a]}f'(t)\,d\mu(t).
\end{align*}
Taking the limit $b\searrow 0$ in \eqref{eq:hsfdiv and NPerror proof2} yields
\begin{align}
\int_{(0,a]}\hsn{t}{q}(\rho\|\sigma)\,df'(t)
&=
f'(a^+)\hsn{a}{q}(\rho\|\sigma)
-(\Tr\sigma)f(a)+(\Tr\sigma)f(0^+)
-
l_1+l_2,
\label{eq:hsfdiv and NPerror proof2-1}
\end{align}
whenever the RHS above is well defined.

Under assumption \ref{differentiable hs1}, both $l_1$ and $l_2$ are finite, 
whence the RHS of \eqref{eq:hsfdiv and NPerror proof2-1} is well defined, and
its sum with the RHS of \eqref{eq:hsfdiv and NPerror proof2} is also well defined.

Under assumption \ref{differentiable hs2}, $l_1$ is finite and the RHS of \eqref{eq:hsfdiv and NPerror proof2} is also finite, 
whence the RHS of \eqref{eq:hsfdiv and NPerror proof2-1} is well defined, and
its sum with the RHS of \eqref{eq:hsfdiv and NPerror proof2} is also well defined.

Finally, under assumption \ref{differentiable hs3}, $l_2$ is finite,
whence the RHS of \eqref{eq:hsfdiv and NPerror proof2-1} is well defined, and
if it is not finite then it is equal to $+\infty$, and therefore in either case, its 
sum with the RHS of \eqref{eq:hsfdiv and NPerror proof2} is also well defined.

Hence, under either of the assumptions 
\ref{differentiable hs1}--\ref{differentiable hs3},
the sum of the terms on the RHSs of \eqref{eq:hsfdiv and NPerror proof2}
and \eqref{eq:hsfdiv and NPerror proof2-1} is well-defined, and 
substituting it 
into \eqref{eq:hsfdiv def} yields
\begin{align*}
\hsd{f}{q}(\rho\|\sigma)
=&
f(a)\Tr \sigma+
\underbrace{f'(a^+)\Tr(\rho-a\sigma)+
f'(a^+)\hsn{a}{q}(\rho\|\sigma)
-
f'(a^+)\hsp{a}{q}(\rho\|\sigma)}_{=0}\\
&-f(a)\Tr\sigma+f(0)\Tr\sigma
-
\int_{(0,+\infty)}f'(t)\,d\bz\hsp{(\valt)}{q}(\rho\|\sigma)\jz(t)\\
&+\lim_{x\to+\infty}f'(x)\hsp{x}{q}(\rho\|\sigma)-\lim_{x\searrow 0}f'(x)\hsn{x}{q}(\rho\|\sigma)\\
=&f(0^+)\Tr\sigma-
\int_{(0,+\infty)}f'(t)\,d\bz\hsp{(\valt)}{q}(\rho\|\sigma)\jz(t)\\
&+\lim_{x\to+\infty}f'(x)\hsp{x}{q}(\rho\|\sigma)-\lim_{x\searrow 0}f'(x)\hsn{x}{q}(\rho\|\sigma),
\end{align*}
proving \eqref{eq:hsfdiv and NPerror}.
\end{proof}

\begin{rem}
Assume in the setting of Proposition \ref{prop:differentiable hs} that there exists some $C\in(0,+\infty)$ such that 
\begin{align}
-d\bz\hsp{(\valt)}{q}(\rho\|\sigma)\jz(t)\le C\,dt.
\end{align} 
Then if $f'$ is non-positive on some non-degenerate interval $(0,a]$ then 
\begin{align}
\int_{(b,a]}f'(t)d\bz\hsp{(\valt)}{q}(\rho\|\sigma)\jz(t)
=
\int_{(b,a]}(-f'(t))\bz-d\bz\hsp{(\valt)}{q}(\rho\|\sigma)\jz(t)\jz
\le C(f(b)-f(a))
\end{align} 
for any $b\in(0,a)$, and taking the limit $b\searrow 0$ and using that $f(0^+)$ is finite, we see that 
assumption 
\ref{differentiable hs3} in Proposition \ref{prop:differentiable hs} is satisfied. 
\end{rem}

\begin{cor}\label{cor:qalpha for differentiable hs}
In the setting of Proposition \ref{prop:differentiable hs}, assume also that 
$(\hsp{t}{q})_{t\in\bR_+}$ is canonically bounded, and 
consider $f_{\alpha}(t):=-t^{\alpha}$, $t\in[0,+\infty)$  for some $\alpha\in(0,1)\cup(1,+\infty)$,
and assume that $\alpha\in(0,1)$, or $\alpha\in(1,+\infty)$ and $\rho^0\le\sigma^0$.
Then 
\begin{align}\label{eq:qalpha for differentiable hs}
\hsq{\alpha}{q}(\rho\|\sigma)=-\alpha\int_{(0,+\infty)}t^{\alpha-1}
\,d\bz\hsp{(\valt)}{q}(\rho\|\sigma)\jz(t)\,.
\end{align}
\end{cor}
\begin{proof}
Since $f_{\alpha}'(t)=-\alpha t^{\alpha-1}\le 0$ for every $t\in(0,+\infty)$, we have 
\begin{align*}
0\ge\lim_{x\searrow 0}f_{\alpha}'(x)\hsn{x}{q}(\rho\|\sigma)
\ge
\lim_{x\searrow 0}\alpha x^{\alpha-1}x\Tr\sigma
=
\alpha(\Tr\sigma)\lim_{x\searrow 0}x^{\alpha}
=0,
\end{align*}
where the inequality is due to \eqref{eq:hsn bounds}, and
\begin{align*}
0\ge\lim_{x\to+\infty}f_{\alpha}'(x)\hsp{x}{q}(\rho\|\sigma)
\begin{cases}
\ge\lim_{x\to+\infty}\alpha x^{\alpha-1}\Tr\rho=0,&\alpha\in(0,1),\\
=0,&\alpha>1\s\&\s\rho^0\le\sigma^0,
\end{cases}
\end{align*}
while 
\begin{align*}
\lim_{x\to+\infty}f_{\alpha}'(x)\hsp{x}{q}(\rho\|\sigma)
&\le
-\lim_{x\to+\infty}\alpha x^{\alpha-1}\Tr\rho(I-\sigma^0)=-\infty,
\ds\ds\ds \alpha>1\s\&\s\rho^0\nleq\sigma^0,
\end{align*}
where the inequalities are due to \eqref{eq:hsp bounds}.
Thus, 
Proposition \ref{prop:differentiable hs} yields \eqref{eq:qalpha for differentiable hs}.
\end{proof}

\begin{cor}\label{cor:measured hsq expressions}
Let $\rho,\sigma\in\B(\hil)\p$ and $\alpha\in(0,1)\cup(1,+\infty)$. 
If 
$\alpha\in(0,1)$, or $\alpha\in(1,+\infty)$ and $\rho^0\le\sigma^0$ then 
\begin{align}\label{eq:qalpha measured hs}
\hsq{\alpha}{\meas}(\rho\|\sigma)
=
\alpha\int_0^{+\infty}t^{\alpha-1}\Tr\sigma\{\rho-t\sigma>0\}\,dt
=
\alpha\int_0^{+\infty}t^{\alpha-1}\errii{t}{}{\rho}{\sigma}\,dt\,.
\end{align}
If $\alpha\in(0,1)$ then we further have
\begin{align}\label{eq:qalpha measured hs4}
\hsq{\alpha}{\meas}(\rho\|\sigma)
=
(1-\alpha)\int_0^{+\infty}t^{\alpha-2}\Tr\rho\{\rho-t\sigma<0\}\,dt
=
(1-\alpha)\int_0^{+\infty}t^{\alpha-2}\erri{t}{}{\rho}{\sigma}\,dt\,,
\end{align}
and if $\alpha\in(1,+\infty)$ then 
\begin{align}
\hsq{\alpha}{\meas}(\rho\|\sigma)
&=
(\alpha-1)\int_0^{+\infty}t^{\alpha-2}\Tr\rho\{\rho-t\sigma>0\}\,dt
=
(\alpha-1)\int_0^{+\infty}t^{\alpha-2}\bz 1-\erri{t}{}{\rho}{\sigma}\jz\,dt\,,
\label{eq:qalpha measured hs2}
\end{align}
independently of whether $\rho^0\le\sigma^0$ or not. 
\end{cor}
\begin{proof}
When 
$\alpha\in(0,1)$, or 
$\alpha\in(1,+\infty)$ and $\rho^0\le\sigma^0$, the 
expression in \eqref{eq:qalpha measured hs}
follows immediately from 
Corollary \ref{cor:qalpha for differentiable hs} and Lemma \ref{lemma:measured hsp derivative}.

When $\alpha\in(0,1)$, we may use  \eqref{eq:qalpha measured hs} 
combined with the symmetry relation \eqref{eq:hsq symmetry} for $q=\meas$
to obtain
\begin{align*}
\hsq{\alpha}{\meas}(\rho\|\sigma)
&=
\hsq{1-\alpha}{\meas}(\sigma\|\rho)
=
(1-\alpha)\int_0^{+\infty}t^{\alpha-1}\Tr\rho\underbrace{\{\sigma-t\rho>0\}}_{=\{\rho-(1/t)\sigma<0\}}\,dt\\
&=
(1-\alpha)\int_0^{+\infty}t^{\alpha-2}\Tr\rho\{\rho-t\sigma<0\}\,dt,
\end{align*}
where the last equality follows by the simple change of variables $u:=1/t$. This proves the first equality in 
\eqref{eq:qalpha measured hs4}, and the second equality follows by definition.

Note that 
\begin{align*}
\Tr\sigma\{\rho-t\sigma>0\}=\frac{1}{t}\left[\Tr\rho\{\rho-t\sigma>0\}-\Tr(\rho-t\sigma)_+\right].
\end{align*}
Substituting this into \eqref{eq:qalpha measured hs} yields
\begin{align}
\hsq{\alpha}{\meas}(\rho\|\sigma)
&=
\alpha\int_0^{+\infty}t^{\alpha-1}t\inv\Tr\rho\{\rho-t\sigma>0\}\,dt
-
\alpha\int_0^{+\infty}t^{\alpha-1}t\inv\Tr(\rho-t\sigma)_+\,dt
\label{eq:measured hsq expressions proof1}
\end{align}
whenever $\alpha\in(1,+\infty)$ and $\rho^0\le\sigma^0$. 
(Note that if $\alpha\in(0,1)$ and $\rho\ne 0$ then both integrals in 
\eqref{eq:measured hsq expressions proof1} are equal to $+\infty$, and hence the difference is not 
well defined.)

Now, if $\alpha\in(1,+\infty)$ then 
by Remark \ref{rem:a=0}, 
\begin{align}
\alpha\int_0^{+\infty}t^{\alpha-2}\Tr(\rho-t\sigma)_+\,dt=\frac{1}{\alpha-1}\hsq{\alpha}{\meas}(\rho\|\sigma).
\label{eq:measured hsq expressions proof2}
\end{align}
If $\rho^0\le\sigma^0$ then 
substituting \eqref{eq:measured hsq expressions proof2} into \eqref{eq:measured hsq expressions proof1} yields the first equality in
\eqref{eq:qalpha measured hs2}.
If $\rho^0\nleq\sigma^0$ then 
\begin{align*}
(\alpha-1)\int_0^{+\infty}t^{\alpha-2}\Tr\rho\{\rho-t\sigma>0\}\,dt
\ge
\underbrace{\Tr\rho(I-\sigma^0)}_{>0}(\alpha-1)\underbrace{\int_0^{+\infty}t^{\alpha-2}\,dt}_{=+\infty}=+\infty
=\hsq{\alpha}{\meas}(\rho\|\sigma),
\end{align*}
where the inequality is due to \eqref{eq:hsp bounds}, and the last equality 
is due to \eqref{eq:hsq infty}.
This proves the first equality in \eqref{eq:qalpha measured hs2}, 
and the second equality 
follows by definition.
\end{proof}

\begin{cor}\label{cor:measured hsq expressions2}
Let $\rho,\sigma\in\B(\hil)\p$ and $\alpha\in(0,1)$.
Then 
\begin{align}\label{eq:qalpha measured hs5}
\hsq{\alpha}{\meas}(\rho\|\sigma)
=
\alpha(1-\alpha)\int_0^{+\infty}t^{\alpha-2}\errm{t}{}{\rho}{\sigma}\,dt.
\end{align}
\end{cor}
\begin{proof}
The assertion follows immediately from \eqref{eq:qalpha measured hs}--\eqref{eq:qalpha measured hs4} as
\begin{align*}
\frac{\hsq{\alpha}{\meas}(\rho\|\sigma)}{\alpha(1-\alpha)}
=
\hsq{\alpha}{\meas}(\rho\|\sigma)\left[\frac{1}{\alpha}+\frac{1}{1-\alpha}\right]
=
\int_0^{+\infty}t^{\alpha-2}\left[t\errii{t}{}{\rho}{\sigma}+\erri{t}{}{\rho}{\sigma}\right]\,dt.
\end{align*}
\end{proof}

\section{Hockey stick $f$-divergences in von Neumann  algebras}\label{sec:V}

\subsection{Hockey stick divergences in von Neumann algebras}\label{sec:V.A}

Let $\M$ be a von Neumann algebra with the identity $\egy$ and the predual $\M_*$, i.e., the space of
$\sigma$-weakly continuous (or normal) functionals on $\M$. We write $\M_{*,\sa}$ for the set of self-adjoint
functionals in $\M_*$ and $\M_{*,\ge0}$ for the set of positive functionals in $\M_*$. For each
$\psi\in\M_{*,\sa}$ let $\psi=\psi_+-\psi_-$ be the Jordan decomposition of $\psi$ (see
\cite[Theorem 4.2, p.~140]{Takesaki} so that $|\psi|=\psi_++\psi_-$, the absolute value of $\psi$. We write
$\{\psi>0\}$ and $\{\psi<0\}$ for the support projections of $\psi_+$ and $\psi_-$ respectively and also
$\{\psi\ge0\}$ for $\{\psi<0\}^\perp=\egy-\{\psi<0\}$, i.e., $\{\psi\ge0\}=\{\psi>0\}+\{\psi=0\}$, where
$\{\psi=0\}$ is the projection onto the kernel of $\psi$.

In this section we may assume that $\M$ is represented in the standard form
\[
\bigl(\M,L^2(\M),J=*,L^2(\M)_{\ge0}\bigr),
\]
where $L^2(\M)$ is \emph{Haagerup's $L^2$-space} and $L^2(\M)_{\ge0}$ is its positive part. Moreover, in
this situation, $\M_*$ is isometrically order-isomorphic to \emph{Haagerup's $L^1$-space} $L^1(\M)$ by the
mapping $\psi\in\M_*\mapsto h_\psi\in L^1(\M)$; see \cite{Terp1981} and \cite[Chap.~9]{Hiai2021}. Under
the trace type functional $\tr:L^1(\M)\to\bC$ given by $\tr(h_\psi)=\psi(\egy)$, $\psi\in\M_*$, any
$\rho\in\M_{*,\ge0}$ is represented as
\[
\rho(A)=\langle h_\rho^{1/2},Ah_\rho^{1/2}\rangle_{L^2(\M)}=\tr\,Ah_\rho,\qquad A\in\M.
\]

One can extend the hockey stick divergences in Examples \ref{eq:Petz hs}--\ref{ex:maximal hsp} to the
von Neumann algebra setting as in the following (I)--(III).

\medskip
(I)\enspace
Let $\rho,\sigma\in\M_{*,\ge0}$ and let $\Delta_{\rho,\sigma}$ be the \emph{relative modular operator}
due to Araki \cite{Araki_relentrII}. We write the spectral decomposition of $\Delta_{\rho,\sigma}$ as
\begin{align}\label{relative modular}
\Delta_{\rho,\sigma}=\int_{[0,+\infty)}x\,dE_{\rho,\sigma}(x).
\end{align}
Then for any $t\in\bR$ the \emph{Petz-type hockey stick divergences} are defined by
\begin{align}
D_{(\id-t)_+}^{\mathrm{P}}(\rho\|\sigma)
&:=\int_{(0,+\infty)}(x-t)_+\,d\|E_{\rho,\sigma}(x)h_\sigma^{1/2}\|^2+\rho(\{\sigma=0\}), \label{Petz-hs+vN}\\
D_{(\id-t)_-}^{\mathrm{P}}(\rho\|\sigma)
&:=\int_{(0,+\infty)}(x-t)_-\,d\|E_{\rho,\sigma}(x)h_\sigma^{1/2}\|^2+t\sigma(\{\rho=0\}). \label{Petz-hs-vN}
\end{align}
See \cite{Hiai_fdiv_standard} and \cite[Chapter 2]{Hiai2021Div} for the definition of the Petz-type (or the
standard) $f$-divergences $S_f(\rho\|\sigma)$ (or $D_f^{\mathrm{P}}(\rho\|\sigma)$ in accordance with the
notation in the present paper).

In particular, when $\M=\B(\hil)$ with a general Hilbert space, $L^2(\M)$ and $L^1(\M)$ are realized as the
Hilbert space consisting of Hilbert--Schmidt operators and the space of trace-class operators on $\hil$,
reepectively. For $\rho,\sigma\in L^1(\M)_{\ge0}$ take the spectral decompositions
\[
\rho=\sum_{r\ge0}rP_r^\rho,\qquad\sigma=\sum_{s\ge0}sP_s^\sigma.
\]
Then the relative modular operator $\Delta_{\rho,\sigma}$ is written as
\[
\Delta_{\rho,\sigma}=L_\rho R_{\sigma^{-1}}=\sum_{r>0,s>0}rs^{-1}L_{P_r^\rho}R_{P_s^\sigma},
\]
where $L_{[-]}$ and $R_{[-]}$ denote the the left and the right multiplications and $\sigma^{-1}$ is the
generalized inverse (typically unbounded) of $\sigma$. The $D_{(\id-t)_+}^{\mathrm{P}}(\rho\|\sigma)$ in
\eqref{Petz-hs+vN} in this setting is given by
\begin{align*}
D_{(\id-t)_+}^{\mathrm{P}}(\rho\|\sigma)
&=\sum_{r>0,s>0}(rs^{-1}-t)_+\|P_r^\rho\sigma^{1/2}P_s^\sigma\|^2+\Tr\rho(P_0^\sigma) \\
&=\sum_{r>0,s>0}(r-ts)_+\Tr P_r^\rho P_s^\sigma+\sum_{r>0}r\Tr P_r^\rho P_0^\sigma \\
&=\sum_{r\ge0,s\ge0}(r-ts)_+\Tr P_r^\rho P_s^\sigma,
\end{align*}
which has the same expression as in Example \ref{eq:Petz hs} in the $\dim\hil<+\infty$ case. The
$D_{(\id-t)_-}^{\mathrm{P}}(\rho\|\sigma)$ in \eqref{Petz-hs-vN} is similarly given.

\medskip
(II)\enspace
The \emph{measured hockey stick divergences} of $\rho,\sigma\in\M_{*,\ge0}$ are defined in the same way
as in Example \ref{ex:hs meas} by
\begin{align*}
D_{(\id-t)_+}^\meas(\rho\|\sigma)&=(\rho-t\sigma)_+(\egy)=(\rho-t\sigma)(\{\rho-t\sigma>0\}) \\
&=\max\{(\rho-t\sigma)(A):A\in\M,\ 0\le A\le\egy\}, \\
D_{(\id-t)_-}^\meas(\rho\|\sigma)&=(\rho-t\sigma)_-(\egy)=(t\sigma-\rho)(\{t\sigma-\rho>0\}) \\
&=\max\{(t\sigma-\rho)(A):A\in\M,\ 0\le A\le\egy\}.
\end{align*}
It is clear that $t\mapsto D_{(\id-t)_+}^\meas$ is monotone decreasing and $t\mapsto D_{(\id-t)_-}^\meas$
is monotone increasing. Both $D_{(\id-t)_\pm}^\meas(\rho\|\sigma)$ are jointly convex in $\rho,\sigma$
as well as convex (hence continuous) in $t\in\bR$. Furthermore, $D_{(\id-t)_\pm}^\meas$ are monotone
under positive trace-preserving maps $\Phi:\M_*\to\N_*$ as
\begin{align}\label{mono-meas-vN}
D_{(\id-t)_\pm}^\meas(\Phi(\rho)\|\Phi(\sigma))\le D_{(\id-t)_\pm}^\meas(\rho\|\sigma),
\qquad\rho,\sigma\in\M_{*,\ge0},
\end{align}
where $\N$ is another von Neumann algebra and the trace-preserving of $\Phi$ means that
$(\Phi\psi)(\egy)=\psi(\egy)$ for $\psi\in\M_*$ (though $\M,\N$ are not assumed to have traces).

\medskip
(III)\enspace
Let $\mathrm{RT}(\rho\|\sigma)$ denote the set of reverse tests for $(\rho,\sigma)$, i.e., the set of
$(\tilde\rho,\tilde\sigma,\Gamma)$ of $\tilde\rho,\tilde\sigma\in\A_{*,\ge0}$, where $\A$ is a
\emph{commutative} von Neumann algebra, and a positive trace-preserving map $\Gamma:\A_*\to\M_*$
such that $\Gamma(\tilde\rho)=\rho$, $\Gamma(\tilde\sigma)=\sigma$. We may let
$\A=L^\infty(\Omega,\F,\mu)$ and $\A_*=L^1(\Omega,\F,\mu)$ on a $\sigma$-finite measure space
$(\Omega,\F,\mu)$. The \emph{maximal hockey stick divergences} are then defined by
\begin{align*}
D_{(\id-t)_+}^{\max}(\rho\|\sigma)&:=\inf\bigl\{D_{(\id-t)_+}(\tilde\rho\|\tilde\sigma):
(\tilde\rho,\tilde\sigma,\Gamma)\in\mathrm{RT}(\rho,\sigma)\bigr\}, \\
D_{(\id-t)_-}^{\max}(\rho\|\sigma)&:=\inf\bigl\{D_{(\id-t)_-}(\tilde\rho\|\tilde\sigma):
(\tilde\rho,\tilde\sigma,\Gamma)\in\mathrm{RT}(\rho,\sigma)\bigr\}.
\end{align*}
It is clear by definition that $t\mapsto D_{(\id-t)_+}^\meas$ is monotone decreasing and
$t\mapsto D_{(\id-t)_-}^\meas$ is monotone increasing. Moreover, both $D_{(\id-t)_\pm}^{\max}$ are
monotone under positive trace-preserving maps $\Phi:\M_*\to\N_*$ between von Neumann algebras.

\medskip
Matsumoto's formulation mentioned in Example \ref{ex:maximal hsp} is also valid in the von Neumann
algebra setting.

\begin{proposition}\label{prop:V.1}
For any $\rho,\sigma\in\M_{*,\ge0}$ and each $t\in\bR$,
\begin{align}
D_{|\id-t|}^{\max}(\rho\|\sigma)
&=(\rho+t\sigma)(\egy)-2\max\bigl\{\eta(\egy):\eta\in\M_{*,\ge0},\,\eta\le\rho,\,\eta\le t\sigma\bigr\},
\label{D-max-abso}\\
D_{(\id-t)_+}^{\max}(\rho\|\sigma)
&=\rho(\egy)-\max\bigl\{\eta(\egy):\eta\in\M_{*,\ge0},\,\eta\le\rho,\,\eta\le t\sigma\bigr\}, \label{D-max+}\\
D_{(\id-t)_-}^{\max}(\rho\|\sigma)
&=t\sigma(\egy)-\max\bigl\{\eta(\egy):\eta\in\M_{*,\ge0},\,\eta\le\rho,\,\eta\le t\sigma\bigr\}. \label{D-max-}
\end{align}
\end{proposition}

\begin{proof}
Let $(\tilde\rho,\tilde\sigma,\Gamma)$ be in $\mathrm{RT}(\rho,\sigma)$, where
$\tilde\rho,\tilde\sigma\in L^1(\Omega,\F,\mu)_{\ge0}$ and $\Gamma:L^1(\Omega,\F,\mu)\to\M_*$ with a
$\sigma$-finite measure space $(\Omega,\F,\mu)$. Since
\[
D_{(\id-t)_\pm}^{\max}(\tilde\rho\|\tilde\sigma)
={1\over2}\bigl(D_{|\id-t|}^{\max}(\tilde\rho\|\tilde\sigma)\pm D_{(\id-t)}^{\max}(\tilde\rho\|\tilde\sigma)\bigr)
={1\over2}\bigl(D_{|\id-t|}^{\max}(\tilde\rho\|\tilde\sigma)\pm(\rho-t\sigma)(\egy)\bigr),
\]
it suffices to show \eqref{D-max-abso} only. Define $\gamma:\bC^3=C(\{0,1,2\})\to\M_{*,\ge0}$ by
\[
\begin{cases}\gamma(\delta_0):=\Gamma(\tilde\rho\wedge(t\tilde\sigma))/c_0, \\
\gamma(\delta_1):=\Gamma((\tilde\rho-t\tilde\sigma)_+)/c_1, \\
\gamma(\delta_2):=\Gamma((t\tilde\sigma-\tilde\rho)_+)/c_2,\end{cases}
\]
where $c_0,c_1,c_2$ are the normalizations, i.e., $c_0:=\|\tilde\rho\wedge(t\tilde\sigma)\|_1$,
$c_1:=\|(\tilde\rho-t\tilde\sigma)_+\|_1$ and $c_2:=\|(t\tilde\sigma-\tilde\rho)_+\|_1$. Let $p:=(c_0,c_1,0)$
and $tq:=(c_0,0,c_2)$. Then $(p,q,\gamma)\in\mathrm{RT}(\rho,\sigma)$; indeed,
\begin{align*}
\gamma(p)&=\Gamma(\tilde\rho\wedge(t\tilde\sigma))+\Gamma((\tilde\rho-t\tilde\sigma)_+)
=\Gamma(\tilde\rho)=\rho, \\
t\gamma(q)&=\Gamma(\tilde\rho\wedge(t\tilde\sigma))+\Gamma((t\tilde\sigma-\tilde\rho)_+)
=\Gamma(t\tilde\sigma)=t\sigma.
\end{align*}
Moreover,
\[
\|\tilde\rho-t\tilde\sigma\|_1=\|(\tilde\rho-t\tilde\sigma)_+\|_1+\|(\tilde\rho-t\tilde\sigma)_-\|_1
=c_1+c_2=\|p-tq\|_1.
\]
Now let $\eta:=\Gamma(\tilde\rho\wedge(t\tilde\sigma))\in\M_{*,\ge0}$; then $0\le\eta\le\rho$,
$\eta\le t\sigma$ and
\[
(\rho+t\sigma)(\egy)-2\eta(\egy)=\|\tilde\rho+t\tilde\sigma\|_1-2\|\tilde\rho\wedge(t\tilde\sigma)\|_1
=\|\tilde\rho-t\tilde\sigma\|_1.
\]
Conversely, choose an $\eta\in\M_{*,\ge0}$ maximizing $\eta(\egy)$ under $\eta\le\rho$ and $\eta\le t\sigma$,
where the existence of the maximizer is due to the $\sigma(\M_*,\M)$-compactness of
$\{\eta\in\M_*^+:0\le\eta\le\rho,\ \eta\le t\sigma\}$ as immediately seen from \cite[Theorem III.5.4]{Takesaki}.
Define $\gamma:\bC^3\to\M_*$ by
\[
\begin{cases}\gamma(\delta_0):=\eta/\eta(\egy), \\
\gamma(\delta_1):=(\rho-\eta)/(\rho-\eta)(\egy), \\
\gamma(\delta_2):=(t\sigma-\eta)/(t\sigma-\eta)(\egy),\end{cases}
\]
and $p:=(\eta(\egy),(\rho-\eta)(\egy),0)$ and $tq=(\eta(\egy),0,(t\sigma-\eta)(\egy))$. Then
$(p,q,\gamma)\in\mathrm{RT}(\rho,\sigma)$ and
\[
\|p-tq\|_1=(\rho-\eta)(\egy)+(t\sigma-\eta)(\egy)=(\rho+t\sigma)(\egy)-2\eta(\egy).
\]
Therefore, we have \eqref{D-max-abso}.
\end{proof}

Now let $f\in(\bR^{(0,+\infty)})_{\mathrm{conv}}$ and $\rho,\sigma\in\M_{*,\ge0}$. For
$q\in\{\mathrm{P},\meas,\max\}$, we can define, for any choice of $a\in(0,+\infty)$, the \emph{hockey stick
$f$-divergence} $D_f^{q,\hs}$ by
\begin{align}\label{D-f-q-hs}
D_f^{q,\hs}(\rho\|\sigma):=f(a)\sigma(\egy)+f'(a^+)(\rho-a\sigma)(\egy)
+\int_{(0,a]}D_{(\id-t)_-}^q(\rho\|\sigma)\,df'(t)+\int_{(a,+\infty)}D_{(\id-t)_+}^q(\rho\|\sigma)\,df'(t).
\end{align}
Indeed, since $D_{(\id-t)_+}^q(\rho\|\sigma)-D_{(\id-t)_-}^q(\rho\|\sigma)=(\rho-t\sigma)(\egy)$,
it is seen as in the proof of Lemma \ref{lemma:hsd def} that the definition \eqref{D-f-q-hs} is independent of
the choice of $a\in(0,+\infty)$. Furthermore, the above definition can be applied to any bounded
family $(D_{(\id-t)_+}^q)_{t\in(0,+\infty)}$ of non-negative quantum hockey stick divergences
as in Section \ref{sec:IV.A} in the von Neumann algebra setting. In particular, if all $D_{(\id-t)_+}^q$ are
monotone under CPTP maps, then as in Corollary \ref{cor:hs fdiv order} it is immediately seen from the
definition that
\begin{align}\label{hs fdiv order vN}
D_f^\meas\le D_f^{\meas,hs}\le D_f^{q,\hs}\le D_f^{\max,\hs}\le D_f^{\max}
\end{align}
for any $f\in(\bR^{(0,+\infty)})_{\mathrm{conv}}$.


As expected, $D_f^{\mathrm{P},\hs}$ coincides with the (standard) Petz-type $f$divergence
$D_f^{\mathrm{P}}$ (or $S_f$). We give a proof for completeness.

\begin{proposition}\label{prop:V.2}
For any $f\in(\bR^{(0,+\infty)})_{\mathrm{conv}}$ and $\rho,\sigma\in\M_{*,\ge0}$,
$D_f^{\mathrm{P},\hs}(\rho\|\sigma)=D_f^{\mathrm{P}}(\rho\|\sigma)$.
\end{proposition}

\begin{proof}
First, recall the definition
\begin{align}\label{standard S_f}
D_f^{\mathrm{P}}(\rho\|\sigma):=\int_{(0,+\infty)}f(x)\,d\|E_{\rho,\sigma}(x)h_\sigma^{1/2}\|^2
+f(0^+)\sigma(\{\rho=0\})+f'(\infty)\rho(\{\sigma=0\})
\end{align}
with the spectral decomposition \eqref{relative modular}. We insert \eqref{eq:convexf repr4} into
\eqref{standard S_f} and use the Fubini theorem to write
\begin{align}
D_f^{\mathrm{P}}(\rho\|\sigma)
&=f(a)\int_{(0,+\infty)}d\|E_{\rho,\sigma}(x)h_\sigma^{1/2}\|^2
+f'(a^+)\int_{(0,+\infty)}(x-a)\,d\|E_{\rho,\sigma}(x)h_\sigma^{1/2}\|^2 \nonumber\\
&\qquad+\int_{(0,a]}\biggl(\int_{(0,+\infty)}(x-t)_-\,d\|E_{\rho,\sigma}(x)h_\sigma^{1/2}\|^2\biggr)\,df'(t)
\nonumber\\
&\qquad+\int_{(a,+\infty)}\biggl(\int_{(0,+\infty)}(x-t)_+\,d\|E_{\rho,\sigma}(x)h_\sigma^{1/2}\|^2\biggr)\,df'(t)
\nonumber\\
&\qquad+f(0^+)\sigma(\{\rho=0\})+f'(+\infty)\rho(\{\sigma=0\}) \nonumber\\
&=(f(a)-af'(a^+))\sigma(\{\rho>0\})+f'(a^+)\rho(\{\sigma>0\}) \nonumber\\
&\qquad+\int_{(0,a]}\biggl(\int_{(0,+\infty)}(x-t)_-\,d\|E_{\rho,\sigma}(x)h_\sigma^{1/2}\|^2\biggr)\,df'(t)
\nonumber\\
&\qquad+\int_{(a,+\infty)}\biggl(\int_{(0,+\infty)}(x-t)_+\,d\|E_{\rho,\sigma}(x)h_\sigma^{1/2}\|^2\biggr)\,df'(t)
\nonumber\\
&\qquad+f(0^+)\sigma(\{\rho=0\})+f'(+\infty)\rho(\{\sigma=0\}). \label{standard S_f2}
\end{align}

On the other hand, we insert \eqref{Petz-hs+vN} and \eqref{Petz-hs-vN} into \eqref{D-f-q-hs} to write
\begin{align}
D^{\mathrm{P},\hs}_f(\rho\|\sigma)&=f(a)\sigma(\egy)+f'(a^+)(\rho(\egy)-a\sigma(\egy)) \nonumber\\
&\qquad+\int_{(0,a]}\biggl(\int_{(0,+\infty)}(x-t)_-\,d\|E_{\rho,\sigma}(x)h_\sigma^{1/2}\|^2\biggr)\,df'(t)
\nonumber\\
&\qquad+\int_{(a,+\infty)}\biggl(\int_{(0,+\infty)}(x-t)_+\,d\|E_{\rho,\sigma}(x)h_\sigma^{1/2}\|^2\biggr)\,df'(t)
\nonumber\\
&\qquad+\biggl(\int_{(0,a]}t\,df'(t)\biggr)\sigma(\{\rho=0\})+(f'(+\infty)-f'(a^+))\rho(\{\sigma=0\}).
\label{D-f-P-hs2}
\end{align}
From \eqref{standard S_f2} and \eqref{D-f-P-hs2} it suffices to see the equality
\begin{align*}
&(f(a)-af'(a^+))\sigma(\{\rho>0\})+f'(a^+)\rho(\{\sigma>0\})
+f(0^+)\sigma(\{\rho=0\})+f'(+\infty)\rho(\{\sigma=0\}) \nonumber\\
&\quad=f(a)\sigma(\egy)+f'(a^+)(\rho(\egy)-a\sigma(\egy))
+\biggl(\int_{(0,a]}t\,df'(t)\biggr)\sigma(\{\rho=0\})+(f'(+\infty)-f'(a^+))\rho(\{\sigma=0\}),
\end{align*}
which can directly be checked by using integration by parts $\int_{(0,a]}t\,df'(t)=af'(a^+)-f(a)+f(0^+)$.
\end{proof}

\subsection{Differentiability of measured hockey stick divergences}\label{sec:V.B}

Our aim of this subsection is to prove the next theorem, which extends
Lemma \ref{lemma:measured hsp derivative} to the von Neumann algebra setting.
Throughout the sequel of Section \ref{sec:V} let $\rho,\sigma\in\M_{*,\ge0}$ be arbitrarily given.

\begin{theorem}\label{theorem:V.3}
For every $t\in\bR$ we have
\begin{align}
\partial_+(\rho-t\sigma)_+(\egy)&=-\sigma(\{\rho-t\sigma>0\}), \label{deriv+}\\
\partial_-(\rho-t\sigma)_+(\egy)&=-\sigma(\{\rho-t\sigma\ge0\}). \label{deriv-}
\end{align}
\end{theorem}

To prove the theorem, we give a few preliminary lemmas.

\begin{lemma}\label{lemma:V.4}
For every $t\in\bR$ we have
\begin{align}
\partial_+(\rho-t\sigma)_+(\egy)&\ge-\sigma(\{\rho-t\sigma>0\}), \label{deriv+ge}\\
\partial_-(\rho-t\sigma)_+(\egy)&\le-\sigma(\{\rho-t\sigma\ge0\}). \label{deriv-le}
\end{align}
Moreover, if $\partial_-(\rho-t\sigma)_+(\egy)=\partial_+(\rho-t\sigma)_+(\egy)$, then
\[
\partial_-(\rho-t\sigma)_+(\egy)=\partial_+(\rho-t\sigma)_+(\egy)
=-\sigma(\{\rho-t\sigma>0\})=-\sigma(\{\rho-t\sigma\ge0\}).
\]
\end{lemma}

\begin{proof}
We notice that
\begin{align*}
(\rho-(t+\delta)\sigma)_+(\egy)&\ge[(\rho-t\sigma)_+-(\rho-t\sigma)_--\delta\sigma](\{\rho-t\sigma>0\}) \\
&=(\rho-t\sigma)_+(\egy)-\delta\sigma(\{\rho-t\sigma>0\}),
\end{align*}
and
\begin{align*}
(\rho-(t-\delta)\sigma)_+(\egy)&\ge[(\rho-t\sigma)_+-(\rho-t\sigma)_-+\delta\sigma](\{\rho-t\sigma\ge0\}) \\
&=(\rho-t\sigma)_+(\egy)+\delta\sigma(\{\rho-t\sigma\ge0\}).
\end{align*}
Therefore,
\begin{align}\label{driv-limit+}
\partial_+(\rho-t\sigma)_+(\egy)
=\lim_{\delta\searrow0}{(\rho-(t+\delta)\sigma)_+(\egy)-(\rho-t\sigma)_+(\egy)\over\delta}
\ge-\sigma(\{\rho-t\sigma>0\}),
\end{align}
and
\begin{align}\label{deriv-limit-}
\partial_-(\rho-t\sigma)_+(\egy)
=\lim_{\delta\searrow0}{(\rho-(t-\delta)\sigma)_+(\egy)-(\rho-t\sigma)_+(\egy)\over-\delta}
\le-\sigma(\{\rho-t\sigma\ge0\}).
\end{align}
If $\partial_-(\rho-t\sigma)_+(\egy)=\partial_+(\rho-t\sigma)_+(\egy)$, then it follows from \eqref{driv-limit+}
and \eqref{deriv-limit-} that
\[
\sigma(\{\rho-t\sigma\ge0\})\le\sigma(\{\rho-t\sigma>0\}).
\]
Since $\sigma(\{\rho-t\sigma>0\})\le\sigma(\{\rho-t\sigma\ge0\})$ obviously, we have the assertion.
\end{proof}

The next lemma plays an essential role to prove the theorem.

\begin{lemma}\label{lemma:V.5}
For any sequence $\{t_j\}$ in $[0,+\infty)$ with $t_j\searrow t_0\ge0$ and any $\eta\in\M_{*\ge0}$,
\[
\eta(\{\rho-t_0\sigma>0\})\le\liminf_{j\to\infty}\eta(\{\rho-t_j\sigma>0\})
\le\sup_{j\ge1}\eta(\{\rho-t_j\sigma>0\}).
\]
\end{lemma}

Let us first give the proof of Theorem \ref{theorem:V.3} under assuming that
Lemma \ref{lemma:V.5} has been proved. The proof of the lemma will be given later on.

\begin{proof}[Proof of Theorem \ref{theorem:V.3}]
Note that $(\rho-t\sigma)_+(\egy)=(\rho-t\sigma)(\egy)$ and $\sigma(\{\rho-t\sigma\ge0\})=\sigma(\egy)$
for all $t\le0$ as well as $\sigma(\{\rho-t\sigma>0\})=\sigma(\egy)$ for all $t<0$. Hence it suffices to show
\eqref{deriv+} for $t\ge0$ and \eqref{deriv-} for $t>0$. For any $t_0\ge0$, since
$t\mapsto(\rho-t\sigma)_+(\egy)$ is convex on $\bR$, by Lemma \ref{lemma:V.4} there exists a sequence
$\{t_j\}$ with $t_j>t_0$ and $t_j\searrow t_0$ such that
\[
\partial_+(\rho-t_j\sigma)_+(\egy)=-\sigma(\{\rho-t_j\sigma>0\}),\qquad j\ge1.
\]
As $\partial_+(\rho-t_j\sigma)_+(\egy)\searrow\partial_+(\rho-t_0\sigma)_+(\egy)$, we have by
Lemma \ref{lemma:V.5}
\[
\partial_+(\rho-t_0\sigma)_+(\egy)=\inf_{j\ge1}\partial_+(\rho-t_j\sigma)_+(\egy)
=-\sup_{j\ge1}\sigma(\{\rho-t_j\sigma>0\})\le-\sigma(\{\rho-t_0\sigma>0\}),
\]
which with \eqref{deriv+ge} yields \eqref{deriv+}. 

On the other hand, for any $t_0>0$ by Lemma \ref{lemma:V.4} there exists a sequence $\{s_j\}$ with
$0<s_j<t_0$ and $s_j\nearrow t_0$ such that
\begin{align}\label{derivconv-}
\partial_-(\rho-s_j\sigma)_+(\egy)=-\sigma(\{\rho-s_j\sigma\ge0\}),\qquad j\ge1.
\end{align}
Note that for any $t>0$,
\[
\sigma(\{\rho-t\sigma\ge0\})=\sigma(\egy-\{t\sigma-\rho>0\})
=\sigma(\egy)-\sigma(\{\sigma-t^{-1}\rho>0\}).
\]
Therefore, as $s_j^{-1}\searrow t_0^{-1}$, by Lemma \ref{lemma:V.5} we have
\begin{align*}
\sigma(\{\rho-t_0\sigma\ge0\})&=\sigma(\egy)-\sigma(\{\sigma-t_0^{-1}\rho>0\})
\ge\inf_{j\ge1}[\sigma(\egy)-\sigma(\{\sigma-s_j^{-1}\rho>0\})] \\
&=\inf_{j\ge1}\sigma(\{\rho-s_j\sigma\ge0\})=-\sup_{j\ge1}\partial_-(\rho-s_j\sigma)_+(\egy)
\end{align*}
thanks to \eqref{derivconv-}.
As $\partial_-(\rho-s_j\sigma)_+(\egy)\nearrow\partial_-(\rho-t_0\sigma)_+(\egy)$, we have
\[
\sigma(\{\rho-t_0\sigma\ge0\})\ge-\partial_-(\rho-t_0\sigma)_+(\egy),
\]
which with \eqref{deriv-le} yields \eqref{deriv-}.
\end{proof}

Now let us prove Lemma \ref{lemma:V.5}. To do so, we first recall the norm continuity of the mapping
$\psi\mapsto|\psi|$ on $\M_*$ due to Kosaki \cite{Kosaki_conti}, which says that for every
$\psi,\omega\in\M_*$,
\begin{align}\label{Kosaki:conti}
\|\,|\psi|-|\omega|\,\|\le(2\,\|\psi+\omega\|\,\|\psi-\omega\|)^{1/2}.
\end{align}
This implies that the mapping $\psi\in\M_{*,\sa}\mapsto\psi_+=(\psi+|\psi|)/2\in\M_{*,\ge0}$ is continuous
in the norm.

An ingredient of the proof below is a small part of Takesaki's structure theorem of von Neumann algebras
(of type III). Let $\M$ be a general von Neumann algebra $\M$ on a Hilbert space $\hil$, and $\vfi_0$ be
a faithful normal semi-finite weight on $\M$. Let $\mathcal{N}:=\M\rtimes_{\sigma^{\vfi_0}}\bR$ be the
crossed product von Neumann algebra by the modular automorphism group $\sigma^{\vfi_0}$, which is
represented on the Hilbert space $\hil\otimes L^2(\bR)$. Then it is known that $\mathcal{N}$ is semi-finite
with the canonical semi-finite trace $\tau$; see \cite[Lemma 2]{takesaki1973} (also
\cite[Theorem 8.15]{Hiai2021}). Since we have a *-isomorphism of $\M$ into $\mathcal{N}$, we can consider
$\M$ as a von Neumann subalgebra of $\mathcal{N}$. Also, let $\widetilde{\mathcal{N}}$ be the space of
$\tau$-measurable operators affiliated with $\mathcal{N}$; see \cite{Nelson1974} and
\cite[Chap.~4]{Hiai2021}). Another ingredient of the proof is Haagerup's $L^1$-space $L^1(\M)$ that is
constructed inside $\widetilde{\mathcal{N}}$, which is isometrically isometric to $\M_*$ as mentioned in the
beginning of Section \ref{sec:V}.

\begin{proof}[Proof of Lemma \ref{lemma:V.5}]
Let $\{t_j\}$ and $\eta$ be as in Lemma \ref{lemma:V.5}. Since $\rho-t_j\sigma\to\rho-t_0\sigma$ in the
norm, it suffices to prove that the mapping $\psi\in\M_{*,\sa}\mapsto\eta(\{\psi>0\})$ is lower semicontinuous
in the norm. Moreover, since $\psi\in\M_{*,\sa}\mapsto\psi_+\in\M_{*,\ge0}$ is norm-continuous as
mentioned above, we may prove that $\psi\in\M_{*,\ge0}\mapsto\eta(\{\psi>0\})$ is lower semicontinuous
in the norm. To show this, we consider $\M$ as a von Neumann subalgebra of
$\mathcal{N}:=\M\rtimes_{\sigma^{\vfi_0}}\bR$ and use $L^1(\M)\subset\widetilde{\mathcal{N}}$, as
mentioned above. Now let $\psi,\psi_j$ ($j\in\bN$) be in $\M_{*,\ge0}$ such that
$\|h_{\psi_j}-h_\psi\|_1=\|\psi_j-\psi\|\to0$. Since the $\|\cdot\|_1$-norm topology on $L^1(\M)$ coincides with
the relative topology induced from $\widetilde{\mathcal{N}}$ with the measure topology (see
\cite[Chap.~II, Proposition 26]{Terp1981}, \cite[Theorem 9.18]{Hiai2021}), it follows that
$h_{\phi_j}\to h_\phi$ in the measure topology in $\widetilde{\mathcal{N}}$. Hence by
\cite[Proposition 3.7]{HiaiKosaki} we havev $(\egy+h_{\phi_j})^{-1}\to(\egy+h_\phi)^{-1}$ in the strong
operator topology, so that $h_{\psi_j}(\egy+h_{\psi_j})^{-1}\to h_\psi(\egy+h_\psi)^{-1}$ in the same sense
by \cite[Lemma A.1]{Hiai2021}. For each $\delta>0$ define a continuous function $f_\delta$ on $[0,\infty)$ by $f_\delta(x):=\min\{\delta^{-1}x,1\}$ for $x\ge0$. Then by \cite{Kadison1968} or
\cite[Theorem A.2]{Stratila1981} we have
\[
f_\delta(h_{\psi_j}(\egy+h_{\psi_j})^{-1})\to f_\delta(h_\psi(\egy+h_\psi)^{-1})
\ \ \mbox{in the strong operator topology as $j\to\infty$}.
\]
In particular, for any $\delta>0$ and any $\xi\in\hil\otimes L^2(\bR)$ it follows that
$\psi\in\M_{*,\ge0}\mapsto\langle\xi,f_\delta(h_\psi(\egy+h_\psi)^{-1})\xi\rangle$ is continuous in the norm.
Since $f_\delta(h_\phi(\egy+h_\phi)^{-1})$ converges increasingly as $\delta\searrow0$ to the support projection
of $h_\psi$ that coincides with $\{\phi>0\}$, we furthermore have
\[
\langle\xi,\{\psi>0\}\xi\rangle=\sup_{\delta>0}\langle\xi,f_\delta(h_\psi(\egy+h_\psi)^{-1})\xi\rangle.
\]
Therefore, $\psi\in\M_*^+\mapsto\langle\xi,\{\psi>0\}\xi\rangle$ is lower semicontinuous in the norm. When
we consider $\M$ ($\subset\mathcal{N}$) as represented on $\hil\otimes L^2(\bR)$, there exists a sequence
$\{\xi_k\}_{k=1}^\infty$ in $\hil\otimes L^2(\bR)$ with $\sum_{k=1}^\infty\|\xi_k\|^2<+\infty$ such that
$\eta(A)=\sum_{k=1}^\infty\langle\xi_k,A\xi_k\rangle$, $A\in\M$; see, e.g.,
\cite[Theorem II.2.6]{Takesaki}. Hence we can write
\[
\eta(\{\psi>0\})=\sup_{K\ge1}\sum_{k=1}^K\langle\xi_k,\{\psi>0\}\xi_k\rangle,\qquad\psi\in\M_{*,\ge0}.
\]
which implies that $\psi\in\M_{*,\ge0}\mapsto\eta(\{\psi>0\})$ is lower semicontinuous in the norm.
\end{proof}

For any $\rho,\sigma\in\M_{*,\ge0}$ and $t\in(0,+\infty)$ let
\begin{align*}
&\beta_{0,t}(\rho\|\sigma):=\rho(\{\rho-t\sigma\le0\}),\qquad
\beta_{1,t}(\rho\|\sigma):=\sigma(\{\rho-t\sigma>0\}), \\
&\beta_{\mathrm{mix},t}(\rho\|\sigma):=\beta_{0,t}(\rho\|\sigma)+\beta_{1,t}(\rho\|\sigma).
\end{align*}
Then the following is Corollary \ref{cor:IV.20} in the von Neumann algebra case.

\begin{cor}\label{cor:V.6}
For any $\rho,\sigma\in\M_{*,\ge0}$,
\begin{align}
\lim_{t\to+\infty}\rho(\{\rho-t\sigma>0\})&=\lim_{t\to+\infty}\rho(\{\rho-t\sigma\ge0\})=\rho(\egy-\{\sigma>0\}),
\label{lim-1}\\
\lim_{t\to+\infty}\beta_{0,t}(\rho\|\sigma)&=\lim_{t\to+\infty}\rho(\{\rho-t\sigma<0\})=\rho(\{\sigma>0\}),
\label{lim-2}\\
\lim_{t\to+\infty}t\beta_{1,t}(\rho\|\sigma)&=\lim_{t\to+\infty}t\sigma(\{\rho-t\sigma>0\})
=\lim_{t\to+\infty}t\sigma(\{\rho-t\sigma\ge0\})=0, \label{lim-3}\\
\lim_{t\to+\infty}\beta_{\mathrm{mix},t}(\rho\|\sigma)&=\rho(\{\sigma>0\}), \label{lim-4}\\
\lim_{t\to+\infty}(\rho-t\sigma)_+(\egy)&=\rho(\egy-\{\sigma>0\}). \label{lim-5}
\end{align}
\end{cor}

\begin{proof}
\eqref{lim-1}.\enspace
By Theorem \ref{theorem:V.3} we have
\begin{align*}
\lim_{t\to+\infty}\rho(\{\rho-t\sigma>0\})
&=\lim_{t\to+\infty}\rho(\{t^{-1}\rho-\sigma>0\})=\lim_{s\searrow0}\rho(\{\sigma-s\rho<0\}) \\
&=\rho(\egy)-\lim_{s\searrow0}\rho(\{\sigma-s\rho\ge0\})
=\rho(\egy)+\lim_{s\searrow0}\partial_-(\sigma-s\rho)_+(1) \\
&=\rho(\egy)+\partial_+(\sigma-s\rho)_+(1)\big|_{s=0}=\rho(\egy)-\rho(\{\sigma>0\}) \\
&=\rho(\egy-\{\sigma>0\}).
\end{align*}
The second equality is similar.

\eqref{lim-2}.\enspace
Since $\beta_{0,t}(\rho\|\sigma)=\rho(\egy)-\rho(\{\rho-t\sigma>0\})$, \eqref{lim-1} gives
\[
\lim_{t\to+\infty}\beta_{0,t}(\rho\|\sigma)=\rho(\egy)-\rho(\egy-\{\sigma>0\})=\rho(\{\sigma>0\}).
\]
The second equality is similar.

\eqref{lim-3}.\enspace
Since the function $\phi(t):=(\rho-t\sigma)_+(\egy)$ is monotone decreasing, non-negative and convex on
$(0,+\infty)$, we have $\lim_{t\to+\infty}\{\phi(t)-t\partial_+\phi(t)\}=\phi(+\infty)$ so that
$\lim_{t\to+\infty}t\partial_+\phi(t)=0$. By Theorem \ref{theorem:V.3},
\[
t\sigma(\{\rho-t\sigma>0\})=-t\partial_+\phi(t)\to0\quad\mbox{as $t\to+\infty$}.
\]
Similarly, $t\sigma(\{\rho-t\sigma\ge0\})=t\partial_-\phi(t)\to0$ as $t\to+\infty$.

\eqref{lim-4}.\enspace
By \eqref{lim-2} and \eqref{lim-3}, $\beta_{\mathrm{mix},t}(\rho\|\sigma)\to\rho(\{\sigma>0\})$ as $t\to+\infty$.

\eqref{lim-5}.\enspace
By \eqref{lim-1} and \eqref{lim-3} we have
\[
\rho(\egy-\{\sigma>0\})
=\lim_{t\to+\infty}[(\rho-t\sigma)(\{\rho-t\sigma>0\})+t\sigma(\{\rho-t\sigma>0\})]
=\lim_{t\to+\infty}(\rho-t\sigma)_+(\egy),
\]
as desired.
\end{proof}

\subsection{Representation from Neyman-Pearson tests}

In this subsection we extend the results in Section \ref{sec:clfdiv from NP} to the general von Neumann
algebra setting. Define for $t\in\bR$,
\begin{align*}
F(t)&:=\beta_{0,t}(\rho\|\sigma):=\rho(\{\rho-t\sigma\le0\}), \\
G(t)&:=\sigma(\egy)-\beta_{1,t}(\rho\|\sigma):=\sigma(\{\rho-t\sigma\le0\}).
\end{align*}
It is immediate to see that $F(t)=G(t)=0$ for all $t\in(-\infty,0)$. So we may confine our consideration to
$F|_{[0,+\infty)}$ and $G|_{[0,+\infty)}$.

\begin{lemma}\label{lemma:V.7}
$F$ and $G$ are monotone increasing and right continuous non-negative functions on $[0,\infty)$ such that
\begin{align}\label{F-G:0infty}
F(0)=0,\quad F(+\infty)=\rho(\{\sigma>0\}),\quad G(0)=\sigma(\{\rho=0\}),\quad G(+\infty)=\sigma(\egy).
\end{align}
\end{lemma}

\begin{proof}
As for $F$, note that for any $t>0$
\[
F(t)=\rho(\{t^{-1}\rho-\sigma\le0\})=\rho(\{\sigma-t^{-1}\rho\ge0\}).
\]
Set $\phi(t):=(\sigma-t\rho)_+(\egy)$ and $\hat\phi(t):=\phi(t^{-1})$ for $t>0$. By
Theorem \ref{theorem:V.3} we have $F(t)=-(\partial_-\phi)(t^{-1})$ for any $t>0$. Since
$\phi$ is convex on $(0,+\infty)$, it follows that $\partial_-\phi$ is increasing on $(0,+\infty)$ so that $F$ is
increasing on $(0,+\infty)$ as well. It is immediate to see that
$(\partial_+\hat\phi)(t)=-{1\over t^2}(\partial_-\phi)(t^{-1})$ so that $F(t)=t^2(\partial_+\hat\phi)(t)$.
This implies that $F$ is right continuous on $(0,+\infty)$. It is obvious that $F(0)=\rho(\{\rho\le0\})=0$. Since
Lemma \ref{lemma:V.5} gives $\rho(\egy)=\rho(\{\rho>0\})\le\liminf_{t\searrow0}\rho(\{\rho-t\sigma>0\})$,
we have
\[
\limsup_{t\searrow0}F(t)=\limsup_{t\searrow0}\bigl[\rho(\egy)-\rho(\{\rho-t\sigma>0\})\bigr]\le0=F(0).
\]
Since $F(t)\ge0$ for all $t\ge0$ obviously, this implies that $F(0)=\lim_{t\searrow0}F(t)$, i.e., $F$ is right
continuous at $0$ as well. Finally, $F(\infty)=\rho(\{\sigma>0\})$ is due to \eqref{lim-2}.

As for $G$, since Theorem \ref{theorem:V.3} implies that
\[
G(t)=\sigma(\egy)-\sigma(\{\rho-t\sigma>0\})
=\sigma(\egy)+\partial_+(\rho-t\sigma)_+(\egy),\qquad t\in\bR,
\]
it follows that $G$ is increasing and right continuous on $[0,\infty)$ (indeed, on $\bR$). Furthermore,
$G(0)=\sigma(\egy)-\sigma(\{\rho>0\})=\sigma(\{\rho=0\})$. By \eqref{lim-3} we have
$\lim_{t\to+\infty}\beta_{1,t}(\rho\|\sigma)=0$, hence $G(+\infty)=\sigma(\egy)$.
\end{proof}

Since the distribution function of $dG|_{(0,+\infty)}$ is
\[
G_0(t):=\begin{cases}0 & (t<0) \\ G(t)-G(0) & (t\ge0),\end{cases}
\]
we consider $G_0$ instead of $G$.

\begin{lemma}\label{lemma:V.8}
There exists a non-negative convex (hence continuous) and monotone increasing function $H$ on $\bR$
such that $H(0)=0$ and
\begin{align}\label{tG_0(t)-F(t)=H(t)}
tG_0(t)-F(t)=H(t)\qquad\mbox{and}\qquad\partial_+H(t)=G_0(t),\qquad t\in\bR.
\end{align}
\end{lemma}

\begin{proof}
Note that for any $t\in\bR$,
\begin{align}
tG(t)-F(t)&=-(\rho-t\sigma)(\{\rho-t\sigma\le0\})=(\rho-t\sigma)_-(\egy) \nonumber\\
&=(\rho-t\sigma)_+(\egy)-(\rho-t\sigma)(\egy)
=(\rho-t\sigma)_+(\egy)-\rho(\egy)+t\sigma(\egy). \label{tG(t)-F(t)}
\end{align}
Since $G(0)=\sigma(\{\rho=0\})$, we have
\[
tG_0(t)-F(t)=(\rho-t\sigma)_+(\egy)-\rho(\egy)+t\sigma(\{\rho>0\}),\qquad t\in\bR.
\]
Define $H(t)$ on $\bR$ by
\[
H(t):=\begin{cases}0 & (t<0), \\ (\rho-t\sigma)_+(\egy)-\rho(\egy)+t\sigma(\{\rho>0\}) & (t\ge0).\end{cases}
\]
For $t<0$, \eqref{tG_0(t)-F(t)=H(t)} is obvious, and $H(0)=0$. For $t\ge0$, by Theorem \ref{theorem:V.3} we have
\begin{align*}
\partial_+H(t)&=\partial_+(\rho-t\sigma)_+(\egy)+\sigma(\{\rho>0\})
=-\sigma(\{\rho-t\sigma>0\})+\sigma(\{\rho>0\}) \\
&=G(t)-\sigma(\egy)+\sigma(\{\rho>0\})=G(t)-\sigma(\{\rho=0\})=G_0(t).
\end{align*}
In particular, $\partial_+H(0)=0$, so that $H$ is convex on $\bR$ and increasing on $[0,+\infty)$.
\end{proof}

\begin{proposition}\label{prop:V.9}
Let $dF,dG$ be the finite positive measures on $[0,+\infty)$ determined by $F,G$ respectively (i.e., the
Lebesgue--Stieltjes measures determined by $F,G$). Then
\[
dF\ll dG\qquad\mbox{and}\qquad{dF\over dG}(x)=x,\quad x\in[0,+\infty).
\]
\end{proposition}

\begin{proof}
It suffices to prove that
\[
dF\ll dG_0\qquad\mbox{and}\qquad{dF\over dG_0}(x)=x,\quad x\in(0,+\infty).
\]
Let $\{t_i\}_{i=1}^N$ where $N\in\{0,1,\dots,\infty\}$ be a countable set of all points in $(0,+\infty)$ at which
at least one of $dF$ or $dG$ has an atom. Then we can write $dF$ and $dG$ as follows:
\[
dF(t):=dF_1(t)+\sum_{i=1}^Na_i\,d\delta_{t_i}(t),\qquad
dG_0(t):=dG_1(t)+\sum_{i=1}^Nb_i\,d\delta_{t_i}(t),\qquad t\in(0,+\infty),
\]
where $F_1$ and $G_1$ are continuous monotone increasing functions with $F_1(0)=G_1(0)=0$, and
$a_i,b_i\ge0$ with $(a_i,b_i)\ne(0,0)$ such that $\sum_{i=1}^Na_i<+\infty$ and $\sum_{i=1}^Nb_i<+\infty$.
We then have
\begin{align}
F(x)=F_1(x)+\sum_{i=1}^N1_{[t_i,+\infty)}(x)a_i,\qquad
G_0(x)=G_1(x)+\sum_{i=1}^N1_{[t_i,+\infty)}(x)b_i,\qquad x\in(0,+\infty). \label{F(x)-G_0(x)}
\end{align}
Since Lemma \ref{lemma:V.8} says that $xG_0(x)-F(x)=H(x)$ is continuous on $(0,+\infty)$, it follows that
the function
\[
\phi(x):=x\sum_{i=1}^N1_{[t_i,+\infty)}(x)b_i-\sum_{i=1}^N1_{[t_i,+\infty)}(x)a_i
=\sum_{i=1}^N1_{[t_i,+\infty)}(x)(b_ix-a_i)
\]
must be continuous on $(0,+\infty)$. For each $t_k$ (with $N>0$ assumed) we find that
\begin{align}\label{phi(t_k)}
\phi(t_k^-)=\sum_{t_i<t_k}(b_it_k-a_i),\qquad
\phi(t_k)=\sum_{t_i\le t_k}(b_it_k-a_i).
\end{align}
Indeed, the above second equality is obvious. For any $\delta>0$, since
\[
\phi(t_k-\delta)=\sum_{t_i\le t_k-\delta}\{b_i(t_k-\delta)-a_i\}
=t_k\sum_{t_i\le t_k-\delta}b_i-\sum_{t_i\le t_k-\delta}a_i-\delta\sum_{t_i\le t_k-\delta}b_i
\]
and as $\delta\searrow0$,
\[
\sum_{t_i\le t_k-\delta}a_i\nearrow\sum_{t_i<t_k}a_i,\qquad
\sum_{t_i\le t_k-\delta}b_i\nearrow\sum_{t_i<t_k}b_i,\qquad
\delta\sum_{t_i\le t_k-\delta}b_i\le\delta\sum_{i=1}^Nb_i\to 0,
\]
we have the first equality in \eqref{phi(t_k)} as well, so that $a_k=b_kt_k>0$ follows for any $k$. Therefore,
\begin{align}
&F(x)=F_1(x)+\sum_{i=1}^N1_{[t_i,+\infty)}(x)b_it_i,\qquad x\in(0,+\infty), \label{F(x)-F_1(x)}\\
&\sum_{i=1}^Na_i\delta_{t_i}=\sum_{i=1}^Nb_it_i\delta_{t_i}. \label{a_i-b_it_i}
\end{align}

Next, since $H(t)$ is convex on $(0,+\infty)$, let $\{s_j\}_{j=1}^K$ where $K\in\{0,1,\dots,\infty\}$ be a
countable set of all points $s\in(0,+\infty)$ for which $\partial_-H(s)<\partial_+H(s)$. Then we can write
\begin{align}\label{partial_+H}
\partial_+H(x)=h(x)+\sum_{j=1}^K1_{[s_j,+\infty)}(x)c_j,\qquad x\in(0,+\infty),
\end{align}
where $h$ is a monotone increasing continuous function on $(0,+\infty)$ and
$c_j:=\partial_+H(s_j)-\partial_-H(s_j)$. Since $G_0(x)=\partial_+H(x)$ for all $x\in(0,+\infty)$ by
Lemma \ref{lemma:V.8}, it follows from \eqref{F(x)-G_0(x)} and \eqref{partial_+H} that
\[
G_1(x)=h(x),\qquad\sum_{i=1}^N1_{[t_i,+\infty)}(x)b_i=\sum_{j=1}^K1_{[s_j,+\infty)}(x)c_j,
\]
so that
\[
\partial_+H(x)=G_1(x)+\sum_{i=1}^N1_{[t_i,+\infty)}(x)b_i,\qquad x\in(0,+\infty),
\]
which implies that
\begin{align}\label{H(x)=}
H(x)=\int_0^xG_1(t)\,dt+\sum_{i=1}^Nb_i(x-t_i)_+,\qquad x\in(0,+\infty).
\end{align}
Moreover, by \eqref{F(x)-G_0(x)} and \eqref{F(x)-F_1(x)},
\begin{align*}
H(x)&=xG_0(x)-F(x)=xG_1(x)-F_1(x)+\sum_{i=1}^N1_{[t_i,+\infty)}(x)(b_ix-b_it_i) \\
&=xG_1(x)-F_1(x)+\sum_{i=1}^Nb_i(x-t_i)_+,\qquad x\in(0,+\infty).
\end{align*}
Combining this with \eqref{H(x)=} yields
$xG_1(x)-F_1(x)=\int_0^xG_1(t)\,dt$ for all $x\in(0,+\infty)$. From the well known integration by part
\[
\int_0^xt\,dG_1(t)=xG_1(x)-\int_0^xG_1(t)\,dt
\]
for the Riemann--Stieltjes integral, we finally have $F_1(x)=\int_0^xt\,dG_1(t)$ for all $x\in(0,+\infty)$, which
implies that
\begin{align}\label{dF_1-dG_1}
dF_1\ll dG_1\qquad\mbox{and}\qquad{dF_1\over dG_1}(x)=x,\quad x\in(0,+\infty).
\end{align}
Hence the assertion follows from \eqref{a_i-b_it_i} and \eqref{dF_1-dG_1}.
\end{proof}

\begin{lemma}\label{lemma:V.10}
For any $t\in[0,+\infty)$ we have
\begin{align}
&D_{(\id-t)_+}(dF\|dG)=D_{(\id-t)_+}(dF|_{(0,+\infty)}\|dG|_{(0,+\infty)})
=D_{(\id-t)_+}^\meas(\rho\|\sigma)-\rho(\{\sigma=0\}), \label{D_+(dF|dG)}\\
&D_{(\id-t)_-}(dF\|dG)=D_{(\id-t)_-}^\meas(\rho\|\sigma), \label{D_-(dF|dG)}\\
&D_{(\id-t)_-}(dF|_{(0,+\infty)}\|dG|_{(0,+\infty)})=D_{(\id-t)_-}^\meas(\rho\|\sigma)-t\sigma(\{\rho=0\}).
\label{D_-(dF|dG)_(0,infty)}
\end{align}
\end{lemma}

\begin{proof}
For any $t\in[0,+\infty)$ we have
\begin{align*}
D_{(\id-t)_+}(dF\|dG)&=\int_{[0,+\infty)}(x-t)_+\,dG(x)=\int_{(t,+\infty)}(x-t)\,dG(x) \\
&=\int_{(t,+\infty)}dF(x)-t\int_{(t,+\infty)}dG(x) \\
&=tG(t)-F(t)+F(+\infty)-tG(+\infty) \\
&=D_{(\id-t)_+}^\meas(\rho\|\sigma)-\rho(\egy)+t\sigma(\egy)+\rho(\{\sigma>0\})-t\sigma(\egy) \\
&=D_{(\id-t)_+}^\meas(\rho\|\sigma)-\rho(\{\sigma=0\}),
\end{align*}
where the first and the third inequalities are due to Proposition \ref{prop:V.9} and the fifth equality follows
from \eqref{tG(t)-F(t)} and \eqref{F-G:0infty}. The computation is similar for
$D_{(\id-t)_+}(dF|_{(0,+\infty)}\|dG|_{(0,+\infty)})$, so that both equalities in $\eqref{D_+(dF|dG)}$ hold.

Next, by \eqref{tG(t)-F(t)} and \eqref{F-G:0infty} again we have
\begin{align*}
D_{(\id-t)_+}(dF|_{(0,+\infty)}\|dG|_{(0,+\infty)})
&=\int_{[0,+\infty)}(x-t)_-\,dG(x)=\int_{[0,t)}(t-x)\,dG(x) \\
&=tG(t)-F(t)=D_{(\id-t)_-}^\meas(\rho\|\sigma), \\
D_{(\id-t)_-}(dF|_{(0,+\infty)}\|dG|_{(0,\infty)})
&=\int_{(0,t)}(t-x)\,dG(x)=t\{G(t)-G(0)\}-F(t) \\
&=D_{\id-t)_-}^\meas(\rho\|\sigma)-t\sigma(\{\rho=0\}),
\end{align*}
as asserted.
\end{proof}

The next theorem is an extension of \cite[Proposition 6.1]{LiuHircheCheng2025} and
Corollary \ref{cor:cl fdiv equality} to the general von Neumann algebra setting.

\begin{theorem}\label{theorem:V.11}
Let $f\in(\bR^{(0,+\infty)})_{\mathrm{conv}}\cup(\bR^{(0,+\infty)})_{\mathrm{conc}}$. Then we have
\[
D_f^{\meas,\hs}(\rho\|\sigma)=D_f(dF\|dG)+f'(+\infty)\rho(\{\sigma=0\}),
\]
where $f'(+\infty):=\lim_{x\to+\infty}f'(x^+)\in[-\infty,+\infty]$.

Moreover, $D_f^{\meas,\hs}(\rho\|\sigma)=D_f(dF\|dG)$ holds if one of the following conditions
hold:
\begin{itemize}
\item[(i)] $\{\rho>0\}\le\{\sigma>0\}$;
\item[(ii)] $f'(+\infty)=0$;
\item[(iii)] $D_f(dF\|dG)=\pm\infty$.
\end{itemize}
\end{theorem}

\begin{proof}
Assume that $f$ is convex. Then $f'(+\infty)\in(-\infty,+\infty]$. By \eqref{D-f-q-hs}, Lemma \ref{lemma:V.10},
\eqref{F-G:0infty} and Corollary \ref{cor:classical fdiv repr0} we have
\begin{align*}
D_f^{\meas,\hs}(\rho\|\sigma)&=f(a)\sigma(\egy)+f'(a^+)(\rho(\egy)-a\sigma(\egy)) \\
&\qquad+\int_{(0,a]}D_{(\id-t)_-}^\meas(\rho\|\sigma)\,df'(t)
+\int_{(a,+\infty)}D_{(\id-t)_+}^\meas(\rho\|\sigma)\,df'(t) \\
&=f(a)\sigma(\egy)+f'(a^+)(\rho(\egy)-a\sigma(\egy)) \\
&\qquad+\int_{(0,a]}D_{(\id-t)_-}(dF\|dG)\,df'(t)
+\int_{(a,+\infty)}D_{(\id-t)_+}(dF\|dG)\,df'(t) \\
&\qquad+(f'(+\infty)-f'(a^+))\rho(\{\sigma=0\}) \\
&=f(a)G(+\infty)+f'(a^+)(F(+\infty)-aG(+\infty)) \\
&\qquad+\int_{(0,a]}D_{(\id-t)_-}(dF\|dG)\,df'(t)
+\int_{(a,+\infty)}D_{(\id-t)_+}(dF\|dG)\,df'(t) \\
&\qquad+f'(+\infty)\rho(\{\sigma=0\}) \\
&=D_f(dF\|dG)+f'(+\infty)\rho(\{\sigma=0\}).
\end{align*}
Hence the first assertion follows. If the condition (i) or (ii) is satisfied, then $f'(+\infty)\rho(\{\sigma=0\})=0$
so that $D_f^{\meas,\hs}(\rho\|\sigma)=D_f(dF\|dG)$ follows. If (iii) is satisfied, then
$D_f^{\meas,\hs}(\rho\|\sigma)$ and $D_f(dF\|dG)$ are $+\infty$. The proof is similar when $f$ is
concave; otherwise, apply the convex case to $-f$.
\end{proof}

\subsection{Joint lower semicontinuity}\label{sec:V.D}

Concerning the joint lower semicontinuity of quantum $f$-divergences in the von Neumann algebra setting,
the following are known:
\begin{itemize}
\item[(a)] When $f$ is operator convex on $(0,+\infty)$, the Petz-type $f$-divergence
$D_f^{\mathrm{P}}(\rho\|\sigma)$ is jointly lower semicontinuous in $\rho,\sigma\in\M_{*,\ge0}$ in the
$\sigma(\M_*,\M)$-topology; see \cite{Hiai_fdiv_standard}.

\item[(b)] When $f$ is as in (a), the maximal $f$-divergence $D_f^{\max}(\rho\|\sigma)$ (or
$\widehat S_f(\rho\|\sigma)$) is jointly lower semicontinuous in $\rho,\sigma\in\M_{*,\ge0}$ in the norm;
see \cite{Hiai_fdiv_maximal}.

\item[(c)] For a general $f\in(\bR^{(0,+\infty)})_{\mathrm{conv}}$ the measured $f$-divergence
$D_f^\meas(\rho\|\sigma)$ is jointly lower semicontinuous in $\rho,\sigma\in\M_{*,\ge0}$ in the
$\sigma(\M_*,\M)$-topology. This is immediate from the definition of $D_f^\meas$ as mentioned in
\cite[Proposition 5.4]{Hiai2021Div}.
\end{itemize}

In this subsection let us discuss the joint lower semicontinuity of the hockey stick $f$-divergences
$D_f^{q,\hs}$. Throughout the subsection we assume that $(D_{(\id-t)_+}^q)_{t\in(0,+\infty)}$ is a
non-negative and monotone family of quantum hockey stick divergences in the sense of
Definitions \ref{def:non-neg-family} and \ref{def:monotone-family}. This is the case, in particular, when
$q\in\{\mathrm{P},\meas,\max\}$.

\begin{proposition}\label{prop:V.12}
Let $f\in(\bR^{(0,+\infty)})_{\mathrm{conv}}$. Assume that $D^q_{(\id-t)_+}(\rho\|\sigma)$ is jointly lower
semicontinuous in $\rho,\sigma\in\M_{*,\ge0}$ in the norm (resp., in the $\sigma(\M_*,\M)$-topology) for
every $t\in(0,+\infty)$. Then so is $D_f^{q,\hs}(\rho\|\sigma)$.
\end{proposition}

\begin{proof}
The assumption also implies the joint lower semicontinuity of $D_{(\id-t)_-}^q(\rho\|\sigma)$ in the norm
(resp., in the $\sigma(\M_*,\M)$-topology). By extracting a countable set $\{s_k\}$ in $(0,+\infty)$ of
discontinuity points of $\partial_+f(t)$ and letting $a_k:=\partial_+f(s_k)-\partial_-f(s_k)>0$, we can write
\[
df'(t)=dF_0(t)+\sum_ka_k\,d\delta_{s_k}(t),\qquad t\in(0,+\infty),
\]
where $F_0$ is a monotone increasing continuous function on $(0,+\infty)$ with $F_0(0^+)=0$. For each
$a\in(0,+\infty)$ we easily see that
\begin{align*}
&\int_{(0,a]}D_{(\id-t)_-}^qu(\rho\|\sigma)\,df'(t) \\
&\qquad=\sup_{K,\Delta}\Biggl[\sum_{k\le K,\,s_k<a}D_{(\id-s_k)_-}^q(\rho\|\sigma)a_k
+\sum_{i=1}^nD_{(\id-t_i)_-}^q(\rho\|\sigma)\{F_0(t_{i+1})-F_0(t_i)\}\Biggr], \\
&\int_{(a,+\infty)}D_{(\id-t)_+}^qu(\rho\|\sigma)\,df'(t) \\
&\qquad=\sup_{K,\Delta'}\Biggl[\sum_{k\le K,\,s_k\ge a}D_{(\id-s_k)_+}^q(\rho\|\sigma)a_k
+\sum_{i=1}^nD_{(\id-t_i)_+}^q(\rho\|\sigma)\{F_0(t_i)-F_0(t_{i-1})\}\Biggr],
\end{align*}
where the above first supremum is taken over all $K\in\bN$ and partitions
$\Delta:0<t_1<\dots<t_n<t_{n+1}=a$, and the second is taken over $K\in\bN$ and partitions
$\Delta':a=t_0<t_1<\dots<t_n<+\infty$. From the assumption, all terms in the sums inside both suprema are
lower semicontinuous in $\rho,\sigma$ in the specified topology. Hence the result follows from the definition
of $D_f^q(\rho\|\sigma)$ in \eqref{D-f-q-hs}. (For the lower semicontinuity in the norm can
more directly follow from Fatou's lemma applied to sequences $\rho_n,\sigma_n$ with $\|\rho_n-\rho\|\to0$
and $\|\sigma_n-\sigma\|\to0$.)
\end{proof}

Since $D_{(\id-t)_+}^\meas$ is jointly lower semicontinuous in $\sigma(\M_*,\M)$ by definition, we have

\begin{cor}\label{cor:V.13}
For every $f\in(\bR^{(0,+\infty)})_{\mathrm{conv}}$, $D_f^{\meas,\hs}(\rho\|\sigma)$ is jointly lower
semicontinuous in $\rho,\sigma\in\M_{*,\ge0}$ in the $\sigma(\M_*,\M)$-topology.
\end{cor}

To discuss the maximal hockey stick divergences, we use the notation
\[
F_t(\rho\|\sigma):=\max\{\eta(\egy):\eta\in\M_*^+,\ 0\le\eta\le\rho,\ \eta\le t\sigma\}
\]
for $t\in(0,+\infty)$ and $\rho,\sigma\in\M_*^+$. By Proposition \ref{prop:V.1} note that
\[
D^{\max}_{(\id-t)_+}(\rho\|\sigma)=\rho(\egy)-F_t(\rho\|\sigma),\qquad
D^{\max}_{(\id-t)_-}(\rho\|\sigma)=t\sigma(\egy)-F_t(\rho\|\sigma).
\]

\begin{lemma}\label{lemma:V.14}
The hockey stick divergence $D_{(\id-t)_+}^{\max}(\rho\|\sigma)$ is jointly lower semicontinuous in
$\rho,\sigma$ in the norm for every $t\in(0,+\infty)$.
\end{lemma}

\begin{proof}
Let $\rho,\sigma,\rho_n,\sigma_n$ for $n\in\bN$ be such that $\|\rho_n-\rho\|\to0$ and
$\|\sigma_n-\sigma\|\to0$. To prove the joint lower semicontinuity, it suffices to show that
$F_t(\rho\|\sigma)\ge\limsup_{n\to\infty}F_t(\rho_n\|\sigma_n)$. To show this, by taking a subsequence we
may assume that $\lim_{n\to\infty}F_t(\rho_n\|\sigma_n)$ exists. For each $n\in\bN$ choose an
$\eta_n\in\M_{*,\ge0}$ such that $0\le\eta_n\le\rho_n$, $\eta_n\le t\sigma_n$ and
$F_t(\rho_n\|\sigma_n)=\eta_n(\egy)$. Define a subset $K$ of $\M_{*,\ge0}$ by
\[
K:=\bigcup_{n=1}^\infty\{\eta\in\M_*^+:\eta\le\rho_n\}.
\]
Let $\{p_j\}$ be a sequence of projections in $\M$ with $p_j\searrow0$. Then for any $\ep>0$
choose an
$n_0\in\bN$ such that $\|\rho_n-\rho\|<\ep/2$ for all $n\ge n_0$. Moreover, choose a $j_0\in\bN$
such that if $j\ge j_0$ then $\rho_n(p_j)<\ep/2$ for $1\le n<n_0$ and $\rho(p_j)<\ep/2$. Now
let $\eta\in K$ and $j\ge j_0$. Then $\eta\le\rho_n$ for some $n$. For $n<n_0$ we have
$\eta(p_j)\le\rho_n(p_j)<\ep/2$. Moreover, for any $n\ge n_0$ we have
\[
\eta(p_j)\le\rho_n(p_j)=(\rho_n-\rho)(p_j)+\rho(p_j)
\le\|\rho_n-\rho\|+{\ep\over2}<\ep.
\]
This implies that $\lim_{j\to\infty}\eta(p_j)=0$ uniformly for $\eta\in K$. Hence $K$ is relatively
$\sigma(\M_*,\M)$-compact due to \cite[Theorem III.5.4]{Takesaki}. So we can choose a subnet
$\{\eta_{n(i)}\}$ of $\{\eta_n\}$ and an $\eta\in\M_*^+$ such that $\eta_{n(i)}\to\eta$ in $\sigma(\M_*,\M)$.
Since $\eta_{n(i)}\le\rho_{n(i)}$ and $\|\rho_{n(i)}-\rho\|\to0$, it follows that $\eta\le\rho$. Also, since
$\eta_{n(i)}\le t\sigma_{n(i)}$ and $\|\sigma_{n(i)}-\sigma\|\to0$, it follows that $\eta\le t\sigma$. Therefore,
\[
F_t(\rho\|\sigma)\ge\eta(\egy)=\lim_i\eta_{n(i)}(\egy)=\lim_iF_t(\rho_{n(i)}\|\sigma_{n(i)})
=\lim_{n\to\infty}F_t(\rho_n\|\sigma_n),
\]
as desired.
\end{proof}

\begin{cor}\label{cor:V.15}
For every $f\in(\bR^{(0,+\infty)})_{\mathrm{conv}}$, $D_f^{\max,\hs}(\rho\|\sigma)$ is jointly lower
semicontinuous in $\rho,\sigma\in\M_{*,\ge0}$ in the norm.
\end{cor}

It seems unknown whether the above stated (a) holds true or not for general convex functions $f$ on
$(0,+\infty)$. By Propositions \ref{prop:V.2} and \ref{prop:V.12} the problem boils down to that of hockey
stick divergence $D_{(\id-t)_+}^{\mathrm{P}}$. This is indeed true in the finite dimensional case, whose
proof is included in Appendix C for the reader's convenience. 

\subsection{Martingale convergence}\label{sec:V.E}

Let $\{\M_\alpha\}$ be an increasing net of unital von Neumann subalgebras of $\M$ such that
$(\bigcup_\alpha\M_\alpha)''=\M$. Let $\rho_\alpha:=\rho|_{\M_\alpha}$ and
$\sigma_\alpha:=\sigma|_{\M_\alpha}$. Roughly speaking, the martingale convergence of quantum
$f$-divergences is an easy consequence from the monotonicity under CPTP maps and the joint lower
semicontinuity in the norm. Therefore, in the rest of this subsection we make the following two assumptions:
\begin{itemize}
\item[(A)] $D^q_{(\id-t)_+}(\rho\|\sigma)$ is jointly lower semicontinuous in $\rho,\sigma\in\M_*^+$ in the
norm for every $t\in(0,+\infty)$.
\item[(B)] $D^q_{(\id-t)_+}(\rho\|\sigma)$ is monotone under TPCP maps for every $t\in(0,+\infty)$.
\end{itemize}

Note that properties (A) and (B) are satisfied for $D_{(\id-t)_+}^\meas$ and $D_{(\id-t)_+}^{\max}$.

\begin{theorem}\label{theorem:V.16}
Assume (A), (B) and that $\M$ is $\sigma$-finite. For every $f\in(\bR^{(0,+\infty)})_{\mathrm{conv}}$, the
martingale convergence holds:
\[
D^{q,\hs}_f(\rho_\alpha\|\sigma_\alpha)\nearrow D^{q,\hs}_f(\rho\|\sigma).
\]
\end{theorem}

\begin{proof}
By (A) and Proposition \ref{prop:V.12}, $D^{q,\hs}_f(\rho\|\sigma)$ is jointly lower semicontinuous in
$\rho,\sigma\in\M_*^+$ in the norm. Also, it is immediate from (B) that $D^{q,\hs}_f(\rho\|\sigma)$ is
monotone under CPTP maps. This implies that $D^{q,\hs}_f(\rho_\alpha\|\sigma_\alpha)$ is increasing in
$\alpha$ and $D^{q,\hs}_f(\rho_\alpha\|\sigma_\alpha)\le D^{q,\hs}_f(\rho\|\sigma)$ for all $\alpha$.
Hence it only remains to show that
\begin{align}\label{ineq:martingale}
D^{q,\hs}_f(\rho\|\sigma)\le\sup_\alpha D^{q,\hs}_f(\rho_\alpha\|\sigma_\alpha).
\end{align}

Since $\M$ is $\sigma$-finite by assumption, there is a faithful $\vfi\in\M_*^+$. Now, let $\E_\alpha$ be the
generalized conditional expectation from $\M$ into $\M_\alpha$ with respect to $\vfi$ due to Accardi and
Cecchini \cite{AccardiCecchini1982} (see also \cite[Sec.\ 5.2]{Hiai2021}), which is a unital completely
positive normal map, so the predual map $(\E_\alpha)_*:(\M_\alpha)_*\to\M_*$ is a CPTP map. Let
$\hat\rho_\alpha:=\rho\circ\E_\alpha=\rho_\alpha\circ\E_\alpha$ and
$\hat\sigma_\alpha:=\sigma\circ\E_\alpha=\sigma_\alpha\circ\E_\alpha$. Then it is known
\cite[Theorem 3]{HiaiTsukada1984} that $\|\hat\rho_\alpha-\rho\|\to0$ and $\|\hat\sigma_\alpha-\sigma\|\to0$.
From the lower semicontinuity in the norm we have
\[
D^{q,\hs}_f(\rho\|\sigma)\le\liminf_\alpha D^{q,\hs}_f(\hat\rho_\alpha\|\hat\sigma_\alpha).
\]
Furthermore, from the monotonicity property we have
$D^{q,\hs}_f(\hat\rho_\alpha\|\hat\sigma_\alpha)\le D^{q,\hs}_f(\rho_\alpha\|\sigma_\alpha)$.
Therefore, \eqref{ineq:martingale} follows.
\end{proof}

\begin{proposition}\label{prop:V.17}
Assume (A), (B) and that $D^q_{(\id-t)_+}(\rho\|\sigma)$ is additive under direct sums for every
$t\in(0,+\infty)$. Let $\{e_\alpha\}$ be any increasing net of projections in $\M$ such that
$e_\alpha\nearrow\egy$. Let $e_\alpha\rho e_\alpha:=\rho|_{e_\alpha\M e_\alpha}$ and
$e_\alpha\sigma e_\alpha:=\sigma|_{e_\alpha\M e_\alpha}$. If either $f\in(\bR^{(0,+\infty)})_{\mathrm{conv}}$
or $f\in(\bR^{(0,+\infty)})_{\mathrm{conc}}$, then
\[
\lim_\alpha D^{q,\hs}_f(e_\alpha\rho e_\alpha\|e_\alpha\sigma e_\alpha)
=D^{q,\hs}_f(\rho\|\sigma).
\]
\end{proposition}

\begin{proof}
We may prove the result when $f\in(\bR^{(0,+\infty)})_{\mathrm{conv}}$ (we may consider $-f$ when
$f\in(\bR^{(0,+\infty)})_{\mathrm{conc}}$). Let $\M_\alpha:=e_\alpha\M e_\alpha+\bC e_\alpha^\perp$ and
$\rho_\alpha:=\rho|_{\M_\alpha}$, $\sigma_\alpha:=\sigma|_{\M_\alpha}$. Clearly, $\{\M_\alpha\}$ is an
increasing von Neumann subalgebras of $\M$ such that $(\bigcup_\alpha\M_\alpha)''=\M$. First we show
that
\begin{align}\label{conv:martingale}
D^{q,\hs}_f(\rho_\alpha\|\sigma_\alpha)\nearrow D^{q,\hs}_f(\rho\|\sigma).
\end{align}
(Of course, this is contained in Theorem \ref{theorem:V.16} when $\M$ is $\sigma$-finite.) Similarly to the
proof of Theorem \ref{theorem:V.16}, it suffices to show \eqref{ineq:martingale}. To show this, instead of the
generalized conditional expectation $\E_\alpha$, for each $\alpha$ we can simply consider
$\E_\alpha:\M\to\M_\alpha$ given by
\[
\E_\alpha(A):=e_\alpha Ae_\alpha+\omega_\alpha(e_\alpha^\perp Ae_\alpha^\perp)e_\alpha^\perp,
\qquad A\in\M,
\]
where $\omega_\alpha$ is an arbitrarily chosen normal state on $e_\alpha^\perp\M e_\alpha^\perp$.
Then $\E_\alpha$ is a unital completely positive normal map, so the predual map
$(\E_\alpha)_*:(\M_\alpha)_*\to\M_*$ is a CPTP map. Since
\[
\rho\circ\E_\alpha=\rho(e_\alpha\cdot e_\alpha)
+\rho(e_\alpha^\perp)\omega_\alpha(e_\alpha^\perp\cdot e_\alpha^\perp)
\]
and $\rho(e_\alpha^\perp)\|\omega_\alpha(e_\alpha^\perp\cdot e_\alpha^\perp)\|\le\rho(e_\alpha^\perp)\to0$,
it follows that $\|\rho\circ\E_\alpha-\rho(e_\alpha\cdot e_\alpha)\|\to0$. To show that
$\|\rho\circ\E_\alpha-\rho\|\to0$, we need to see that $\|\rho-\rho(e_\alpha\cdot e_\alpha)\|\to0$. For this we
may assume as before that $\M$ is standardly represented on Haagerup's $L^2(\M)$ so that
$\M_*\cong L^1(\M)$ by $\psi\leftrightarrow h_\psi$. We then have
\begin{align*}
\|\rho-\rho(e_\alpha\cdot e_\alpha)\|&=\|h_\rho-e_\alpha h_\rho e_\alpha\|_1
\le\|e_\alpha^\perp h_\rho\|_1+\|e_\alpha h_\rho e_\alpha^\perp\|_1 \\
&\le\|e_\alpha^\perp h_\rho\|_1+\|h_\rho e_\alpha^\perp\|_1=2\|h_\rho e_\alpha^\perp\|_1 \\
&\le2\|h_\rho^{1/2}\|_2\|h_\rho^{1/2}e_\alpha^\perp\|_2
=2(\tr\,h_\rho)^{1/2}(\tr\,e_\alpha^\perp h_\rho e_\alpha^\perp)^{1/2} \\
&=2(\tr\,h_\rho)^{1/2}(\tr\,h_{e_\alpha^\perp\rho e_\alpha^\perp})^{1/2}
=2\rho(1)^{1/2}((e_\alpha^\perp\rho e_\alpha^\perp)(1))^{1/2} \\
&=2\rho(1)^{1/2}\rho(e_\alpha^\perp)^{1/2}\to0.
\end{align*}
Hence $\|\rho\circ\E_\alpha-\rho\|\to0$ and similarly $\|\sigma\circ\E_\alpha-\sigma\|\to0$. From the lower
semicontinuity in the norm we have
\[
D^{q,\hs}_f(\rho\|\sigma)\le\liminf_\alpha D^{q,\hs}_f(\rho\circ\E_\alpha\|\sigma\circ\E_\alpha).
\]
Since $\rho\circ\E_\alpha=\rho_\alpha\circ\E_\alpha$ and
$\sigma\circ\E_\alpha=\sigma_\alpha\circ\E_\alpha$, from the monotonicity property we furthermore have
$D^{q,\hs}_f(\rho\circ\E_\alpha\|\sigma\circ\E_\alpha)\le D^{q,\hs}_f(\rho_\alpha\|\sigma_\alpha)$.
Therefore, \eqref{ineq:martingale} follows so that \eqref{conv:martingale} follows.

Next, since $D^q_{(\id-t)_+}$ is additive under direct sums for every $t$, it is immediate to see that the same
holds for $D^{q,\hs}_f$. Therefore, in view of \eqref{hs fdiv order vN} we have
\[
D^{q,\hs}_f(\rho_\alpha\|\sigma_\alpha)
=D^{q,\hs}_f(e_\alpha\rho e_\alpha\|e_\alpha\sigma e_\alpha)
+\rho(e_\alpha^\perp)f\biggl({\rho(e_\alpha^\perp)\over\sigma(e_\alpha^\perp)}\biggr).
\]
Then the remaining proof is the same as that of \cite[Theorem 4.5]{Hiai_fdiv_standard}.
\end{proof}

\begin{corollary}\label{cor:V.18}
Under the same assumptions as in Proposition \ref{prop:V.17}, let $\hil$ be an arbitrary Hilbert space and
$\{E_\alpha\}$ be an increasing net of finite dimensional projections on $\hil$ such that
$E_\alpha\nearrow I$. If $\rho,\sigma$ are density operators on $\hil$ and if either
$f\in(\bR^{(0,+\infty)})_{\mathrm{conv}}$ or $f\in(\bR^{(0,+\infty)})_{\mathrm{conc}}$, then we have
\[
\lim_\alpha D^{q,\hs}_f(E_\alpha\rho E_\alpha\|E_\alpha\sigma E_\alpha)
=D^{q,\hs}_f(\rho\|\sigma).
\]
\end{corollary}

From this martingale type convergence the hockey stick $f$-divergences in the infinite dimensional $\B(\hil)$
case can largely reduce to the finite-dimensional case.

We can extend Frenkel's representation of the relative entropy to the injective (or AFD) von Neumann
algebra case.

\begin{corollary}\label{cor:V.19}
Assume that $\M$ is an injective von Neumann algebra. Then for any $\rho,\sigma\in\M_{*,\ge0}$ the relative
entropy $D(\rho\|\sigma)$ is equal to the measured hockey stick $f$-divergence
$D_f^{\meas,\hs}(\rho\|\sigma)$ for $f(x):=x\log x$ ($x>0$).
\end{corollary}

\begin{proof}
The equality $D(\rho\|\sigma)=D_f^{\meas,\hs}(\rho\|\sigma)$ holds in the finite dimensional case; see
\cite{Frenkel_integral} (also \cite{Jencova2023}, \cite{HircheTomamichel_integral}). Assume that $\M$ is
injective, and let $e$ be the support projection of $\rho+\sigma$. Since $e\M e$ remains to be injective and
both $D(\rho\|\sigma)$ and $D_{x\log x}^{\meas,\hs}(\rho\|\sigma)$ do not change if we define them for
$\rho,\sigma$ as elements in $(e\M e)_{*,\ge0}$. Therefore, we may assume without loss of generality that
$\M$ is $\sigma$-finite. Then we can apply Theorem \ref{theorem:V.16} to $D_{x\log x}^{\meas,\hs}$ and
\cite[Theorem 4.1(v)]{Hiai_fdiv_standard} to $D(\rho\|\sigma)$.
\end{proof}

\begin{remark}\label{remark:V.20}
Remark that Frenkel's representation of the relative entropy has been shown for AFD von Neumann algebras
in a recent paper \cite[Sec.~C]{vLuWi2026} and has been furthermore extended to general von Neumann
algebras in a very recent paper \cite{daSilva2026}.
\end{remark}

\subsection{Sufficiency via measured hockey stick $f$-divergences}\label{sec:V.F}

Since Petz' seminal papers \cite{Petz1986,Petz1988}, the sufficiency (or reversibility) theorems
have been obtained by use of different quantum divergences; see \cite{Hiai2021Div,HiaiMosonyi2017,
Jencova_NCLp,Jencova_NCLpII,HiaiJencova2024} for a few among others. In this subsection we consider
the sufficiency via measured hockey stick $f$-divergences in the von Neumann algebra setting. First, we note
that for any $f\in(\bR^{(0,+\infty)})_{\mathrm{conv}}$, the \emph{DPI} (the \emph{data processing inequality})
\begin{align}\label{DPI vN}
D_f^{\meas,\hs}(\Phi(\rho)\|\Phi(\sigma))\le D_f^{\meas,\hs}(\rho\|\sigma),\qquad\rho,\sigma\in\M_{*,\ge0}
\end{align}
holds for any positive trace-preserving map $\Phi:\M_*\to\N_*$. This follows immediately from
\eqref{mono-meas-vN}.

In the rest of this subsection let $\Phi:\M_*\to\N_*$ be a $2$-positive trace-preserving map between
von Neumann algebras. We say that $\Phi$ is \emph{sufficient} (or \emph{reversible}) for
$\rho,\sigma\in\M_{*,\ge0}$ if there exists a $2$-positive trace-preserving map $\Psi:\N_*\to\M_*$ such that
$\Psi(\Phi(\rho))=\rho$ and $\Psi(\Phi(\sigma))=\sigma$. Following Jen{\u c}ov\'a's approach
in \cite{Jencova2023} we present the following:

\begin{theorem}\label{theorem:V.21}
Assume that $\M$ and $\N$ are injective. Let $\Phi:\M_*\to\N_*$ be as stated above and
$\rho,\sigma\in\M_{*,\ge0}$. Then the following conditions are equivalent:
\begin{itemize}
\item[(i)] $\Phi$ is reversible for $\rho,\sigma$.
\item[(ii)] $D_{(\id-t)_+}^\meas(\Phi(\rho)\|\Phi(\sigma))=D_{(\id-t)_+}^\meas(\rho\|\sigma)$ for all
$t\in(0,+\infty)$.
\item[(iii)] $\|\Phi(\rho)-t\Phi(\sigma)\|=\|\rho-t\sigma\|$ for all $t\in(0,+\infty)$.
\item[(iv)] $D_f^{\meas,\hs}(\Phi(\rho)\|\Phi(\sigma))=D_f^{\meas,\hs}(\rho\|\sigma)$ for all
$f\in(\bR^{(0,+\infty)})_\mathrm{conv}$.
\item[(v)] $D_f^{\meas,\hs}(\Phi(\rho)\|\Phi(\sigma))=D_f^{\meas,\hs}(\rho\|\sigma)<+\infty$ for some
$f\in(\bR^{(0,+\infty)})_\mathrm{conv}$ such that the topological support $\supp(df')$ of $df'$ is $(0,+\infty)$.
\item[(vi)] $D_f^{\meas,\hs}(\Phi(\rho)\|\Phi({\rho+\sigma\over2}))=
D_f^{\meas,\hs}(\rho\|{\rho+\sigma\over2})$ for some $f\in(\bR^{(0,+\infty)})_\mathrm{conv}$ such that
$f(0^+)<+\infty$ and $\supp(df')=(0,+\infty)$.
\item[(vii)] $D(\Phi(\rho)\|\Phi({\rho+\sigma\over2}))=D(\rho\|{\rho+\sigma\over2})$.
\end{itemize}
\end{theorem}

\begin{proof}
First, we give remarks about items (v) and (vi). As for (v), note that there is an
$f\in(\bR^{(0,+\infty)})_{\mathrm{\conv}}$ such that $D_f^{\meas,\hs}(\rho\|\sigma)<+\infty$ for all
$\rho,\sigma\in\M_{*,\ge0}$ and $\supp(df')=\bR$. For instance, consider $f(t):=(1+t)^{-1}$ for $t>0$. Then
$f'(t)=-(1+t)^{-2}$ and $f''(t)=2(1+t)^{-3}$ for $t>0$, so that $df'(t)$ is equivalent to $dt$ on $(0,+\infty)$ and
$D_f^{\meas,\hs}(\rho\|\sigma)<+\infty$ for all $\rho,\sigma$. As for (vi), since
$D_{(\id-t)_-}^\meas(\rho\|{\rho+\sigma\over2})\le t\,{\rho+\sigma\over2}(\egy)$ for all $t>0$ and
$D_{(\id-t)_+}^\meas(\rho\|{\rho+\sigma\over2})=0$ for all $t\ge2$, we automatically
$D_f^{\meas,\hs}(\rho\|{\rho+\sigma\over2})<+\infty$. Indeed, for any $\ep\in(0,1)$, since
$\int_{(\ep,1)}t\,df'(t)=f(\ep)-\ep f'(\ep)-f(1)+f'(1)$, it follows that $\int_{(0,1)}t\,df'(t)<+\infty$ if and only if
$\lim_{\ep\searrow0}\{f(\ep)-\ep f'(\ep)\}<+\infty$, which is equivalent to $f(0^+)<+\infty$.

Since (i) is equivalent to the same for $\rho,{\rho+\sigma\over2}$ instead of $\rho,\sigma$, (i)$\iff$(vii) is well
known; see, e.g., \cite[Theorem 6.19]{Hiai2021Div}. (i)$\implies$(iv)$\implies$(v) and (i)$\implies$(vi) are
obvious by the DPI in \eqref{DPI vN}. Since for any $t\in(0,+\infty)$,
\begin{align}
\|\rho-t\sigma\|&=D_{(\id-t)_+}^\meas(\rho\|\sigma)+D_{(\id-t)_-}^\meas(\rho\|\sigma), \nonumber\\
D_{(\id-t)_-}^\meas(\rho\|\sigma)&=D_{(\id-t)_+}^\meas(\rho\|\sigma)-(\rho-t\sigma)(\egy),
\label{D-meas-+}
\end{align}
(ii)$\iff$(iii) is obvious. Moreover, (ii)$\iff$(iv) is immediately seen from \eqref{D-meas-+} and the definition
\eqref{D-f-q-hs}.

Next, let us prove that (v)$\implies$(ii). With any $a\in(0,+\infty)$ we have
\begin{align*}
D_f^{\meas,\hs}(\rho\|\sigma)
&=f(a)\sigma(\egy)+f'(a^+)(\rho-a\sigma)(\egy)
+\int_{(0,a]}D_{(\id-t)_-}^\meas(\rho\|\sigma)\,df'(t) \\
&\qquad+\int_{(a,+\infty)}D_{(\id-t)_+}^\meas(\rho\|\sigma)\,df'(t)<+\infty, \\
D_f^{\meas,\hs}(\rho\circ\Phi\|\sigma\circ\Phi)
&=f(a)\sigma(\egy)+f'(a^+)(\rho-a\sigma)(\egy)
+\int_{(0,a]}D_{(\id-t)_-}^\meas(\rho\circ\Phi\|\sigma\circ\Phi)\,df'(t) \\
&\qquad+\int_{(a,+\infty)}D_{(\id-t)_+}^\meas(\rho\circ\Phi\|\sigma\circ\Phi)\,df'(t).
\end{align*}
By \eqref{DPI vN} this implies that
\begin{align*}
D_{(\id-t)_-}^\meas(\Phi(\rho)\|\Phi(\sigma))&=D_{(\id-t)_-}^\meas(\rho\|\sigma)
\ \ \mbox{a.e.\ $t\in(0,a]$ with respect to $df'$}, \\
D_{(\id-t)_+}^\meas(\Phi(\rho)\|\Phi(\sigma))&=D_{(\id-t)_+}^\meas(\rho\|\sigma)
\ \ \mbox{a.e.\ $t\in(a,+\infty)$ with respect to $df'$}.
\end{align*}
From \eqref{D-meas-+} and $\supp(df')=\bR$ these imply (ii) since
$t\mapsto D_{(\id-t)_\pm}^\meas(\rho\|\sigma)$ are continuous on $(0,+\infty)$. Hence (v)$\implies$(ii) has
been shown.

Furthermore, since for any $t\in(0,2)$,
\begin{equation}
\begin{aligned}\label{D-meas++}
\Bigl(\rho-t\,{\rho+\sigma\over2}\Bigr)_+
&={2-t\over2}\biggl(\rho-{t\over2-t}\,\sigma\biggr)_+, \\
\Bigl(\rho\circ\Phi-t\,{\rho+\sigma\over2}\circ\Phi\Bigr)_+
&={2-t\over2}\Bigl(\rho\circ\Phi-{t\over2-t}\,\sigma\circ\Phi\Bigr)_+,
\end{aligned}
\end{equation}
it follows that (ii) is equivalent to the same for $\rho,{\rho+\sigma\over2}$. Hence (vi)$\implies$(ii) is similar
to (v)$\implies$(ii).

Finally, let us prove that (ii)$\implies$(vii). Since (ii) implies that (iv) holds for $\rho,{\rho+\sigma\over2}$,
we apply this to $f(x)=x\log x$ to have
\[
D_{x\log x}^{\meas,\hs}\Bigl(\Phi(\rho)\Big\|\Phi\Bigl({\rho+\sigma\over2}\Bigr)\Bigr)
=D_{x\log x}^{\meas,\hs}\Bigl(\rho\Big\|{\rho+\sigma\over2}\Bigr).
\]
Since $\M$ and $\N$ are injective, Corollary \ref{cor:V.19} gives
\begin{align*}
D_{x\log x}^{\meas,\hs}\Bigl(\Phi(\rho)\Big\|\Phi\Bigl({\rho+\sigma\over2}\Bigr)\Bigr)
&=D\Bigl(\Phi(\rho)\Big\|\Phi\Bigl({\rho+\sigma\over2}\Bigr)\Bigr), \\
D_{x\log x}^{\meas,\hs}\Bigl(\rho\Big\|{\rho+\sigma\over2}\Bigr)
&=D\Bigl(\rho\Big\|{\rho+\sigma\over2}\Bigr).
\end{align*}
Hence (ii)$\implies$(vii) has been shown.
\end{proof}

\begin{remark}\label{remark:V.22}
The assumption of $\M$ and $\N$ being injective is necessary only to use Corollary \ref{cor:V.19}.
Therefore, this assumption can be removed due to a recent paper \cite{daSilva2026} mentioned in
Remark \ref{remark:V.20}.
\end{remark}

\section{Regularized version of hockey stick R\'enyi $\alpha$-divergences}\label{sec:VI}

For each $\alpha\in(0,+\infty)\setminus\{1\}$ let $f_\alpha(x):=x^\alpha$ for $x\in(0,+\infty)$. For any
$\rho,\sigma\in\M_{*,\ge0}$ the \emph{measured hockey stick R\'enyi $\alpha$-quantity} is given by
\begin{align}
Q_\alpha^{\meas,\hs}(\rho\|\sigma)&:=D_{f_\alpha}^{\meas,\hs}(\rho\|\sigma) \nonumber\\
&\ =\alpha\rho(\egy)+(1-\alpha)\sigma(\egy) \nonumber\\
&\quad+\int_1^{+\infty}D_{(\id-t)_+}^\meas(\rho\|\sigma)\alpha(\alpha-1)t^{\alpha-2}\,dt
+\int_0^1D_{\id-t)_-}^\meas(\rho\|\sigma)\alpha(\alpha-1)t^{\alpha-2}\,dt \nonumber\\
&\ =\alpha\rho(\egy)+(1-\alpha)\sigma(\egy) \nonumber\\
&\quad+\int_1^{+\infty}\alpha(\alpha-1)t^{\alpha-2}D_{(\id-t)_+}^\meas(\rho\|\sigma)\,dt
+\int_0^1\alpha(\alpha-1)t^{\alpha-1}D_{(\id-t^{-1})_+}^\meas(\sigma\|\rho)\,dt \nonumber\\
&\hskip6cm
\Bigl(\mbox{since $D_{(\id-t)_-}^\meas(\rho\|\sigma)=tD_{(\id-t^{-1})_+}^\meas(\sigma\|\rho)$}\Bigr)
\nonumber\\
&\ =\alpha\rho(\egy)+(1-\alpha)\sigma(\egy) \nonumber\\
&\quad+\alpha(\alpha-1)\int_1^{+\infty}\bigl\{t^{\alpha-2}D_{(\id-t)_+}^\meas(\rho\|\sigma)
+t^{-1-\alpha}D_{(\id-t)_+}^\meas(\sigma\|\rho)\bigr\}\,dt. \label{(VI.162)}
\end{align}
When $\alpha>1$, since $f_\alpha(0^+)=f_\alpha'(0^+)=0$, by letting $a\searrow0$ in
\eqref{eq:interval restriction2} the quantity $Q_\alpha^{\meas,\hs}(\rho\|\sigma)$ can be also expressed as
\begin{align}\label{(VI.163)}
Q_\alpha^{\meas,\hs}(\rho\|\sigma)
=\alpha(\alpha-1)\int_0^{+\infty}t^{\alpha-2}D_{(\id-t)_+}^\meas(\rho\|\sigma)\,dt.
\end{align}

Then the \emph{measured hockey stick R\'enyi $\alpha$-divergences} is given by
\[
D_\alpha^{\meas,\hs}(\rho\|\sigma):={1\over\alpha-1}\log Q_\alpha^{\meas,\hs}(\rho\|\sigma).
\]
Here we consider the \emph{regularized version} of $D_\alpha^{\meas,\hs}(\rho\|\sigma)$, that is,
\[
\overline D_\alpha^{\meas,\hs}(\rho\|\sigma)
:=\lim_{n\to\infty}{1\over n}\,D_\alpha^{\meas,\hs}(\rho^{\otimes n}\|\sigma^{\otimes n})
\]
whenever the limit exists.

\subsection{Haagerup's reduction theorem}\label{sec:VI.A}

Let $\M$ be a general $\sigma$-finite von Neumann algebra. Let $\omega$ be a faithful normal state of $\M$
and $\sigma_t^\omega$ ($t\in\bR$) be the associated modular automorphism group. Consider the discrete
additive group $G:=\bigcup_{n\in\bN}2^{-n}\bZ$ and define $\hat\M:=\M\rtimes_{\sigma^\omega}G$, the
crossed product of $\M$ by the action $\sigma^\omega|_G$. Then the dual weight $\hat\omega$ is a faithful
normal state of $\hat\M$, and we have $\hat\omega=\omega\circ E_\M$, where $E_\M:\hat\M\to\M$ is the
canonical conditional expectation (see, e.g., \cite[Sec.~8.1]{Hiai2021}, also \cite[Sec.~2.5]{FGR2026}).
In this setting, Haagerup's reduction theorem is summarized as follows:

There exists an increasing sequence $\{\M_n\}_{n\ge1}$ of von Neumann subalgebras of $\hat\M$,
containing the unit of $\hat\M$, such that the following hold:
\begin{itemize}
\item[(i)] Each $\M_n$ is finite with a faithful normal tracial state $\tau_n$.
\item[(ii)] $\bigl(\bigcup_{n\ge1}\M_n\bigr)''=\hat\M$.
\item[(iii)] For every $n$ there exists a (unique) faithful normal conditional expectation
$E_{\M_n}:\hat\M\to\M_n$ satisfying
\[
\hat\omega\circ E_{\M_n}=\hat\omega,\qquad
\sigma_t^{\hat\omega}\circ E_{\M_n}=E_{\M_n}\circ\sigma_t^{\hat\omega},\quad t\in\bR.
\]
Moreover, for any $A\in\hat\M$, $E_{\M_n}(A)\to A$ in the $\sigma$-strong topology.
\end{itemize}

Furthermore, for any $\rho,\sigma\in\M_*^+$ define $\hat\rho:=\rho\circ E_\M$ and
$\hat\sigma:=\sigma\circ E_\M$. Then $\hat\rho\circ E_{\M_n}\to\hat\rho$ in the norm.
If either $0<\alpha<1$ and $z\ge\max\{\alpha,1-\alpha\}$, or $\alpha>1$ and $z\ge\max\{\alpha/2,\alpha-1\}$,
then
\[
D_{\alpha,z}(\rho\|\sigma)=D_{\alpha,z}(\hat\rho\|\hat\sigma)
=\lim_{n\to\infty}D_{\alpha,z}(\hat\rho|_{\M_n}\|\hat\sigma|_{\M_n})\quad\mbox{increasingly},
\]
where $D_{\alpha,z}$ is the $\alpha$-$z$-R\'enyi divergence (see \cite[Lemma B.2]{HiaiJencova2024}).
In particular, if $\alpha\in[1/2,+\infty)\setminus\{1\}$, then
\begin{align}\label{(VI.164)}
\widetilde D_\alpha(\rho\|\sigma)=\widetilde D_\alpha(\hat\rho\|\hat\sigma)
=\lim_{n\to\infty}\widetilde D_\alpha(\hat\rho|_{\M_n}\|\hat\sigma|_{\M_n})\quad\mbox{increasingly},
\end{align}
where $\widetilde D_\alpha$ is the sandwiched R\'enyi $\alpha$-divergence. If $\alpha\in(0,2]\setminus\{1\}$,
then
\begin{align}\label{(VI.165)}
D_\alpha^{\mathrm{P}}(\rho\|\sigma)=D_\alpha^{\mathrm{P}}(\hat\rho\|\hat\sigma)
=\lim_{n\to\infty}D_\alpha^{\mathrm{P}}(\hat\rho|_{\M_n}\|\hat\sigma|_{\M_n})\quad\mbox{increasingly},
\end{align}
where $D_\alpha^{\mathrm{P}}$ is the standard (Petz-type) R\'enyi $\alpha$-divergence.

\subsection{Regularized measured R\'enyi $\alpha$-divergences}\label{sec:VI.B}

For any $\alpha\in[1/2,+\infty)\setminus\{1\}$, the measured R\'enyi $\alpha$-divergence and its regularized
version are defined by
\begin{align*}
D_\alpha^\meas(\rho\|\sigma)
&:=\sup\{D_\alpha(\fM(\rho)\|\fM(\sigma)):\mbox{$\fM$ is a measurement in $\M$}\}, \\
\overline D_\alpha^\meas(\rho\|\sigma)
&:=\sup_{n\in\bN}{1\over n}\,D_\alpha^\meas(\rho^{\otimes n}\|\sigma^{\otimes n})
=\lim_{n\to\infty}{1\over n}\,D_\alpha^\meas(\rho^{\otimes n}\|\sigma^{\otimes n}).
\end{align*}

The next lemma was proved in \cite[Theorem 3.7]{MO} in the finite dimensional case and recently
generalized in \cite[Theorem 5.5]{FGR2026} to the semi-finite von Neumann algebra case.

\begin{lemma}\label{lemma:VI.1}
Let $\alpha>1$ and $\rho,\sigma\in\M_{*,\ge0}$ with $s(\rho)\le s(\sigma)$, where $s(\rho)$ is the support
projection of $\rho$. Then
\[
\widetilde D_\alpha(\rho\|\sigma)=\overline D_\alpha^\meas(\rho\|\sigma).
\]
\end{lemma}

\begin{proof}
Since $\overline D_\alpha^\meas(\rho\|\sigma)\le\widetilde D_\alpha(\rho\|\sigma)$ is clear, it suffices to show
the reverse inequality. In the situation of Section \ref{sec:VI.A} let $\hat\rho_n:=\hat\rho|_{\M_n}$ and
$\hat\sigma_n:=\hat\sigma|\M_n$. Since $\M_n$ is finite, from \cite[Theorem 5.5]{FGR2026} we have
\begin{align}\label{(VI.166)}
\widetilde D_\alpha(\hat\rho_n\|\hat\sigma_n)
=D_\alpha^\meas(\hat\rho_n\|\hat\sigma_n)
=\sup_{k\in\bN,\,\fM}{1\over k}\,D_\alpha(\fM(\hat\rho_n^{\otimes k})\|\fM(\hat\sigma_n^{\otimes k})),
\end{align}
where $\fM$ is taken over measurements in $\M_n^{\otimes k}$. Note that measurements with finite outcomes
are sufficient to take the supremum in \eqref{(VI.166)}. Such a measurement $\fM$ in $\M_n^{\otimes k}$ is given
as the predual map of a unital positive map $\Phi:\bC^d\to\M_n^{\otimes k}$ in such a way that
$\fM(\hat\rho_n^{\otimes k})=\hat\rho_n^{\otimes k}\circ\Phi$. Consider the
composition
\[
\Psi:\bC^d\overset{\Phi}\longrightarrow\M_n^{\otimes k}\overset{\iota_n^{\otimes k}}
\longhookrightarrow\hat\M^{\otimes k}\overset{E_\M^{\otimes k}}\longrightarrow\M^{\otimes k},
\]
where $\iota_n$ is the imbedding $\M_n\hookrightarrow\hat\M$. This $\Psi$ is a unital positive map. Since
\begin{align*}
\rho^{\otimes k}\circ\Psi&=\rho^{\otimes k}\circ E_\M^{\otimes k}\circ\iota_n^{\otimes k}\circ\Phi
=(\rho\circ E_\M\circ\iota_n)^{\otimes k}\circ\Phi \\
&=(\hat\rho|_{\M_n})^{\otimes k}\circ\Phi
=\hat\rho_n^{\otimes k}\circ\Phi=\fM(\hat\rho_n^{\otimes k})
\end{align*}
and similarly $\sigma^{\otimes k}\circ\Psi=\fM(\hat\sigma_n^{\otimes k})$, we have
\[
D_\alpha(\fM(\hat\rho_n^{\otimes k})\|\fM(\hat\sigma_n^{\otimes k}))
=D_\alpha(\rho^{\otimes k}\circ\Psi\|\sigma^{\otimes k}\circ\Psi)
\le D_\alpha^\meas(\rho^{\otimes k}\|\sigma^{\otimes k}).
\]
Therefore, by \eqref{(VI.166)}
\begin{align}\label{(VI.167)}
\widetilde D_\alpha(\hat\rho_n\|\hat\sigma_n)
\le\sup_{k\in\bN}{1\over k}\,D_\alpha^\meas(\rho^{\otimes k}\|\sigma^{\otimes k})
=\overline D_\alpha^\meas(\rho\|\sigma).
\end{align}
By the martingale convergence we have
$\widetilde D_\alpha(\hat\rho_n\|\hat\sigma_n)\nearrow\widetilde D_\alpha(\hat\rho\|\hat\sigma)$ as $n\to\infty$.
From this and \eqref{(VI.167)} together with \eqref{(VI.164)} we have
$\widetilde D_\alpha(\rho\|\sigma)\le\overline D_\alpha^\meas(\rho\|\sigma)$, as desired.
\end{proof}

\subsection{$\overline D_\alpha^{\meas,\hs}(\rho\|\sigma)=\widetilde D_\alpha(\rho\|\sigma)$ for
$\alpha\in(1,+\infty)$}\label{sec:VI.C}

The next result is the $\alpha>1$ case of \cite[Theorem 3.2]{HircheTomamichel_integral}. The proof below is
similar to that in \cite{HircheTomamichel_integral}.

\begin{theorem}\label{theorem:VI.2}
Let $\alpha>1$ and $\rho,\sigma\in\M_{*,\ge0}$. Assume that either $\widetilde D_\beta(\rho\|\sigma)<+\infty$
for some $\beta\in(\alpha,+\infty)$ or $\widetilde D_\alpha(\rho\|\sigma)=+\infty$. Then
$\overline D_\alpha^{\meas,\hs}(\rho\|\sigma)=
\lim_{n\to\infty}{1\over n}\,D_\alpha^{\meas,\hs}(\rho^{\otimes n}\|\sigma^{\otimes n})$ exists and
\[
\overline D_\alpha^{\meas,\hs}(\rho\|\sigma)=\widetilde D_\alpha(\rho\|\sigma).
\]
\end{theorem}

\begin{proof}
We need to prove that
\begin{align}
\liminf_{n\to\infty}{1\over n}\,D_\alpha^{\meas,\hs}(\rho^{\otimes n}\|\sigma^{\otimes n})
&\ge\widetilde D_\alpha(\rho\|\sigma), \label{(VI.168)}\\
\limsup_{n\to\infty}{1\over n}\,D_\alpha^{\meas,\hs}(\rho^{\otimes n}\|\sigma^{\otimes n})
&\le\widetilde D_\alpha(\rho\|\sigma). \label{(VI.169)}
\end{align}

To prove \eqref{(VI.168)}, the assumption in the theorem is unnecessary. For each $n\in\bN$ and any
measurement $\fM$ in $\M^{\otimes n}$, from the DPI of $D_{(\id-t)_+}^\meas$ it follows that
\[
D_{(\id-t)_+}^\meas(\rho^{\otimes n}\|\sigma^{\otimes n})
\ge D_{(\id-t)_+}^\meas(\fM(\rho^{\otimes n})\|\fM(\sigma^{\otimes n}))
=D_{(\id-t)_+}(\fM(\rho^{\otimes n})\|\fM(\sigma^{\otimes n}))
\]
for all $t\in(0,+\infty)$. Therefore, by the expression \eqref{(VI.163)},
\begin{align*}
Q_\alpha^{\meas,\hs}(\rho^{\otimes n}\|\sigma^{\otimes n})
&\ge\alpha(\alpha-1)\int_0^{+\infty}t^{\alpha-2}D_{(\id-t)_+}(\fM(\rho^{\otimes n})\|\fM(\sigma^{\otimes n}))\,dt \\
&=Q_\alpha(\fM(\rho^{\otimes n})\|\fM(\sigma^{\otimes n}))
\end{align*}
so that
\[
D_\alpha^{\meas,\hs}(\rho^{\otimes n}\|\sigma^{\otimes n})
\ge D_\alpha(\fM(\rho^{\otimes n})\|\fM(\sigma^{\otimes n}))
\]
for all measurement $\fM$ in $\M^{\otimes n}$. This gives that
\[
\liminf_{n\to\infty}{1\over n}\,D_\alpha^{\meas,\hs}(\rho^{\otimes n}\|\sigma^{\otimes n})
\ge\lim_{n\to\infty}{1\over n}\,D_\alpha^\meas(\rho^{\otimes n}\|\sigma^{\otimes n})
\]
so that \eqref{(VI.168)} holds thanks to Lemma \ref{lemma:VI.1}.

Next prove \eqref{(VI.169)}. As it is trivial when $\widetilde D_\alpha(\rho\|\sigma)=+\infty$, assume that
$\widetilde D_\beta(\rho\|\sigma)<+\infty$ for some $\beta\in(\alpha,+\infty)$. For simplicity write
$\rho_n:=\rho^{\otimes n}$, $\sigma_n:=\sigma^{\otimes n}$ and $P_{t,n}:=\{\rho_n-t\sigma_n>0\}$ for any
$n\in\bN$ and $t\in(0,+\infty)$. Since
\begin{align}
\rho_n(P_{t,n})&=\rho_n(P_{t,n})^\alpha\rho_n(P_{t,n})^{1-\alpha}
\le\rho_n(P_{t,n})^\alpha\{t\sigma_n(P_{t,n})\}^{1-\alpha} \nonumber\\
&=t^{1-\alpha}\bigl\{\rho_n(P_{t,n})^\alpha\sigma_n(P_{t,n})^{1-\alpha}
+\rho_n(\egy-P_{t,n})^\alpha\sigma_n(\egy-P_{t,n})^{1-\alpha}\bigr\} \label{(VI.170)}\\
&\le t^{1-\alpha}\widetilde Q_\alpha(\rho_n\|\sigma_n), \nonumber
\end{align}
we have
\begin{align}\label{(VI.171)}
D_{(\id-t)_+}^\meas(\rho_n\|\sigma_n)
=(\rho_n-t\sigma_n)(P_{t,n})\le\rho_n(P_{t,n})\le t^{1-\alpha}\widetilde Q_\alpha(\rho_n\|\sigma_n).
\end{align}
Let $1<\alpha'<\alpha<\alpha''<\beta$. Using \eqref{(VI.163)} we estimate
\begin{align}
Q_\alpha^{\meas,\hs}(\rho_n\|\sigma_n)
&=\alpha(\alpha-1)\int_0^{+\infty}t^{\alpha-2}D_{(\id-t)_+}^\meas(\rho\|\sigma)\,dt \nonumber\\
&\le\alpha(\alpha-1)\biggl[\int_0^1t^{\alpha-2}t^{1-\alpha'}\widetilde Q_{\alpha'}(\rho_n\|\sigma_n)\,dt
+\int_1^{+\infty}t^{\alpha-2}t^{1-\alpha''}\widetilde Q_{\alpha''}(\rho_n\|\sigma_n)\,dt\biggr] \nonumber\\
&=\alpha(\alpha-1)\biggr[\biggl(\int_0^1t^{\alpha-\alpha'-1}\,dt\biggr)\widetilde Q_{\alpha'}(\rho_n\|\sigma_n)
+\biggl(\int_1^{+\infty}t^{\alpha-\alpha''-1}\,dt\biggr)\widetilde Q_{\alpha''}(\rho_n\|\sigma_n)\biggr]
\nonumber\\
&=\alpha(\alpha-1)\biggl[{1\over\alpha-\alpha'}\widetilde Q_{\alpha'}(\rho\|\sigma)^n
+{1\over\alpha''-\alpha}\widetilde Q_{\alpha''}(\rho\|\sigma)^n\biggr] \nonumber\\
&\le\alpha(\alpha-1)\biggl({1\over\alpha-\alpha'}+{1\over\alpha''-\alpha}\biggr)
\Bigl[\max\bigl\{\widetilde Q_{\alpha'}(\rho\|\sigma),\widetilde Q_{\alpha''}(\rho\|\sigma)\bigr\}\Bigr]^n.
\label{(VI.172)}
\end{align}
Therefore, we have
\begin{align*}
\limsup_{n\to\infty}{1\over n}\,D_\alpha^{\meas,\hs}(\rho_n\|\sigma_n)
&\le{1\over\alpha-1}\log\Bigl[\max\bigl\{\widetilde Q_{\alpha'}(\rho\|\sigma),
\widetilde Q_{\alpha''}(\rho\|\sigma)\bigr\}\Bigr] \\
&=\max\bigl\{\widetilde D_{\alpha'}(\rho\|\sigma),\widetilde D_{\alpha''}(\rho\|\sigma)\bigr\}.
\end{align*}
Now recall (\cite[Lemma 8]{BST}, \cite[Proposition 3.7]{Jencova_NCLp}) that
$\alpha\mapsto\widetilde D_\alpha(\rho\|\sigma)$ is monotone increasing and convex on $(1,+\infty)$ (since
$\alpha\mapsto\widetilde Q_\alpha(\rho\|\sigma)$ is log-convex). From the assumption that
$\widetilde D_\beta(\rho\|\sigma)<+\infty$, it follows that $\alpha\mapsto\widetilde D_\alpha(\rho\|\sigma)$ is
continuous on $(1,\beta)$. Thus letting $\alpha'\nearrow\alpha$ and $\alpha''\searrow\alpha$ yields \eqref{(VI.169)}.
\end{proof}

\begin{remark}\label{remark:VI.3}
The second part of the proof of Theorem \ref{theorem:VI.2} does not work when
$\widetilde D_\alpha(\rho\|\sigma)<+\infty$ and $\widetilde D_\beta(\rho\|\sigma)=+\infty$ for all $\beta>\alpha$.
It is written in \cite[Lemma 8]{BST} that the function $\alpha\mapsto\widetilde D_\alpha(\rho\|\sigma)$ is
continuous and monotonically increasing on $(1,+\infty)$ since it is convex on $(1,+\infty)$. But convexity does
not imply continuity when the function takes the $+\infty$ value.
\end{remark}

The test-measured and the projective test-measured R\'enyi $\alpha$-divergences and their regularized
versions are defined \cite{HiaiMosonyi_testmeasured,sc_vN} by
\begin{align*}
D_\alpha^{\mathrm{test}}(\rho\|\sigma)
&:=\sup\{D_\alpha((\rho(T),\rho(\egy-T))\|(\sigma(T),\sigma(\egy-T))):T\in\M,\,0\le T\le\egy\}, \\
\overline D_\alpha^{\mathrm{test}}(\rho\|\sigma)
&:=\sup_{n\in\bN}{1\over n}\,D_\alpha^{\mathrm{test}}(\rho^{\otimes n}\|\sigma^{\otimes n}), \\
D_\alpha^\proj(\rho\|\sigma)
&:=\sup\{D_\alpha((\rho(P),\rho(P^\perp))\|(\sigma(P),\sigma(P^\perp))):\mbox{$P$ is a projection in $\M$}\}, \\
\overline D_\alpha^\proj(\rho\|\sigma)
&:=\sup_{n\in\bN}{1\over n}\,D_\alpha^\proj(\rho^{\otimes n}\|\sigma^{\otimes n}).
\end{align*}
For $\rho,\sigma\in\B(\hil)_{\ge0}$ with $\rho^0\le\sigma^0$ in the finite dimensional case, it was shown in
\cite[Corollary 4.6]{MO} that for $\alpha>1$,
\[
\widetilde D_\alpha(\rho\|\sigma)
=\lim_{n\to\infty}{1\over n}\,D_\alpha^\proj(\rho^{\otimes n}\|\sigma^{\otimes n}),
\]
more strongly than $\overline D_\alpha^\proj(\rho\|\sigma)=\widetilde D_\alpha(\rho\|\sigma)$. When $\M$ is
an injective von Neumann algebra, it was shown in \cite[Proposition 3.13]{sc_vN} that for $\alpha>1$,
\[
\widetilde D_\alpha(\rho\|\sigma)
=\lim_{n\to\infty}{1\over n}\,D_\alpha^{\mathrm{test}}(\rho^{\otimes n}\|\sigma^{\otimes n}),
\]
more strongly than $\overline D_\alpha^{\mathrm{test}}(\rho\|\sigma)=\widetilde D_\alpha(\rho\|\sigma)$.
The next corollary extends $\overline D_\alpha^\proj(\rho\|\sigma)=\widetilde D_\alpha(\rho\|\sigma)$ to the
general von Neumann algebra case.

\begin{cor}\label{cor:VI.4}
For any $\alpha>1$ and $\rho,\sigma\in\M_{*,\ge0}$,
\[
\overline D_\alpha^\proj(\rho\|\sigma)=\overline D_\alpha^{\mathrm{test}}(\rho\|\sigma)
=\overline D_\alpha^\meas(\rho\|\sigma)=\widetilde D_\alpha(\rho\|\sigma).
\]
\end{cor}

\begin{proof}
From \eqref{(VI.170)}--\eqref{(VI.172)} in the above proof, we indeed have
\[
D_\alpha^{\meas,\hs}(\rho_n\|\sigma_n)
\le\alpha(\alpha-1)\biggl({1\over\alpha-\alpha'}+{1\over\alpha''-\alpha}\biggr)
\max\bigl\{Q_{\alpha'}^\proj(\rho_n\|\sigma_n),Q_{\alpha''}^\proj(\rho_n\|\sigma_n)\bigr\}.
\]
This implies that
\begin{align}\label{(VI.173)}
\limsup_{n\to\infty}{1\over n}\,D_\alpha^{\meas,\hs}(\rho_n\|\sigma_n)
\le\max\bigl\{\overline D_{\alpha'}^\proj(\rho\|\sigma),\overline D_{\alpha''}^\proj(\rho\|\sigma)\bigr\}.
\end{align}
Now assume that $\rho\le c\sigma$ for some $c>0$. Then $\widetilde D_\alpha(\rho\|\sigma)<+\infty$ for all
$\alpha\in(1,+\infty)$. Hence the LHS of \eqref{(VI.173)} is equal to $\widetilde D_\alpha(\rho\|\sigma)$. Since
$\alpha\mapsto\overline D_\alpha^\proj(\rho\|\sigma)$ is convex and hence continuous on $(1,+\infty)$, the
RHS of \eqref{(VI.173)} goes to $\overline D_\alpha^\proj(\rho\|\sigma)$ as $\alpha'\nearrow\alpha$ and
$\alpha''\searrow\alpha$. Therefore, $\widetilde D_\alpha(\rho\|\sigma)\le\overline D_\alpha^\proj(\rho\|\sigma)$.
For general $\rho,\sigma$ note (see \cite[Lemma 3.13]{Hiai2021Div}) that
\[
\lim_{\ep\searrow0}\widetilde D_\alpha(\rho\|\sigma+\ep\rho)=\widetilde D_\alpha(\rho\|\sigma).
\]
Moreover,
\begin{align*}
\lim_{\ep\searrow0}\overline D_\alpha^\proj(\rho\|\sigma+\ep\rho)
&=\sup_{\ep>0}\sup_{n\in\bN}{1\over n}\,D_\alpha^\proj(\rho^{\otimes n}\|(\sigma+\ep\rho)^{\otimes n}) \\
&=\sup_{n\in\bN}\sup_{\ep>0}{1\over n}\,D_\alpha^\proj(\rho^{\otimes n}\|(\sigma+\ep\rho)^{\otimes n}) \\
&=\sup_{n\in\bN}{1\over n}\,D_\alpha^\proj(\rho^{\otimes n}\|\sigma^{\otimes n})
=\overline D_\alpha^\proj(\rho\|\sigma).
\end{align*}
Hence $\widetilde D_\alpha(\rho\|\sigma)\le\overline D_\alpha^\proj(\rho\|\sigma)$.
\end{proof}

\subsection{$\overline D_\alpha^{\meas,\hs}(\rho\|\sigma)=D_\alpha(\rho\|\sigma)$ for $\alpha\in(0,1)$}
\label{sec:VI.D}

The next lemma is from \cite[Proposition 1.1]{Ogata2011}.

\begin{lemma}\label{lemma:VI.5}
For any $\rho,\sigma\in\M_{+,\ge0}$ and any $0\le s\le1$ we have
\[
\|\Delta_{\rho,\sigma}^{s/2}h_\rho^{1/2}\|^2
\ge{1\over2}(\rho+\sigma)(\egy)-{1\over2}\|\rho-\sigma\|.
\]
\end{lemma}

Since $\Delta_{\rho,\sigma}^{s/2}h_\rho^{1/2}=h_\rho^{s\over2}h_\sigma^{1-s\over2}$, we have
$\|\Delta_{\rho,\sigma}^{s/2}h_\rho^{1/2}\|^2=\tr(h_\rho^sh_\sigma^{1-s})=Q_s(\rho\|\sigma)$. Hence the above
lemma gives
\begin{align}
D_{(\id-t)_+}^\meas(\rho\|\sigma)
&=(\rho-t\sigma)_+(\egy)={1\over2}\{\|\rho-t\sigma\|+(\rho-t\sigma)(\egy)\} \nonumber\\
&\ge\rho(\egy)-t^{1-s}Q_s(\rho\|\sigma). \label{(VI.174)}
\end{align}

\begin{lemma}\label{lemma:VI.6}
For any $\rho,\sigma\in\M_{*,\ge0}$,
\[
\liminf_{n\to\infty}{1\over n}D_\alpha^{\meas,\hs}(\rho^{\otimes n}\|\sigma^{\otimes n})
\ge D_\alpha^{\mathrm{P}}(\rho\|\sigma).
\]
\end{lemma}

\begin{proof}
By \eqref{(VI.162)} we have
\begin{align*}
&Q_\alpha^{\meas,\hs}(\rho\|\sigma) \\
&\quad\le\alpha\rho(\egy)+(1-\alpha)\sigma(\egy)
+\alpha(\alpha-1)\int_1^{+\infty}\{t^{\alpha-2}\rho(\egy)+t^{-1-\alpha}\sigma(\egy)\}\,dt \\
&\qquad+\alpha(1-\alpha)\int_1^{+\infty}
\bigl\{t^{\alpha-1-s}Q_s(\rho\|\sigma)+t^{-\alpha-s'}Q_{s'}(\sigma\|\rho)\bigr\}\,dt \\
&\quad=\alpha\rho(\egy)+(1-\alpha)\sigma(\egy)
+\alpha(\alpha-1)\Bigl\{{\rho(\egy)\over1-\alpha}+{\sigma(\egy)\over\alpha}\Bigr\} \\
&\qquad+\alpha(1-\alpha)\int_1^{+\infty}
\bigl\{t^{\alpha-1-s}Q_s(\rho\|\sigma)+t^{-\alpha-s'}Q_{1-s'}(\rho\|\sigma)\bigr\}\,dt \\
&\quad=\alpha(1-\alpha)\int_1^{+\infty}
\bigl\{t^{\alpha-1-s}Q_s(\rho\|\sigma)+t^{-\alpha-s'}Q_{1-s'}(\rho\|\sigma)\bigr\}\,dt,
\end{align*}
where the inequality above is due to \eqref{(VI.174)}, and $s,s'$ can be arbitrary in $(0,1)$. Now let
$s:=\alpha+\ep$ and $s':=1-\alpha+\ep$, where $0<\ep<\min\{\alpha,1-\alpha\}$. Then the above inequality
becomes
\[
Q_\alpha^\meas(\rho\|\sigma)\le{\alpha(1-\alpha)\over\ep}
\{Q_{\alpha+\ep}(\rho|\sigma)+Q_{\alpha-\ep}(\rho\|\sigma)\}.
\]
Replacing $\rho,\sigma$ with $\rho^{\otimes n},\sigma^{\otimes n}$ for any $n\in\bN$ we have
\begin{align*}
\liminf_{n\to\infty}{1\over n}D_\alpha^{\meas,\hs}(\rho^{\otimes n},\sigma^{\otimes n})
&\ge\liminf_{n\to\infty}{1\over n}{1\over\alpha-1}\log
\bigl\{Q_{\alpha+\ep}(\rho^{\otimes n}|\sigma^{\otimes n})
+Q_{\alpha-\ep}(\rho^{\otimes n}\|\sigma^{\otimes n})\bigr\} \\
&\ge\liminf_{n\to\infty}{1\over n}{1\over\alpha-1}\log
\Bigl[2\bigl(\max\bigl\{Q_{\alpha+\ep}(\rho\|\sigma),Q_{\alpha-\ep}(\rho\|\sigma)\bigr\}\bigr)^n\Bigr] \\
&={1\over\alpha-1}\log\max\{Q_{\alpha+\ep}(\rho\|\sigma),Q_{\alpha-\ep}(\rho\|\sigma)\}.
\end{align*}
Since $\alpha\mapsto Q_\alpha(\rho\|\sigma)$ is continuous on $(0,1)$, the assertion follows by letting
$\ep\searrow0$.
\end{proof}

Next, to prove the opposite inequality
\begin{align}\label{(VI.175)}
\limsup_{n\to\infty}{1\over n}D_\alpha^{\meas,\hs}(\rho^{\otimes n}\|\sigma^{\otimes n})
\le D_\alpha(\rho\|\sigma),
\end{align}
we assume first that $\M$ is a finite von Neumann algebra with a faithful normal finite trace $\tau$. Let
$L^1(\M,\tau)$ and $L^2(\M,\tau)$ be the $L^1$- and $L^2$-spaces with respect to $\tau$. Then the standard
form of $\M$ is given as $(\M,L^2(\M,\tau),J=\,^*,L^2(\M,\tau)_+)$, and $\M_{*,\ge0}$ is identified with
$L^1(\M,\tau)_+$. Moreover, note that the space $\widetilde\M$ of $\tau$-measurable operators is the set of
all densely defined closed operators affiliated with $\M$. See \cite{Hiai2021} for $\tau$-measurable operators
and the $L^p$-spaces with respect to $\tau$.

Let $\rho,\sigma\in L^1(\M,\tau)_+$, whose spectral decompositions are written as
\[
\rho=\int_{[0,+\infty)}x\,dE(x),\qquad\sigma=\int_{[0,+\infty)}y\,dF(y).
\]
The left multiplication $L_\rho$ by $\rho$ and the right multiplication $R_\sigma$ by $\sigma$ are positive
self-adjoint operators on $L^2(\M,\tau)$, whose spectral decompositions are
\[
L_\rho=\int_{[0,+\infty)}x\,dL_{E(x)},\qquad R_\sigma=\int_{[0,+\infty)}y\,dR_{F(y)},
\]
with the commuting spectral measures $L_{E(\cdot)}$ and $R_{F(\cdot)}$ on the Borel space
$\bigl([0,+\infty),\B_{[0,+\infty)}\bigr)$. Then we have the product spectral measure $G=E\times F$ on the
Borel space $\bigl([0,+\infty)^2,\B_{[0,+\infty)^2}\bigr)$ such that
\[
G(S\times T)=L_{E(S)}R_{F(T)}\ \,(\mbox{projection on $L^2(\M,\tau)$}),\qquad S,T\in\B_{[0,+\infty)}.
\]
Define a probability measure $\mu$ on $\bigl([0,+\infty)^2,\B_{[0,+\infty)^2}\bigr)$ by
$\mu(\cdot):=\tau(G(\cdot)\egy)$; then we have a probability measure
$\bigl([0,+\infty)^2,\B_{[0,+\infty)^2},\mu\bigr)$ so that
\begin{align}\label{(VI.176)}
\mu(S\times T)=\tau(E(S)F(T))=\|E(S)F(T)\|_2^2=\|G(S\times T)\egy\|_2^2,\qquad S,T\in\B_{[0,+\infty)}.
\end{align}
Let $p(x,y):=x$ and $q(x,y):=y$ on $[0,+\infty)^2$; then $p\,d\mu$ and $q\,d\mu$ are regarded as the
counterparts of the Nussbaum--Szko\l a distributions for $\rho,\sigma$ in the finite von Neumann algebra setting.
In fact, we have
\begin{align*}
\int_{[0,+\infty)^2}p\,d\mu&=\int_{[0,\infty)^2}x\,d\tau(E(x)F(y))
=\int_{[0,+\infty)}x\,d\tau(E(x))=\tau(\rho), \\
\int_{[0,+\infty)^2}q\,d\mu&=\int_{[0,\infty)^2}y\,d\tau(E(x)F(y))
=\int_{[0,+\infty)}y\,d\tau(F(y))=\tau(\sigma),
\end{align*}
and we have the following:

\begin{lemma}\label{lemma:VI.7}
For every $f\in(\bR^{(0,+\infty)})_{\conv}$,
\[
D_f^{\mathrm{P}}(\rho\|\sigma)=D_f(p\,d\mu\|q\,d\mu).
\]
\end{lemma}

\begin{proof}
As easily shown by, e.g., \cite[Theorem 10.7]{Hiai2021}, note that the modular operator $\Delta_{\rho,\sigma}$
is $L_\rho R_{\sigma^{-1}}$ in the present situation, where $\sigma^{-1}$ is the generalized inverse of $\sigma$.
Hence we have
\[
\Delta_{\rho,\sigma}=L_\rho R_{\sigma^{-1}}=\int_{\{x\ge0,y>0\}}{x\over y}\,dL_{E(x)}R_{F(y)}
=\int_{\{x\ge0,y>0\}}{x\over y}\,dG(x,y).
\]
This implies that the spectral resolution of $\Delta_{\rho,\sigma}$ is given by
\[
E_{\rho,\sigma}(t)=G(R_t),\quad\mbox{where}\quad
R_t:=\Bigl\{(x,y)\in[0,+\infty)\times(0,+\infty):{x\over y}\le t\Bigr\},\ \ t\ge0.
\]
In other word, the spectral measure $E_{\rho,\sigma}$ on $\bigl([0,+\infty),\B_{[0,+\infty)}\bigr)$ is the
push-forward of $G$ restricted on $[0,+\infty)\times(0,+\infty)$ by the map $(x,y)\mapsto{x\over y}$. Hence one
can easily see that
\begin{align*}
\int_{\{t>0\}}f(t)\,d\|E_{\rho,\sigma}(t)\sigma^{1/2}\|_2^2
&=\int_{\{x>0,y>0\}}f\Bigl({x\over y}\Bigr)\,d\|G(x,y)\sigma^{1/2}\|_2^2 \\
&=\int_{\{x>0,y>0\}}f\Bigl({x\over y}\Bigr)\,d\langle\sigma^{1/2},E(x)\sigma^{1/2}F(y)\rangle_2 \\
&=\int_{\{x>0,y>0\}}f\Bigl({x\over y}\Bigr)\,d\tau(E(x)\sigma F(y)) \\
&=\int_{\{x>0,y>0\}}yf\Bigl({x\over y}\Bigr)\,d\mu(x,y).
\end{align*}
Moreover, since
\begin{align*}
\int_{\{p=0\}}q\,d\mu=\int_{\{x=0\}}y\,d\tau(E(x)F(y))=\tau(\sigma\{\rho=0\}), \\
\int_{\{q=0\}}p\,d\mu=\int_{\{y=0\}}x\,d\tau(E(x)F(y))=\tau(\rho\{\sigma=0\}),
\end{align*}
we have
\begin{align*}
D_f^{\mathrm{P}}(\rho\|\sigma)&=\int_{\{t>0\}}f(t)\,d\|E_{\rho,\sigma}(t)\sigma^{1/2}\|_2^2
+f(0^+)\tau(\sigma\{\rho=0\})+f'(+\infty)\tau(\rho\{\sigma=0\}) \\
&=\int_{\{p>0,q>0\}}qf\Bigl({p\over q}\Bigr)\,d\mu
+f(0^+)\int_{\{p=0\}}q\,d\mu+f'(+\infty)\int_{\{q=0\}}p\,d\mu \\
&=D_f(p\,d\mu\|q\,d\mu),
\end{align*}
as asserted.
\end{proof}

In particular, for any $\alpha\in(0,1)$ we have
\[
D_\alpha(\rho\|\sigma)=D_\alpha(p\,d\mu\|q\,d\mu),
\]
which will be used below. Note that this can be shown in a simper way as follows: Since
$f_\alpha(0^+)=f'(+\infty)=0$ for $f_\alpha(x):=x^\alpha$,
\begin{align}
Q_\alpha(p\,d\mu\|q\,d\mu)&=\int_{\{p,q>0\}}p^\alpha q^{1-\alpha}\,d\mu
=\int_{\{x,y>0\}}x^\alpha y^{1-\alpha}\,d\tau(E(x)F(y)) \nonumber\\
&=\tau\biggl(\int_{(0,+\infty)}x^\alpha\,dE(x)\int_{(0,+\infty)}y^{1-\alpha}\,dF(y)\biggr) \nonumber\\
&=\tau(\rho^\alpha\sigma^{1-\alpha})=Q_\alpha(\rho\|\sigma). \label{(VI.177)}
\end{align}

\begin{remark}\label{remark:VI.8}
In a recent paper \cite{AnAn2026} the Nussbaum--Szko\l a distributions have been defined in the semi-finite von
Neumann algebra setting and the equality of the quantum $f$-divergence with the classical $f$-divergence for
those classical distributions has been shown. Our discussion in the finite von Neumann algebra setting is
considerably simpler than that in \cite{AnAn2026}.
\end{remark}

The next lemma is the counterpart of the main ingredient in the proof of \cite{NSz} and of
\cite[Proposition 2]{ANSzV} in the finite von Neumann algebra setting.

\begin{lemma}\label{lemma:VI.9}
For any $\lambda\in(0,1)$ we have
\begin{align}\label{(VI.178)}
\inf_P\bigl\{\lambda\tau(\rho P)+(1-\lambda)\tau(\sigma(\egy-P))\bigr\}
\ge{1\over2}\inf_\phi\int_{[0,+\infty)^2}\bigl\{\lambda p\phi+(1-\lambda)q(1-\phi)\}\,d\mu,
\end{align}
where $\inf_P$ is taken over all projections $P$ in $\M$ and $\inf_\phi$ is over all Borel functions $\phi$ with
$0\le\phi\le1$.
\end{lemma}

\begin{proof}
It is obvious that $\inf_\phi$ in the RHS is equal to
\[
\int_{[0,+\infty)^2}\min\{\lambda p,(1-\lambda)q\}\,d\mu.
\]
Hence we need to show that for any projection $P\in\M$,
\begin{align}\label{(VI.179)}
\lambda\tau(\rho P)+(1-\lambda)\tau(\sigma(\egy-P))
\ge{1\over2}\int_{[0,+\infty)^2}\min\{\lambda p,(1-\lambda)q\}\,d\mu.
\end{align}
Consider finite partitions
\begin{align*}
\Delta_x&:\ 0=x_0<x_1<x_2<\dots<x_m<+\infty, \\
\Delta_y&:\ 0=y_0<y_1<y_2<\dots<y_n<+\infty.
\end{align*}
Note that
\begin{align*}
&\lambda\tau(\rho P)+(1-\lambda)\tau(\sigma(\egy-P)) \\
&\quad=\int_{[0,+\infty)}\lambda x\,d\tau(F(x)P)+\int_{[0,+\infty)}(1-\lambda)y\,d\tau((\egy-P)F(y)) \\
&\quad=\int_{[0,+\infty)^2}\lambda x\,d\|G(x,y)P\|_2^2+\int_{[0,+\infty)^2}(1-\lambda)y\,d\|G(x,y)(\egy-P)\|_2^2 \\
&\quad=\sup_{\Delta_x,\Delta_y}\sum_{i=1}^m\sum_{j=1}^n
\bigl\{\lambda x_{i-1}\|G(R_{ij})P\|_2^2+(1-\lambda)y_{j-1}\|G(R_{ij})(\egy-P)\|_2^2\bigr\},
\end{align*}
where $R_{ij}:=[x_{i-1},x_i)\times[y_{j-1},y_j)$. On the other hand,
\[
\int_{[0,+\infty)^2}\min\{\lambda p,(1-\lambda)q\}\,d\mu
=\sup_{\Delta_x,\Delta_y}\sum_{i=1}^m\sum_{j=1}^n\min\{\lambda x_i,(1-\lambda)y_j\}\mu(R_{ij}).
\]
Hence, to prove \eqref{(VI.179)}, it suffices to show that
\[
x\|G(R)P\|_2^2+y\|G(R)(\egy-P)\|_2^2\ge{1\over2}\min\{x,y\}\mu(R)
\]
for any $x,y\in[0,+\infty)$ and any rectangle $R=S\times T$ with $R,S\in\B_{[0,+\infty)}$. But this is easy as
\begin{align*}
x\|G(R)P\|_2^2+y\|G(R)(\egy-P)\|_2^2
&\ge\min\{x,y\}(\|G(R)P\|_2^2+\|G(R)(\egy-P)\|_2^2) \\
&\ge\min\{x,y\}{1\over2}(\|G(R)P\|_2+\|G(R)(\i-P)\|_2)^2 \\
&\ge{1\over2}\min\{x,y\}\|G(R)P+G(R)(\egy-P)\|_2^2 \\
&={1\over2}\min\{x,y\}\|G(R)\egy\|_2^2={1\over2}\min\{x,y\}\mu(R)
\end{align*}
thanks to \eqref{(VI.176)}.
\end{proof}

The following corollary and lemma are shown similarly to \cite[Sec.~3.1.2]{HircheTomamichel_integral}.

\begin{cor}\label{cor:VI.10}
For any $t>0$,
\[
D_{(\id-t)_+}^\meas(\rho\|\sigma)\le{1\over2}\,\tau(\rho)+{1\over2}D_{(\id-t)_+}(p\,d\mu\|q\,d\mu).
\]
\end{cor}

\begin{proof}
The LHS of \eqref{(VI.178)} is
\begin{align*}
&(1-\lambda)\tau(\sigma)+\inf_P\tau((\lambda\rho-(1-\lambda)\sigma)P) \\
&\quad=(1-\lambda)\tau(\sigma)-\lambda\tau\Bigl(\rho-{1-\lambda\over\lambda}\,\sigma\Bigr)_- \\
&\quad=(1-\lambda)\tau(\sigma)-\lambda\Bigl[\tau\Bigl(\rho-{1-\lambda\over\lambda}\,\sigma\Bigr)_+
-\tau\Bigl(\rho-{1-\lambda\over\lambda}\,\sigma\Bigr)\Bigr] \\
&\quad=\lambda\tau(\rho)-\lambda D_{(\id-{1-\lambda\over\lambda})_+}^\meas(\rho\|\sigma),
\end{align*}
and similarly the RHS of \eqref{(VI.178)} is
\[
{\lambda\over2}\tau(\rho)-{\lambda\over2}\,D_{(\id-{1-\lambda\over\lambda})_+}(p\,d\mu\|q\,d\mu).
\]
Hence \eqref{(VI.178)} implies that
\[
\lambda\tau(\rho)-\lambda D_{(\id-{1-\lambda\over\lambda})_+}^\meas(\rho\|\sigma)
\ge{\lambda\over2}\tau(\rho)-{\lambda\over2}\,D_{(\id-{1-\lambda\over\lambda})_+}(p\,d\mu\|q\,d\mu),
\]
that is,
\[
D_{(\id-{1-\lambda\over\lambda})_+}^\meas(\rho\|\sigma)
\le{1\over2}\,\tau(\rho)+{1\over2}\,D_{(\id-{1-\lambda\over\lambda})_+}(p\,d\mu\|q\,d\mu).
\]
Letting $t={1-\lambda\over\lambda}$ gives the asserted inequality.
\end{proof}

\begin{lemma}\label{lemma:VI.11}
Assume that $\M$ is a finite von Neumann algebra with a faithful normal finite trace $\tau$, and let
$\rho,\sigma\in\M_{*,\ge0}$. Then for any $\alpha\in(0,1)$,
\[
D_\alpha^{\meas,\hs}(\rho\|\sigma)\le D_\alpha^{\mathrm{P}}(\rho\|\sigma)+{1\over1-\alpha}\log 2.
\]
\end{lemma}

\begin{proof}
By \eqref{(VI.162)} in the classical case we have
\begin{align*}
Q_\alpha(p\,d\mu\|q\,d\mu)&=\alpha\tau(\rho)+(1-\alpha)\tau(\sigma) \\
&\quad+\alpha(\alpha-1)\int_1^{+\infty}\bigl\{t^{\alpha-2}D_{(\id-t)_+}(p\,d\mu\|q\,d\mu)
+t^{1-\alpha}D_{(\id-t)_+}(q\,d\mu\|p\,d\mu)\bigr\}\,dt.
\end{align*}
From \eqref{(VI.162)} and Corollary \ref{cor:VI.10} together with the above expression it follows that
\begin{align*}
Q_\alpha^{\meas,\hs}(\rho\|\sigma)
&\ge\alpha\tau(\rho)+(1-\alpha)\tau(\sigma) \\
&\quad+\alpha(\alpha-1)\int_1^{+\infty}\Bigl[
t^{\alpha-2}\Bigl({1\over2}\,\tau(\rho)+{1\over2}\,D_{\id-t)_+}^\meas(p\,d\mu\|q\,d\mu)\Bigr) \\
&\hskip3.5cm
+t^{-1-\alpha}\Bigl({1\over2}\tau(\sigma)+{1\over2}\,D_{(\id-t)_+}^\meas(q\,d\mu\|p\,d\mu)\Bigr)\Bigr]\,dt \\
&={1\over2}\{\alpha\tau(\rho)+(1-\alpha)\tau(\sigma)\}
+{\alpha(\alpha-1)\over2}\int_1^{+\infty}\bigl\{t^{\alpha-2}\tau(\rho)+t^{-1-\alpha}\tau(\sigma)\bigr\}\,dt \\
&\qquad+{1\over2}\,Q_\alpha(p\,d\mu\|q\,d\mu) \\
&={1\over2}\{\alpha\tau(\rho)+(1-\alpha)\tau(\sigma)\}
+{\alpha(\alpha-1)\over2}\Bigl\{{\tau(\rho)\over1-\alpha}+{\tau(\sigma)\over\alpha}\Bigr\}
+{1\over2}\,Q_\alpha(p\,d\mu\|q\,d\mu) \\
&={1\over2}\,Q_\alpha(\rho\|\sigma).
\end{align*}
thanks to \eqref{(VI.177)}. If If $\rho(\egy)\le\sigma(\egy)$, then we have
\[
Q_\alpha^{\meas,\hs}(\rho\|\sigma)\ge{1\over2}\,Q_\alpha(\rho\|\sigma)
\]
which yields the asserted inequality.
\end{proof}

\begin{lemma}\label{lemma:VI.12}
The assertion of Lemma \ref{lemma:VI.11} holds in the general von Neumann algebra $\M$.
\end{lemma}

\begin{proof}
We may assume that $\M$ is $\sigma$-finite, since both $D_\alpha^{\meas,\hs}(\rho\|\sigma)$ and
$D_\alpha^{\mathrm{P}}(\rho\|\sigma)$ are the same as those of $\rho,\sigma$ restricted on $e\M e$,
respectively, where $e$ is the support projection of $\rho+\sigma$. We use the Haagerup reduction theorem.
In the notations in Sec.~0.1, Lemma \ref{lemma:VI.11} implies that
\[
D_\alpha^{\meas,\hs}(\hat\rho|_{\M_n}\|\hat\sigma|_{\M_n})
\le D_\alpha^{\mathrm{P}}(\hat\rho|_{\M_n}\|\hat\sigma|_{\M_n})+{1\over1-\alpha}\log 2
\]
for all $n\in\bN$. By \eqref{(VI.165)} we have
$D_\alpha^{\mathrm{P}}(\hat\rho|_{\M_n}\|\hat\sigma|_{\M_n})\nearrow D_\alpha^{\mathrm{P}}(\rho\|\sigma)$.
By the martingale convergence in Proposition \ref{prop:V.17} we have
\[
Q_\alpha^{\meas,\hs}(\hat\rho|_{\M_n}\|\hat\sigma|_{\M_n})\searrow
Q_\alpha^{\meas,\hs}(\hat\rho\|\hat\sigma)
\]
so that
\[
D_\alpha^{\meas,\hs}(\hat\rho|_{\M_n}\|\hat\sigma|_{\M_n})\nearrow
D_\alpha^{\meas,\hs}(\hat\rho\|\hat\sigma).
\]
By the DPI of $D_\alpha^{\meas,\hs}$ it further follows that
\[
D_\alpha^{\meas,\hs}(\hat\rho\|\hat\sigma)=D_\alpha^{\meas,\hs}(\rho\|\sigma).
\]
Hence the result follows.
\end{proof}

\begin{prop}\label{prop:VI.13}
Let $\M$ be a general von Neumann algebra and $\rho,\sigma\in\M_{+,\ge0}$. If $\rho(\egy)\le\sigma(\egy)$, then
for any $\alpha\in(0,1)$, $\overline D_\alpha^{\meas,\hs}(\rho\|\sigma)=
\lim_{n\to\infty}{1\over n}\,D_\alpha^{\meas,\hs}(\rho^{\otimes n}\|\sigma^{\otimes n})$ exists and
\[
\overline D_\alpha^{\meas,\hs}(\rho\|\sigma)=D_\alpha(\rho\|\sigma).
\]
\end{prop}

\begin{proof}
Since Lemma \ref{lemma:VI.12} gives
\[
{1\over n}\,D_\alpha^{\meas,\hs}(\rho^{\otimes n}\|\sigma^{\otimes n})
\le D_\alpha^{\mathrm{P}}(\rho\|\sigma)+{1\over n}{1\over1-\alpha}\log 2
\]
for all $n\in\bN$, we have \eqref{(VI.175)}. The result follows from Lemma \ref{lemma:VI.6} and \eqref{(VI.175)}.
\end{proof}

\section{Equality cases between $D_f^\meas(\rho\|\sigma)$, $D_f^{\meas,\hs}(\rho\|\sigma)$,
$D_f^{\max,\hs}(\rho\|\sigma)$ and $D_f^{\max}(\rho\|\sigma)$}\label{sec:VII}

We have the ordering \eqref{hs fdiv order vN} for any $f\in(\bR^{(0,+\infty)})_{\conv}$ and
$\rho,\sigma\in\M_{*,\ge0}$, and all of the above inequalities become equality if $\rho,\sigma$ commute
(see \cite[Definition A.57]{Hiai2021Div} on the commutativity of $\rho,\sigma$ in the von Neumann algebra
setting). In this notes we examine the converse direction: Does the equality case of an inequality in the
above imply that $\rho,\sigma$ commute?

\subsection{Case $D_f^\meas(\rho\|\sigma)=D_f^{\meas,\hs}(\rho\|\sigma)$}\label{sec:VII.A}

First, let us consider the finite dimensional case for simplicity. Let $\rho,\sigma\in\B(\hil)_{\ge0}$ where
$d:=\dim\hil<+\infty$, and let $f:(0,+\infty)\to\bR$ be a convex function such that $\supp(df')=(0,+\infty)$.
Assume that
\begin{align}\label{(VII.180)}
D_f^\meas(\rho\|\sigma)=D_f^{\meas,\hs}(\rho\|\sigma)<+\infty.
\end{align}
By \cite[Proposition 4.17]{HiaiMosonyi2017} there exists a measurement $\fM=(M_i)_{i=1}^n$ (where we can
let $n\le d^2$) such that
\[
D_f^\meas(\rho\|\sigma)=D_f(\fM(\rho)\|\fM(\sigma))
\ \biggl(=\sum_{i=1}^n(\Tr\,M_i\sigma)f\Bigl({\Tr\,M_i\rho\over\Tr\,M_i\sigma}\Bigr)\biggr),
\]
so that we have
\begin{align}
D_f^\meas(\rho\|\sigma)&=f(1)\Tr\,\sigma+f'(1^+)\Tr(\rho-\sigma) \nonumber\\
&\quad+\int_{(0,1]}D_{(\id-t)_-}(\fM(\rho)\|\fM(\sigma))\,df'(t)
+\int_{(1,+\infty)}D_{(\id-t)_+}(\fM(\rho)\|\fM(\sigma))\,df'(t). \label{(VII.181)}
\end{align}
On the other hand,
\begin{align}
D_f^{\meas,\hs}(\rho\|\sigma)&=f(1)\Tr\,\sigma+f'(1^+)\Tr(\rho-\sigma) \nonumber\\
&\quad+\int_{(0,1]}D_{(\id-t)_-}^\meas(\rho\|\sigma)\,df'(t)
+\int_{(1,+\infty)}D_{(\id-t)_+}^\meas(\rho\|\sigma)\,df'(t)<+\infty. \label{(VII.182)}
\end{align}
Since $D_{(\id-t)_\pm}(\fM(\rho)\|\fM(\sigma))\le D_{(\id-t)_\pm}^\meas(\rho\|\sigma)$ for all $t\in(0,+\infty)$.
we have
\[
\begin{cases}D_{(\id-t)_-}(\fM(\rho)\|\fM(\sigma))=D_{(\id-t)_-}^\meas(\rho\|\sigma)
& \mbox{for $df'$-a.e.\ $t\in(0,1]$}, \\
D_{(\id-t)_+}(\fM(\rho)\|\fM(\sigma))=D_{(\id-t)_-}^\meas(\rho\|\sigma)
& \mbox{for $df'$-a.e.\ $t\in(1,+\infty)$}.
\end{cases}
\]
From \eqref{eq:hsp def3}--\eqref{eq:hsp meas} these imply that
$\|\fM(\rho)-t\fM(\sigma)\|_1=\|\rho-t\sigma\|_1$ for a.e.\ $df'$-a.e.\ $t\in(0,+\infty)$. From the assumption
on the support of $df'$ and the continuity of $t\mapsto\|\fM(\rho)-\fM(\sigma)\|_1$ and
$t\mapsto\|\rho-t\sigma\|_1$ on $(0,+\infty)$ we have
\[
\|\fM(\rho)-t\fM(\sigma)\|_1=\|\rho-t\sigma\|_1\quad\mbox{for all $t\in(0,+\infty)$},
\]
from which it is easy to see that for any $\lambda_i\in(0,1)$, $i=1,2$,
\begin{align*}
&\|\{\lambda_1\fM(\rho)+(1-\lambda_1)\fM(\sigma)\}-t\{\lambda_2\fM(\rho)+(1-\lambda_2)\fM(\sigma)\}\|_1 \\
&\quad=\|\{\lambda_1\rho+(1-\lambda_1)\sigma\}-t\{\lambda_2\rho+(1-\lambda_2)\sigma\}\|_1
\quad\mbox{for all $t\in(0,+\infty)$}.
\end{align*}
This implies by Theorem \ref{theorem:V.21} that
\[
D(\fM(\lambda_1\rho+(1-\lambda_1)\sigma)\|\fM(\lambda_2\rho+(1-\lambda_2)\sigma))
=D(\lambda_1\rho+(1-\lambda_1)\sigma\|\lambda_2\rho+(1-\lambda_2)\sigma),
\]
so that $\lambda_1\rho+(1-\lambda_2)\sigma$ and $\lambda_2\rho+(1-\lambda_2)\sigma$ commute by
\cite[Theorem 7.10]{Hiai2021Div}. Hence $\rho,\sigma$ commute.

Next, we consider the von Neumann algebra case.

\begin{thm}\label{theorem:VII.1}
Assume that $\M$ is an injective von Neumann algebra, and let $\rho,\sigma\in\M_{*,\ge0}$. If
$D_f^\meas(\rho\|\sigma)=D_f^{\meas,\hs}(\rho\|\sigma)<+\infty$ for some $f\in(\bR^{(0,+\infty)})_{\conv}$ with
$\supp(df')=(0,+\infty)$, then $\rho,\sigma$ commute, i.e., $h_\rho h_\sigma=h_\sigma h_\rho$ (see
\cite[Definition A.57]{Hiai2021Div}).
\end{thm}

\begin{proof}
Note that there exists a sequence $\{\fM_n\}$ of measurements (possibly with finite outcomes) in $\M$ such
that $D_f(\fM_n(\rho)\|\fM_n(\sigma))\to D_f^\meas(\rho\|\sigma)$ as $n\to\infty$ (see
\cite[Sec.~5.1]{Hiai2021Div}). Then the assumption says that
$D_f(\fM_n(\rho)\|\fM_n(\sigma))\to D_f^{\meas,\hs}(\rho\|\sigma)<+\infty$ as $n\to\infty$, so that
\begin{align*}
&\int_{(0,1]}D_{(\id-t)_-}(\fM_n(\rho)\|\fM_n(\sigma))\,df'(t)
+\int_{(1,+\infty)}D_{(\id-t)_+}(\fM_n(\rho)\|\fM_n(\sigma))\,df'(t) \\
&\quad\to\int_{(0,1]}D_{(\id-t)_-}^\meas(\rho\|\sigma)\,df'(t)
+\int_{(1,+\infty)}D_{(\id-t)_+}^\meas(\rho\|\sigma)\,df'(t)<+\infty.
\end{align*}
Since $D_{(\id-t)_\pm}(\fM_n(\rho)\|\fM_n(\sigma))\le D_{(\id-t)_\pm}^\meas(\rho\|\sigma)$, this means that
\begin{align*}
&1_{(0,1]}(t)D_{(\id-t)_-}(\fM_n(\rho)\|\fM_n(\sigma))
+1_{(1,+\infty)}(t)D_{(\id-t)_+}(\fM_n(\rho)\|\fM_n(\sigma)) \\
&\quad\to1_{(0,1]}(t)D_{(\id-t)_-}^\meas(\rho\|\sigma)
+1_{(1,+\infty)}(t)D_{(\id-t)_+}^\meas(\rho\|\sigma)
\end{align*}
in the $L^1$-norm on $L^1((0,+\infty),df')$. Hence, by taking a subsequence of $\{\fM_n\}$ we may assume
that
\[
\begin{cases}
D_{(\id-t)_-}(\fM_n(\rho)\|\fM_n(\sigma)\to D_{(\id-t)_-}^\meas(\rho\|\sigma)
& \mbox{for $df'$-a.e.\ $t\in(0,1]$}, \\
D_{(\id-t)_+}(\fM_n(\rho)\|\fM_n(\sigma)\to D_{(\id-t)_+}^\meas(\rho\|\sigma)
& \mbox{for $df'$-a.e.\ $t\in(1,+\infty)$},
\end{cases}
\]
which implies that
\[
D_{(\id-t)_\pm}(\fM_n(\rho)\|\fM_n(\sigma))\to D_{(\id-t)_\pm}^\meas(\rho\|\sigma)
\quad\mbox{for $df'$-a.e.\ $t\in(0,+\infty)$}.
\]
Since both $D_{(\id-t)_+}(\fM_n(\rho)\|\fM_n(\sigma))$ and $D_{(\id-t)_-}(\fM_n(\rho)\|\fM_n(\sigma))$ are
convex in $t\in(0,+\infty)$, we have, as shown in Lemma \ref{lemma:VII.2} below just for completeness,
\begin{align}\label{(VII.183)}
D_{(\id-t)_\pm}(\fM_n(\rho)\|\fM_n(\sigma))\to D_{(\id-t)_\pm}^\meas(\rho\|\sigma),\qquad t\in(0,+\infty).
\end{align}
which implies in view of \eqref{D-meas++} that
\begin{align}\label{(VII.184)}
D_{(\id-t)_\pm}\Bigl(\fM_n(\rho)\Big\|\fM_n\Bigl({\rho+\sigma\over2}\Bigr)\Bigr)
\to D_{(\id-t)_\pm}^\meas\Bigl(\rho\Big\|{\rho+\sigma\over2}\Bigr),\qquad t\in(0,+\infty).
\end{align}
Since
\begin{align*}
D_{(\id-t)_\pm}(\fM_n(\rho)\|\fM_n(\sigma))&=tD_{(\id-t^{-1})_\mp}(\fM_n(\sigma)\|\fM_n(\rho)), \\
D_{(\id-t)_\pm}^\meas(\rho\|\sigma)&=tD_{(\id-t^{-1})_\mp}^\meas(\sigma\|\rho)
\end{align*}
for all $t\in(0,+\infty)$, it follows from \eqref{(VII.184)} that
\[
D_{(\id-t)_\pm}\Bigl(\fM_n\Bigl({\rho+\sigma\over2}\Bigr)\Big\|\fM_n(\rho)\Bigr)
\to D_{(\id-t)_\pm}^\meas\Bigl({\rho+\sigma\over2}\Big\|\rho\Bigr),\qquad t\in(0,+\infty).
\]
which implies as in \eqref{(VII.183)}--\eqref{(VII.184)} above that
\[
D_{(\id-t)_\pm}\Bigl(\fM_n\Bigl({\rho+\sigma\over2}\Bigr)\Big\|\fM_n\Bigl({3\rho+\sigma\over4}\Bigr)\Bigr)
\to D_{(\id-t)_\pm}^\meas\Bigl({\rho+\sigma\over2}\Big\|{3\rho+\sigma\over4}\Bigr),\qquad t\in(0,+\infty).
\]
Note that the functions $t\mapsto D_{(\id-t)_-}\bigl(\fM_n\bigl({\rho+\sigma\over2}\bigr)\big\|
\fM_n\bigl({3\rho+\sigma\over4}\bigr)\bigr)$ are supported on $[a,+\infty)$ with some $a>0$ and uniformly
bounded there, and $t\mapsto D_{(\id-t)_+}\bigl(\fM_n\bigl({\rho+\sigma\over2}\bigr)\big\|
\fM_n\bigl({3\rho+\sigma\over4}\bigr)\bigr)$ are supported on $(0,b]$ with some $b<+\infty$ and uniformly
bounded there. Hence by the Lebesgue convergence theorem we have
\[
D_{x\log x}\Bigl(\fM_n\Bigl({\rho+\sigma\over2}\Bigr)\Big\|\fM_n\Bigl({3\rho+\sigma\over4}\Bigr)\Bigr)
\to D_{x\log x}^\meas\Bigl({\rho+\sigma\over2}\Big\|{3\rho+\sigma\over4}\Bigr),
\]
which means by Corollary \ref{cor:V.19} that
\[
D\Bigl(\fM_n\Bigl({\rho+\sigma\over2}\Bigr)\Big\|\fM_n\Bigl({3\rho+\sigma\over4}\Bigr)\Bigr)
\to D\Bigl({\rho+\sigma\over2}\Big\|{3\rho+\sigma\over4}\Bigr).
\]
Thus \cite[Theorem 7.10]{Hiai2021Div} implies that ${\rho+\sigma\over2}$ and ${3\rho+\sigma\over4}$ commute,
so that $\rho,\sigma$ commute.
\end{proof}

\begin{lemma}\label{lemma:VII.2}
Let $\{\phi_n\}$ be a sequence of $\bR$-valued convex functions on $(0,+\infty)$. If there exists a dense subset
$S$ of $(0,+\infty)$ such that $\lim_{n\to\infty}\phi_n(s)$ exists in $\bR$ for all $s\in S$, then $\phi_n(t)$
converges for all $x\in(0,+\infty)$.
\end{lemma}

\begin{proof}
Let $t\in(0,+\infty)$ be arbitrary. Choose $s_i\in S$, $1\le i\le4$, such that $s_1<s_2<t<s_3<s_4$. Since
$\lim_{n\to\infty}\phi_n(s_i)\in\bR$ for $1\le i\le4$, it follows that
\[
\beta:=\sup_{n\in\bN}\biggl\{\Big|{\phi_n(s_1)-\phi_n(s_2)\over s_1-s_2}\Big|,
\Big|{\phi_n(s_3)-\phi_n(s_4)\over s_3-s_4}\Big|\biggr\}<+\infty.
\]
For any $\ep>0$ choose an $s\in S$ such that $s\in(s_2,s_3)$ and $\beta|t-s|<\ep$. Then for each
$m,n\in\bN$ we have
\[
|\phi_m(t)-\phi_n(t)|\le|\phi_m(t)-\phi_m(s)|+|\phi_m(s)-\phi_n(s)|+|\phi_n(s)-\phi_n(t)|.
\]
By convexity we furthermore have
\[
{\phi_m(s_1)-\phi_m(s_2)\over s_1-s_2}\le{\phi_m(t)-\phi_m(s)\over t-s}
\le{\phi_m(s_3)-\phi_m(s_4)\over s_3-s_4}.
\]
Hence $\big|{\phi_m(t)-\phi_m(s)\over t-s}\big|\le\beta$ so that $|\phi_m(t)-\phi_m(s)|\le\ep$ and similarly
$|\phi_n(s)-\phi_n(t)|\le\ep$. Moreover, there exists an $n_0\in\bN$ such that $|\phi_m(s)-\phi_n(s)|\le\ep$
for all $m,n\ge n_0$. Therefore, $|\phi_m(t)-\phi_n(t)|\le3\ep$ for all $m,n\ge n_0$, implying that $\{\phi_n(t)\}$
is Cauchy, and the assertion follows.
\end{proof}

\subsection{Case $D_f^{\meas,\hs}(\rho\|\sigma)=D_f^{\max}(\rho\|\sigma)$}\label{sec:VII.B}

Let $\M$ be an injective von Neumann algebra, and  for simplicity let $f$ be a non-linear operator convex
function. Then $D_f^{\max}(\rho\|\sigma)$ can be defined in the method of operator perspectives (see
\cite[Chap.~4]{Hiai2021Div}). For each $\rho,\sigma\in\M_*^+$ we have a reverse test $(p,q,\Gamma)$ such
that $D_f^{\max}(\rho\|\sigma)=D_f(p\|q)$. (In fact, we can take $\Gamma:L^1([0,1],\nu)\to\M_*$ with a finite
Borel measure $\nu$, $p(x)=x$ and $q(x)=1-x$ for $x\in[0,1]$.) Now, assume that
\[
D_f^{\meas,\hs}(\rho\|\sigma)=D_f^{\max}(\rho\|\sigma)<+\infty.
\]
Then we have
\begin{align*}
&\int_{(0,1]}D_{(\id-t)_-}^\meas(\rho\|\sigma)\,df'(t)+\int_{(1,+\infty)}D_{(\id-t)_+}^\meas(\rho\|\sigma)\,d'(t) \\
&\quad=\int_{(0,1]}D_{(\id-t)_-}(p\|q)\,df'(t)+\int_{(1,+\infty)}D_{(\id-t)_+}(p\|q)\,df'(t)<+\infty.
\end{align*}
Since $D_{(\id-t)_\pm}^\meas(\rho\|\sigma)\le D_{(\id-t)_\pm}(p\|q)$ for all $t\in(0,+\infty)$, it follows that
$D_{(\id-t)_-}^\meas(\rho\|\sigma)=D_{(\id-t)_-}(p\|q)$ for $df'$-a.e.\ $t\in(0,1]$ and
$D_{(\id-t)_+}^\meas(\rho\|\sigma)=D_{(\id-t)_+}(p\|q)$ for $df'$-a.e.\ $t\in(1,+\infty)$. Since $f$ is non-linear
and operator convex, note that $\supp(df')=(0,+\infty)$. Hence we have $\|\rho-t\sigma\|=\|p-tq\|_1$ for all
$t\in(0,+\infty)$. Hence by Theorem \ref{theorem:V.21} it follows that $\Gamma$ is reversible for $p,q$, so that
there is a positive trace-preserving map $\Lambda:\M_*\to L^1([0,1],\nu)$ such that $\Lambda(\rho)=p$ and
$\Lambda(\sigma)=q$. Therefore,
\[
\Gamma\circ\Lambda(\alpha\rho+\beta\sigma)=\alpha\rho+\beta\sigma
\]
for any $\alpha,\beta\ge0$. By \cite[Theorem 7.10]{Hiai2021Div} this shows that $\rho,\sigma$ commute.

The result is written in the following:

\begin{thm}\label{theorem:VII.3}
Assume that $\M$ is injective, and let $\rho,\sigma\in\M_{*,\ge0}$. If
$D_f^{\meas,\hs}(\rho\|\sigma)=D_f^{\max}(\rho\|\sigma)<+\infty$ for some non-linear operator convex function
$f$ on $(0,+\infty)$, then $\rho,\sigma$ commute.
\end{thm}


\begin{remark}\label{remark:VII.4}
Note that the above proof works when $D_f^{\meas,\hs}(\rho\|\sigma)=D_f^{\max}(\rho\|\sigma)<+\infty$ for
some $f\in(\bR^{(0,+\infty)})_{\conv}$ and there exists a minimal reverse test $(p,q,\Gamma)$ so that
$D_f^{\max}(\rho\|\sigma)=D_f(p\|q)$. However, it seems unknown if a minimal reverse test for 
$D_f^{\max}(\rho\|\sigma)$ exists for general convex functions $f$, even in the finite dimensional case.
\end{remark}

\begin{remark}\label{remark:VII.5}
We make a speculation as to how we can more directly prove Theorem \ref{theorem:VII.3} in the finite
dimensional case. Let $\rho,\sigma\in\B(\hil)_{\ge0}$ where $d:=\dim\hil$, and $f$ be a non-linear operator
convex function on $(0,+\infty)$. We have a minimal reverse test $(p,q,\Gamma)$ with
$\Gamma:\bC^d\to\B(\hil)$ for which $D_f^{\max}(\rho\|\sigma)=D_f(p\|q)$. If
$D_f^{\meas,\hs}(\rho\|\sigma)=D_f^{\max}(\rho\|\sigma)<+\infty$, then we have
$D_{(\id-t)_\pm}^\meas(\rho\|\sigma)=D_{(\id-t)_\pm}(p\|q)$ for all $t\in(0,+\infty)$, which implies that
\begin{align}
\Tr(\rho-t\sigma)_+&=\sum_{i=1}^d(p(i)-tq(i))_+, \label{(VII.185)}\\
\|\rho-t\sigma\|_1&=\sum_{i=1}^d|p(i)-tq(i)| \label{(VII.186)}
\end{align}
for all $t\in(0,+\infty)$. It is obvious that the RHSs of \eqref{(VII.185)} and \eqref{(VII.186)} are piecewise
linear functions of $t\in(0,+\infty)$. Thus, the result follows when we can prove that if
$\rho\sigma\ne\sigma\rho$ then the function $t\mapsto\Tr(\rho-t\sigma)_+$ on $(0,+\infty)$ is not piecewise
linear. This can easily be verified in the qubit case, while it does not seem easy 
\end{remark}

\subsection{Case $D_f^{\meas,\hs}(\rho\|\sigma)=D_f^{\max,\hs}(\rho\|\sigma)$?}\label{sec:VII.C}

This is the most interesting equality case. To discuss this equality case, the following observation seems
useful, which is an extension of Matsumoto's argument in \cite[Sec.~8.2]{Matsumoto_newfdiv} to the
von Neumann algebra case.

\begin{lemma}\label{lemma:VII.6}
For any $\rho,\sigma\in\M_{*,\ge0}$ the following conditions are equivalent:
\begin{itemize}
\item[(i)] $D_{(\id-1)_+}^\meas(\rho\|\sigma)=D_{(\id-1)_+}^{\max}(\rho\|\sigma)$;
\item[(ii)] $\max\{\eta(\egy):\eta\in\M_{*,\sa},\,\eta\le\rho,\,\eta\le\sigma\}
=\max\{\eta(\egy):\eta\in\M_{*,\ge0},\,\eta\le\rho,\,\eta\le\sigma\}$;
\item[(iii)] $\rho-(\rho-\sigma)_+\ge0$.
\end{itemize}
\end{lemma}

\begin{proof}
Note that
\begin{align}
&\sup\{\eta(\egy):\eta\in\M_{*,\sa},\,\eta\le\rho,\,\eta\le\sigma\} \label{(VII.187)}\\
&\quad=\sup\{\eta(\egy)-(\rho-\eta)(\egy):\eta\in\M_{*,\sa},\,\rho-\eta\ge0,\,\rho-\eta\ge\rho-\sigma\} \nonumber\\
&\quad=\rho(\egy)-\inf\{\zeta(\egy):\zeta\in\M_{*,\ge0},\,\zeta\ge\rho-\sigma\}. \label{(VII.188)}
\end{align}
Obviously $(\rho-\sigma)_+\ge\rho-\sigma$. If $\zeta\in\M_{*,\ge0}$ and $\zeta\ge\rho-\sigma$, then with
$e:=\{\rho-\sigma\ge0\}$,
\begin{align}\label{(VII.189)}
\zeta(\egy)\ge\zeta(e)\ge(\rho-\sigma)(e)=(\rho-\sigma)_+(\egy).
\end{align}
Hence the infimum in \eqref{(VII.188)} is $(\rho-\sigma)(\egy)$ and it is attained by $\zeta=(\rho-\sigma)_+$. Assume
that $\zeta\in\M_{*,\ge0}$ satisfies $\zeta\ge\rho-\sigma$ and $\zeta(\egy)=(\rho-\sigma)_+(\egy)$. Then
\eqref{(VII.189)} implies that $\zeta(\egy)=\zeta(e)=(\rho-\sigma)_+(\egy)$ so that
\[
\zeta=\zeta(e\cdot e)\ge(\rho-\sigma)(e\cdot e)=(\rho-\sigma)_+.
\]
Hence $\zeta=(\rho-\sigma)_+$, so that $\zeta=(\rho-\sigma)_+$ is a unique minimizer for the infimum in
\eqref{(VII.188)}. We have seen that the LHS of the equality in (ii) is
$\rho(\egy)-D_{(\id-1)_+}^\meas(\rho\|\sigma)$, while the RHS is
$\rho(\egy)-D_{(\id-1)_+}^{\max}(\rho\|\sigma)$ by \eqref{D-max+}. Thus (i)\,$\iff$\,(ii) holds.
Since $\eta=\rho-(\rho-\sigma)_+$ is a unique maximizer of the LHS in (ii), we have (ii)\,$\iff$\,(iii) as well.
\end{proof}

By Lemma \ref{lemma:VII.6} we have:

\begin{prop}\label{prop:VII.7}
Let $\rho,\sigma\in\M_{*,\ge0}$. Then the following conditions are equivalent:
\begin{itemize}
\item[(i)] $D_f^{\meas,\hs}(\rho\|\sigma)=D_f^{\max,\hs}(\rho\|\sigma)$ for all $f\in(\bR^{(0,+\infty)})_{\conv}$;
\item[(ii)] $D_f^{\meas,\hs}(\rho\|\sigma)=D_f^{\max,\hs}(\rho\|\sigma)<+\infty$ for some
$f\in(\bR^{(0,+\infty)})_{\conv}$ with $\supp(df')=(0,+\infty)$.
\item[(iii)] $D_{(\id-t)_+}^\meas(\rho\|\sigma)=D_{(\id-t)_+}^{\max}(\rho\|\sigma)$ for all $t\in(0,+\infty)$;
\item[(iv)] $\max\{\eta(\egy):\eta\in\M_{*,\sa},\,\eta\le\rho,\,\eta\le t\sigma\}
=\max\{\eta(\egy):\eta\in\M_{*,\ge0},\,\eta\le\rho,\,\eta\le t\sigma\}$ for all $t\in(0,+\infty)$;
\item[(v)] $\rho-(\rho-t\sigma)_+\ge0$ for all $t\in(0,+\infty)$.
\end{itemize}
\end{prop}


\begin{example}\label{example:VII.8}
Here we give an explicit qubit example showing that condition (v) does not imply $\rho\sigma=\sigma\rho$.
Let $\sigma:=I_2$ and $\rho:=\begin{bmatrix}1&1/2\\egy/2&1\end{bmatrix}$. The diagonalization of $\rho$ is
\[
\rho=\begin{bmatrix}{1\over\sqrt2}&{1\over\sqrt2}\\-{1\over\sqrt2}&{1\over\sqrt2}\end{bmatrix}
\begin{bmatrix}{1\over2}&0\\0&{3\over2}\end{bmatrix}
\begin{bmatrix}{1\over\sqrt2}&-{1\over\sqrt2}\\{1\over\sqrt2}&{1\over\sqrt2}\end{bmatrix}.
\]
Therefore, for each $t\in[1/2,3/2]$,
\begin{align*}
\rho-(\rho-t\sigma)_+
&=\begin{bmatrix}{1\over\sqrt2}&{1\over\sqrt2}\\-{1\over\sqrt2}&{1\over\sqrt2}\end{bmatrix}
\begin{bmatrix}{1\over2}-\bigl({1\over2}-t\bigr)_+&0\\0&{3\over2}-\bigl({3\over2}-t\bigr)_+\end{bmatrix}
\begin{bmatrix}{1\over\sqrt2}&-{1\over\sqrt2}\\{1\over\sqrt2}&{1\over\sqrt2}\end{bmatrix} \\
&=\begin{bmatrix}{1\over\sqrt2}&{1\over\sqrt2}\\-{1\over\sqrt2}&{1\over\sqrt2}\end{bmatrix}
\begin{bmatrix}{1\over2}&0\\0&t\end{bmatrix}
\begin{bmatrix}{1\over\sqrt2}&-{1\over\sqrt2}\\{1\over\sqrt2}&{1\over\sqrt2}\end{bmatrix}
=\begin{bmatrix}{t\over2}+{1\over4}&{t\over2}-{1\over4}\\{t\over2}-{1\over4}&{t\over2}+{1\over4}
\end{bmatrix}.
\end{align*}
Since $\det(\rho-(\rho-t\sigma)_+)=t/2>0$, it follows that $\rho-(\rho-t\sigma)_+>0$ for all $t\in[1/2,3/2]$. Also,
for $0<t\le1/2$, $\rho-t\sigma\ge0$ so that $\rho-(\rho-t\sigma)_+=t\sigma>0$. For $t\ge3/2$, since
$\rho-t\sigma\le0$ so that $\rho-(\rho-t\sigma)_+=\rho>0$. Hence $\rho-(\rho-t\sigma)_+>0$ for all $t>0$. So,
there is a $\delta>0$ such that $\rho-(\rho-t\sigma)_+\ge\delta I_2$ for all $t\ge1/2$. For $\ep\in(0,1)$ let
$\sigma_\ep:=\begin{bmatrix}1&0\\0&1-\ep\end{bmatrix}$. Then it is readily verified that
$(\rho-t\sigma_\ep)_+\to(\rho-t\sigma)_+$ as $\ep\searrow0$ uniformly for $t\in[1/2,2]$. Hence we can
choose an $\ep>0$ such that $\rho\le2\sigma_\ep$ and $\rho-(\rho-t\sigma_\ep)_+\ge{\delta\over2}\,I_2$
for all $t\in[1/2,2]$. Since $\rho\ge{1\over2}\,\sigma\ge t\sigma_\ep$ for $0<t\le1/2$ and
$\rho\le2\sigma_\ep\le t\sigma_\ep$ for $t\ge2$, we have $\rho-(\rho-t\sigma_\ep)_+>0$ for all $t>0$. But
$\rho\sigma_\ep\ne\sigma_\ep\rho$. We can modify the above argument to see that for each commuting
$\rho_0,\sigma_0>0$ in $\B(\hil)$ there exists an $\ep>0$ such that for any $\rho,\sigma\in\B(\hil)_{\ge0}$
with $\|\rho-\rho_0\|<\ep$ and $\|\sigma-\sigma_0\|<\ep$ we have $\rho-(\rho-t\sigma)_+>0$ for all $t>0$.
Thus, a non-commuting pair $(\rho,\sigma)$ satisfying (v) exists near $(\rho_0,\sigma_0)$.
\end{example}

\begin{example}\label{example:VII.9}
Here we show that if at least one of qubit $\rho$ and $\sigma$ is rank $1$, then condition (v) implies
$\rho\sigma=\sigma\rho$. Since
\begin{align*}
\rho-(\rho-t\sigma)_+&=\rho-t(t^{-1}\rho-\sigma)_+=\rho-t(\sigma-t^{-1}\rho)_- \\
&=t\{\sigma-(\sigma-t^{-1}\rho)_+\},
\end{align*}
we may assume that $\sigma$ is rank 1, so we may let $\sigma=\begin{bmatrix}1&0\\0&0\end{bmatrix}$.
Let us show that if $\rho\sigma\ne\sigma\rho$ then (v) fails to hold. Without loss of generality (by multiplying a
positive constant to $\rho$), we may assume that
\begin{align}\label{(VII.190)}
\rho=\begin{bmatrix}1&c\\\overline c&b\end{bmatrix},\qquad b\ge|c|^2>0,
\end{align}
and show that $\rho-(\rho-\sigma)_+\not\ge0$. Since
$\rho-\sigma=\begin{bmatrix}0&c\\\overline c&b\end{bmatrix}$, the eigenvalues of $\rho-\sigma$ are
\begin{align}\label{(VII.191)}
\lambda_{\pm}={b\pm\sqrt{b^2+4|c|^2}\over2}
\end{align}
so that $\lambda_-<0<\lambda_+$. The unital eigenvector for $\lambda_+$ is
\[
u_+={1\over\sqrt{|c|^2+\lambda_+^2}}\begin{bmatrix}c\\\lambda_+\end{bmatrix}
\]
so that
\[
(\rho-\sigma)_+=\lambda_+\pr{u_+}
={\lambda_+\over|c|^2+\lambda_+^2}
\begin{bmatrix}|c|^2&c\lambda_+\\\overline c\lambda_+&\lambda_+^2\end{bmatrix}.
\]
Therefore,
\[
\rho-(\rho-\sigma)_+=\begin{bmatrix}1&c\\\overline c&b\end{bmatrix}
-{\lambda_+\over|c|^2+\lambda_+^2}
\begin{bmatrix}|c|^2&c\lambda_+\\\overline c\lambda_+&\lambda_+^2\end{bmatrix}
=:\begin{bmatrix}a_{11}&a_{12}\\\overline{a_{12}}&a_{22}\end{bmatrix}
\]
where
\[
\begin{cases}
a_{11}=1-{|c|^2\lambda_+\over|c|^2+\lambda_+^2}
={|c|^2(1-\lambda_+)+\lambda_+^2\over|c|^2+\lambda_+^2}, \\
a_{22}=b-{\lambda_+^3\over|c|^2+\lambda_+^2}
={|c|^2b+b\lambda_+^2-\lambda_+^3\over|c|^2+\lambda_+^2}, \\
a_{12}=c-{c\lambda_+^2\over|c|^2+\lambda_+^2}={c|c|^2\over|c|^2+\lambda_+^2}.
\end{cases}
\]
Hence we have
\[
\det(\rho-(\rho-\sigma)_+)
={\{|c|^2(1-\lambda_+)+\lambda_+^2\}\{|c|^2b+b\lambda_+^2-\lambda_+^3\}-|c|^6\over
(|c|^2+\lambda_+^2)^2}.
\]
Since $|c|^2=\lambda_+^2-b\lambda_+$ by \eqref{(VII.191)}, we compute the above numerator as follows:
\begin{align*}
&\{(\lambda_+^2-b\lambda_+)(1-\lambda_+)+\lambda_+^2\}
\{(\lambda_+^2-b\lambda_+)b+b\lambda_+^2-\lambda_+^3\}-(\lambda_+^2-b\lambda_+)^3 \\
&\quad=\lambda_+^2\bigl[\{(\lambda_+-b)(1-\lambda_+)+\lambda_+\}
\{(\lambda_+-b)b+b\lambda_+-\lambda_+^2\}-(\lambda_+-b)^2(\lambda_+^2-b\lambda_+)\bigr] \\
&\quad=\lambda_+^2\bigl[-(\lambda_+-b-\lambda_+^2+b\lambda_++\lambda_+)(\lambda_+-b)^2
-(\lambda_+-b)^2(\lambda_+^2-b\lambda_+)\bigr] \\
&\quad=\lambda_+^2(\lambda_+-b)^2(-2\lambda_++b).
\end{align*}
Since $\lambda_+>b$ by \eqref{(VII.191)} and also $-2\lambda_++b<-b<0$ by \eqref{(VII.190)}, we have
$\det(\rho-(\rho-\sigma)_+)<0$ so that $\rho-(\rho-\sigma)_+\not\ge0$. By taking $\sigma+\ep I_2$ with the
above $\sigma$ and a small $\ep>0$ and the above $\rho$ with $b>|c|^2$, we find a non-commuting
$\rho,\sigma>0$ for which condition (v) does not hold.
\end{example}

The mechanism behind condition (v) is therefore metric -- governed by how large $\rho-t\sigma$ is relative
to how non-degenerate $\rho$ and $t\sigma$ are -- rather than algebraic. Thus, characterizing pairs
$(\rho,\sigma)\in\B(\hil)_{\ge0}$ satisfying the conditions of Proposition \ref{prop:VII.7} seems quite
complicated even in the qubit case. Moreover, we leave the last equality case
$D_f^{\max,\hs}(\rho\|\sigma)=D_f^{\max}(\rho\|\sigma)$ open.

\section*{Acknowledgments}

MM and MT were supported by the Ministry of Education, Singapore, through
grant T2EP20124-0005.
MM was also partially supported by the National Research, Development and Innovation Office
of Hungary (NKFIH) via the research grants K 146380 and EXCELLENCE 151342, and by the
Ministry of Culture and Innovation and the National Research, Development and Innovation
Office within the Quantum Information National Laboratory of Hungary (Grant No. 2022-2.1.1-
NL-2022-00004). The authors are grateful to Ludovico Lami for the opportunity to present part of this work at the 
YPetz conference \url{https://ypetz.filippo.info/}.

\appendix

\section{Radon-Nikodym derivatives}

We will use the notation
\begin{align*}
A\symdiff B:=(A\setminus B)\cup(B\setminus A)
\end{align*}
for the symmetric difference of two sets $A$ and $B$.
We will repeatedly use the following simple observation:

\begin{lemma}\label{lemma:zero RN}
Let $\mu,\nu$ be positive measures on a measurable space $(\Omega,\F)$ such that $\nu\ll\mu$. Then 
\begin{align}\label{eq:zero RN}
\nu\bz\left\{\frac{d\nu}{d\mu}=0\right\}\jz=0,\ds\ds\ds\ds
\nu\bz\left\{\frac{d\nu}{d\mu}>0\right\}\jz=\nu(\Omega).
\end{align}
\end{lemma}
\begin{proof}
With $N:=\{d\nu/d\mu=0\}$, we have
\begin{align*}
\nu(N)=\int_{\Omega}\underbrace{\egy_N\frac{d\nu}{d\mu}}_{\equiv 0}\,d\mu=0,
\end{align*}
proving the first equality in \eqref{eq:zero RN}, and the second equality is just a reformulation.
\end{proof}

\begin{lemma}\label{lemma:error welldef}
Let $P,Q$ be finite positive measures on a measurable space $(\Omega,\F)$, let 
$F:\,[0,+\infty)^2\to\bR$ be measurable and positive homogeneous of order $\kappa>0$,
and let $J$ be any of the intervals $(-\infty,0)$, $(-\infty,0]$,
$(0,+\infty)$, $[0,+\infty)$, $[0]$.     
For any positive measure $\mu$ on $(\Omega,\F)$ such that $P,Q\ll\mu$, 
and any representatives of the Radon-Nikodym derivatives
$dP/d\mu$, $dQ/d\mu$, $dP/d(P+Q)$, $dQ/d(P+Q)$,  
\begin{align}
&(P+Q)\bz\left\{F\bz\frac{dP}{d\mu},\frac{dQ}{d\mu}\jz\in J\right\}
\symdiff\left\{F\bz\frac{dP}{d(P+Q)},\frac{dQ}{d(P+Q)}\jz\in J\right\}\jz=0.
\label{eq:classial error welldef1}
\end{align}
Moreover, for any measurable homogeneous function $G:\,[0,+\infty)^2\to\bR\cup\{+\infty\}$,
\begin{align}\label{eq:classial error welldef2}
\int_{\left\{F\bz\frac{dP}{d\mu},\frac{dQ}{d\mu}\jz\in J\right\}}G\bz\frac{dP}{d\mu},\frac{dQ}{d\mu}\jz\,d\mu
=
\int_{\left\{F\bz\frac{dP}{d(P+Q)},\frac{dQ}{d(P+Q)}\jz\in J\right\}}G\bz\frac{dP}{d(P+Q)},\frac{dQ}{d(P+Q)}\jz\,d(P+Q).
\end{align}
\end{lemma}
\begin{proof} 
Let $\nu:=P+Q$, and fix some representative of the Radon-Nikodym derivative $d\nu/d\mu$; 
then there exists a set $E\in\F$ with $\mu(E)=0$ (and hence $\nu(E)=0$), such that 
\begin{align}\label{eq:error welldef proof0}
\frac{dP}{d\mu}(\omega)=\frac{dP}{d\nu}(\omega)\frac{d\nu}{d\mu}(\omega),
\ds\ds\ds
\frac{dQ}{d\mu}(\omega)=\frac{dQ}{d\nu}(\omega)\frac{d\nu}{d\mu}(\omega),
\ds\ds\omega\in\Omega\setminus E,
\end{align}
and hence
\begin{align}\label{eq:error welldef proof1}
F\bz\frac{dP}{d\mu}(\omega),\frac{dQ}{d\mu}(\omega)\jz
=
F\bz\frac{dP}{d\nu}(\omega)\frac{d\nu}{d\mu}(\omega),\frac{dQ}{d\nu}(\omega)\frac{d\nu}{d\mu}(\omega)\jz
=
F\bz\frac{dP}{d\nu}(\omega),\frac{dQ}{d\nu}(\omega)\jz
\bz\frac{d\nu}{d\mu}(\omega)\jz^{\kappa},
\ds\ds\omega\in\Omega\setminus E.
\end{align}
Let $N:=\{\omega\in\Omega:\,(d\nu/d\mu)(\omega)=0\}$; then  
\begin{align}\label{eq:error welldef proof2}
F\bz\frac{dP}{d\mu}(\omega),\frac{dQ}{d\mu}(\omega)\jz\in J\ds\iff\ds
F\bz\frac{dP}{d\nu}(\omega),\frac{dQ}{d\nu}(\omega)\jz\in J,\ds\ds\ds \omega\in\Omega\setminus(N\cup E),
\end{align}
according to \eqref{eq:error welldef proof1}.
Since $\nu(N)=0$ by Lemma \ref{lemma:zero RN}, 
we get that the set of $\omega$ where \eqref{eq:error welldef proof2} does not hold is of $\nu$-measure $0$,
completing the proof of \eqref{eq:classial error welldef1}.

The  proof of \eqref{eq:classial error welldef2} follows as
\begin{align}
\int_{\left\{F\bz\frac{dP}{d\mu},\frac{dQ}{d\mu}\jz\in J\right\}}G\bz\frac{dP}{d\mu},\frac{dQ}{d\mu}\jz\,d\mu
&=
\int_{\left\{F\bz\frac{dP}{d\mu},\frac{dQ}{d\mu}\jz\in J\right\}\setminus E}G\bz\frac{dP}{d\mu},\frac{dQ}{d\mu}\jz\,d\mu\\
&=
\int_{\left\{F\bz\frac{dP}{d\mu},\frac{dQ}{d\mu}\jz\in J\right\}\setminus E}
G\bz\frac{dP}{d\nu}\frac{d\nu}{d\mu},\frac{dQ}{d\nu}\frac{d\nu}{d\mu}\jz\,d\mu\\
&=
\int_{\left\{F\bz\frac{dP}{d\mu},\frac{dQ}{d\mu}\jz\in J\right\}\setminus E}
G\bz\frac{dP}{d\nu},\frac{dQ}{d\nu}\jz\,\underbrace{\frac{d\nu}{d\mu}d\mu}_{=d\nu}\\
&=
\int_{\left\{F\bz\frac{dP}{d\mu},\frac{dQ}{d\mu}\jz\in J\right\}\setminus (N\cup E)}
G\bz\frac{dP}{d\nu},\frac{dQ}{d\nu}\jz\,d\nu\\
&=
\int_{\left\{F\bz\frac{dP}{d\nu},\frac{dQ}{d\nu}\jz\in J\right\}\setminus (N\cup E)}
G\bz\frac{dP}{d\nu},\frac{dQ}{d\nu}\jz\,d\nu\\
&=
\int_{\left\{F\bz\frac{dP}{d\nu},\frac{dQ}{d\nu}\jz\in J\right\}}
G\bz\frac{dP}{d\nu},\frac{dQ}{d\nu}\jz\,d\nu,
\end{align}
where the first equality follows from $\mu(E)=0$, 
the second one from \eqref{eq:error welldef proof0},
the third one from the positive homogeneity of $G$,
the fourth one from $\nu(N)=0$,
the fifth one from \eqref{eq:error welldef proof2},
and the sixth one from $\nu(N\cup E)=0$.
\end{proof}

\begin{lemma}\label{lemma:RN composition}
Let $\mu$ be a signed measure on the Borel sets of $(0,+\infty)$ 
that is absolutely continuous to the Lebesgue measure with Radon-Nikodym derivative 
$h=\frac{d\mu}{d\lambda}$. For any continuously differentiable strictly decreasing function 
$g:\,(0,+\infty)\to(0,+\infty)$, the push-forward measure 
$\mu\circ g:=(g\inv)_*\mu$ is absolutely continuous to the Lebesgue measure with Radon-Nikodym derivative 
\begin{align}\label{eq:RN composition}
\frac{d(\mu\circ g)}{d\lambda}=(-g')\bz\frac{d\mu}{d\lambda}\circ g\jz.
\end{align}
\end{lemma}
\begin{proof}
Let $h:=d\mu/d\lambda$. For any Borel set $A\subseteq(0,+\infty)$, 
\begin{align}\label{eq:RN composition proof1}
(\mu\circ g)(A)&=\int\egy_{g(A)}h\,d\lambda=\int(\egy_A\circ g\inv)((h\circ g)\circ g\inv)\,d\lambda
=
\int\egy_A\cdot(h\circ g)\,d(g\inv)_*\lambda. 
\end{align}
Note that for any $0<a<b$, 
\begin{align*}
(\lambda\circ g)((a,b))=\lambda\bz(g(b),g(a))\jz=g(a)-g(b)=-\int_{(a,b)}g'\,d\lambda,
\end{align*}
whence $\lambda\circ g\ll\lambda$ with 
\begin{align*}
\frac{d(\lambda\circ g)}{d\lambda}=-g'.
\end{align*}
Writing this back into \eqref{eq:RN composition proof1}, we get 
\begin{align*}
(\mu\circ g)(A)&=
\int\egy_A\cdot(h\circ g)(-g')\,d\lambda,
\end{align*}
proving \eqref{eq:RN composition}.
\end{proof}

\section{Differentiability}

\begin{lemma}\label{lemma:onesided derivative}
Let $a\in[0,+\infty)$, $b,x_0\in\bR$, and $f$ be a function defined on an interval
containing $b-ax_0$ and taking values in a Banach space $X$.
\begin{enumerate}
\item\label{right derivative}
Assume that $f$ is defined on 
$(b-ax_0-\delta,b-ax_0]$ for some $\delta>0$ and $\partial^-f(b-ax_0)$ exists.
Then $x\mapsto f(b-ax)$ has a right derivative at $b-ax_0$, given by 
\begin{align*}
\partial^+ f(b-ax)\big\vert_{x=x_0}=
(-a)\partial^-f(b-ax_0).
\end{align*}

\item\label{left derivative}
Assume that $f$ is defined on 
$[b-ax_0,b-ax_0+\delta)$ for some $\delta>0$ and $\partial^+f(b-ax_0)$ exists.
Then $x\mapsto f(b-ax)$ has a left derivative at $b-ax_0$, given by 
\begin{align*}
\partial^- f(b-ax)\big\vert_{x=x_0}=
(-a)\partial^+f(b-ax_0).
\end{align*}
\end{enumerate}
\end{lemma}
\begin{proof}
We only prove \ref{right derivative}, as the proof of 
\ref{left derivative} goes exactly the same way.
The case $a=0$ is trivial, so we assume $a\in(0,+\infty)$. 
By definition, 
\begin{align*}
\partial^+ f(b-ax)\big\vert_{x=x_0}
&=
\lim_{0<\ep\searrow 0}\frac{1}{\ep}\bz f(b-a(x_0+\ep))-f(b-ax_0)\jz\\
&=
\lim_{0<\ep\searrow 0}(-a)\frac{1}{-a\ep}\bz f(b-ax_0-a\ep)-f(b-ax_0)\jz\\
&=
(-a)\lim_{0>\eta\nearrow 0}\frac{1}{\eta}\bz f(b-ax_0+\eta)-f(b-ax_0)\jz\\
&=
(-a)\partial^-f(b-ax_0).
\end{align*}
\end{proof}

\begin{cor}\label{cor:hs derivative}
Let $a\in[0,+\infty)$ and $b\in\bR$. For any $x_0\in\bR$,
\begin{align*}
&\exists\s\partial^+(b-ax)_+\big\vert_{x=x_0}=
\left\{\begin{array}{ll}
(-a)\cdot 0=0,&b-ax_0\in(-\infty,0],\\
(-a)\cdot 1=-a,&b-ax_0\in(0,+\infty)
\end{array}
\right\}
=
(-a)\egy_{(0,+\infty)}(b-ax_0),\\
&\exists\s\partial^-(b-ax)_+\big\vert_{x=x_0}=
\left\{\begin{array}{ll}
(-a)\cdot 0=0,&b-ax_0\in(-\infty,0),\\
(-a)\cdot 1=-a,&b-ax_0\in[0,+\infty)
\end{array}
\right\}
=
(-a)\egy_{[0,+\infty)}(b-ax_0).
\end{align*}
\end{cor}
\begin{proof}
Follows immediately from Lemma \ref{lemma:onesided derivative} applied to the function
$f(x):=x_+:=x\egy_{[0,+\infty)}(x)$.
\end{proof}

\begin{lemma}\label{lemma:integral partial derivative}
Let $(\Omega,\F,\mu)$ be a measure space, and let $f:\,\Omega\times J\to\bR$ be 
a function, where
$J$ is a non-degenerate interval.
Assume that for all $x\in J$, $f(\valt,x)\in L^1(\mu)$.
\begin{enumerate}
\item\label{integral partial1}
If 
$J=[x_0,x_0+\delta)$ for some $x_0\in\bR$ and $\delta>0$,
and there exists a $g\in L^1(\mu)$ such that 
$\sup_{x\in[x_0,x_0+\delta)}|\partial_2^+f(\omega,x)|\le g(\omega)$, $\omega\in\Omega$,
then 
\begin{align*}
\exists\s\partial_2^{+}\int_{\Omega}f(\omega,x)\,d\mu(\omega)\Bigg\vert_{x=x_0}
=
\int_{\Omega}\partial_2^+f(\omega,x)\big\vert_{x=x_0}\,d\mu(\omega).
\end{align*}
\item\label{integral partial2}
If 
$J=(x_0-\delta,x_0]$ for some $x_0\in\bR$ and $\delta>0$,
and there exists a $g\in L^1(\mu)$ such that 
$\sup_{x\in(x_0-\delta,x_0]}|\partial_2^-f(\omega,x)|\le g(\omega)$, $\omega\in\Omega$,
then 
\begin{align*}
\exists\s\partial_2^{-}\int_{\Omega}f(\omega,x)\,d\mu(\omega)\Bigg\vert_{x=x_0}
=
\int_{\Omega}\partial_2^-f(\omega,x)\big\vert_{x=x_0}\,d\mu(\omega).
\end{align*}
\end{enumerate}
\end{lemma}
\begin{proof}
We only prove \ref{integral partial1}, as the proof of \ref{integral partial2} goes exactly the same way.
Convexity of $f(\omega,\valt)$ yields that for every $\omega\in\Omega$, 
$\partial_2^+f(\omega,x)$ exists at $x\in[x_0,x_0+\delta)$, and
$f(\omega,x)=f(\omega,x_0)+\int_{x_0}^x\partial_2^+f(\omega,s)\,ds$, whence
\begin{align}\label{integral partial derivative proof1}
\abs{\frac{f(\omega,x)-f(\omega,x_0)}{x-x_0}}\le \sup_{x\in[x_0,x_0+\delta)}|\partial_2^+f(\omega,x)|\le g(\omega).
\end{align}
Thus, for any sequence $x_0<x_n\searrow x_0$,
\begin{align*}
&\frac{1}{x_n-x_0}\bz\int_{\Omega}f(\omega,x_n)\,d\mu(\omega)-\int_{\Omega}f(\omega,x_0)\,d\mu(\omega)\jz
=
\int_{\Omega}\underbrace{\frac{f(\omega,x_n)-f(\omega,x_0)}{x_n-x_0}}_{\to \partial_2^+f(\omega,x_0)}\,d\mu(\omega)\\
&\xrightarrow[n\to+\infty]{}
\int_{\Omega}\partial_2^+f(\omega,x)\big\vert_{x=x_0}\,d\mu(\omega),
\end{align*}
where the convergence in the second line follows by the Lebesgue dominated convergence theorem due to 
\eqref{integral partial derivative proof1}.
\end{proof}

\section{Joint lower semicontinuity of the Petz-type $f$-divergences in the finite dimensional case}

\begin{proposition}\label{prop:C.1}
Assume that $\hil$ is finite dimensional. Then for every $f\in(\bR^{(0,+\infty)})_{\mathrm{conv}}$ the Petz-type
divergence $D_f^{\mathrm{P}}(\rho\|\sigma)$ is jointly lower semicontinuous in $\rho,\sigma\in\B(\hil)\p$.
\end{proposition}

\begin{proof}
We may show for the Petz-type hockey stick divergence $D_{(\id-t)_+}^{\mathrm{P}}(\rho\|\sigma)$.
Let $d:=\dim\hil$. We write each $\rho\in\B(\hil)\p$ in the form
\[
\rho=\sum_{i=1}^d\lambda_i^\rho|v_i^\rho\rangle\langle v_i^\rho|,
\]
where $\lambda_1^\rho\ge\dots\ge\lambda_d^\rho$ are the eigenvalues of $\rho$ with multiplicities in the
non-increasing order and $\{v_1^\rho,\dots,v_d^\rho\}$ is an orthonormal basis of $\hil$ consisting of the
corresponding eigen vectors. Then one can represent $D_{(\id-t)_+}^{\mathrm{P}}(\rho\|\sigma)$ as
\[
D_{(\id-t)_+}^{\mathrm{P}}(\rho\|\sigma)
=\sum_{i,j=1}^d(\lambda_i^\rho-t\lambda_j^\sigma)_+|\langle v_i^\rho,v_j^\sigma\rangle|^2
\]
independently of the choice of $\{v_i^\rho\}$ and $\{v_i^\sigma\}$. Now, let
$\rho,\sigma,\rho_n,\sigma_n\in\B(\hil)\p$, $n\in\bN$, be such that $\|\rho_n-\rho\|\to0$ and
$\|\sigma_n-\sigma\|\to0$. It is well known that $(\lambda_i^{\rho_n})\to(\lambda_i^\rho)$ and
$(\lambda_i^{\sigma_n})\to(\lambda_i^\sigma)$ as $n\to\infty$. Let $\mathrm{ONB}(\hil)$ denote the set of
all orthonormal bases $\{v_i\}_{i=1}^d$ of $\hil$, which becomes a compact set with respect to the metirc
$d(\{v_i\},\{w_i\}):=\sum_{i=1}^d\|v_i-w_i\|$ for $\{v_i\},\{w_i\}\in\mathrm{ONB}(\hil)$. For any limit point
$\alpha$ of $\{D_{(\id-t)_+}^{\mathrm{P}}(\rho_n\|\sigma_n)\}_{n=1}^\infty$ one can choose a subsequence
$\{n_k\}$ of $\{n\}$ such that $D_{(\id-t)_+}^{\mathrm{P}}(\rho_{n_k}\|\sigma_{n_k})\to\alpha$ and
$\{v_i^{\rho_{n_k}}\}\to\{v_i\}$, $\{v_i^{\sigma_{n_k}}\}\to\{w_i\}$ in the above metric for some
$\{v_i\},\{w_i\}\in\mathrm{ONB}(\hil)$. Then we have
\[
\rho=\lim_k\rho_{n_k}=\lim_k\sum_{i=1}^d\lambda_i^{\rho_{n_k}}v_i^{\rho_{n_k}}
=\sum_{i=1}^d\lambda_i^\rho v_i,
\]
and similarly $\sigma=\sum_{i=1}^d\lambda_i^\sigma w_i$. Therefore,
\begin{align*}
\alpha&=\lim_kD_{(\id-t)_+}^{\mathrm{P}}(\rho_{n_k}\|\sigma_{n_k})
=\lim_k\sum_{i,j=1}^d(\lambda_i^{\rho_{n_k}}-t\lambda_j^{\sigma_{n_k}})_+
|\langle v_i^{\rho_{n_k}},v_j^{\sigma_{n_k}}\rangle|^2 \\
&=\sum_{i,j=1}^d(\lambda_i^\rho-t\lambda_j^\sigma)_+|\langle v_i,w_j\rangle|^2
=D_{(\id-t)_+}^{\mathrm{P}}(\rho\|\sigma),
\end{align*}
so that $D_{(\id-t)_+}^{\mathrm{P}}(\rho\|\sigma)$ is a unique limit point of
$\{D_{(\id-t)_+}^{\mathrm{P}}(\rho_n\|\sigma_n)\}$ and the assertion follows.
\end{proof}

\bibliography{bibliography251125}

\begin{thebibliography}{10}

\bibitem{AccardiCecchini1982}
L.~Accardi and C.~Cecchini.
\newblock Conditional expectations in von neumann algebras and a theorem of
  takesaki.
\newblock {\em J. Funct. Anal.}, 45:245--273, 1982.

\bibitem{AliSilvey}
S.~M. Ali and S.~D. Silvey.
\newblock A general class of coefficients of divergence of one distribution
  from another.
\newblock {\em Journal of the Royal Statistical Society: Series B
  (Methodological)}, 1966.

\bibitem{AnAn2026}
T.~Anastasiadis and G.~Androulakis.
\newblock Quantum $f$-divergences via {N}ussbaum-{S}zko\l a distributions in
  semifinite von {N}eumann algebras.
\newblock 2026.
\newblock arXiv:2604.19853.

\bibitem{Araki_relentrII}
H.~Araki.
\newblock Relative entropy of states of von {Neumann} algebras ii.
\newblock {\em Publ. Res. Inst. Math. Sci.}, 13(3):173--192, 1977.

\bibitem{ANSzV}
K.~M.~R. Audenaert, M.~Nussbaum, A.~Szkola, and F.~Verstraete.
\newblock Asymptotic error rates in quantum hypothesis testing.
\newblock {\em Communications in Mathematical Physics}, 279:251--283, 2008.

\bibitem{AD}
Koenraad M.~R. Audenaert and Nilanjana Datta.
\newblock $\alpha$-$z$-relative {R}enyi entropies.
\newblock {\em J.~Math.~Phys.}, 56:022202, 2015.

\bibitem{BHT_fdiv}
Salman Beigi, Christoph Hirche, and Marco Tomamichel.
\newblock Some properties and applications of the new quantum f-divergences.
\newblock arXiv:2501.03799, 2025.

\bibitem{BS}
V.~P. Belavkin and P.~Staszewski.
\newblock {$C^\ast$}-algebraic generalization of relative entropy and entropy.
\newblock {\em Ann. Inst. H. Poincar\'e Phys. Th\'eor.}, 37:51--58, 1982.

\bibitem{Bernstein_matrix}
D.~S. Bernstein.
\newblock {\em Matrix mathematics: theory, facts, and formulas}.
\newblock Princeton University Press, 2009.

\bibitem{BST}
Mario Berta, Volkher~B. Scholz, and Marco Tomamichel.
\newblock R\'enyi divergences as weighted non-commutative vector-valued
  ${L}_p$-spaces.
\newblock {\em Ann. Henri Poincar\'e}, 19:1843--1867, 2018.
\newblock arXiv:1608.05317.

\bibitem{Csiszar-fdiv}
Imre Csisz\'ar.
\newblock Information type measure of difference of probability distributions
  and indirect observations.
\newblock {\em Studia Sci.~Math.~Hungar.}, 2:299--318, 1967.

\bibitem{FawziFawzi2021}
Hamza Fawzi and Omar Fawzi.
\newblock Defining quantum divergences via convex optimization.
\newblock {\em {Quantum}}, 5:387, January 2021.

\bibitem{Frenkel_integral}
P\'eter~E. Frenkel.
\newblock Integral formula for quantum relative entropy implies data processing
  inequality.
\newblock {\em Quantum}, 7:1102, 2023.
\newblock arXiv:2208.12194.

\bibitem{Hiai2021}
F.~Hiai.
\newblock {\em {L}ectures on {S}elected {T}opics in von {N}eumann {A}lgebras}.
\newblock EMS Press, Berlin, 2021.

\bibitem{Hiai2021Div}
F.~Hiai.
\newblock {\em {Q}uantum $f$-{D}ivergences in von {N}eumann {A}lgebras:
  {R}eversibility of {Q}uantum {O}perations}.
\newblock Mathematical Physics Studies. Springer, Singapore, 2021.

\bibitem{HiaiKosaki}
F.~Hiai and H.~Kosaki.
\newblock Connections of unbounded operators and some related topics: von
  {N}eumann algebra case.
\newblock {\em Internat. J. Math.}, 32:2150024, 88 pp., 2021.

\bibitem{HiaiTsukada1984}
F.~Hiai and M.~Tsukada.
\newblock Strong martingale convergence of generalized conditional expectations
  on von neumann algebras.
\newblock {\em Trans. Amer. Math. Soc.}, 282:791--798, 1984.

\bibitem{Hiai_fdiv_standard}
Fumio Hiai.
\newblock Quantum $f$-divergences in von {N}eumann algebras {I}. {S}tandard
  $f$-divergences.
\newblock {\em J. Math. Phys.}, 59:102202, 2018.

\bibitem{Hiai_fdiv_maximal}
Fumio Hiai.
\newblock Quantum $f$-divergences in von {N}eumann algebras {II}. {M}aximal
  $f$-divergences.
\newblock {\em J. Math. Phys.}, 60:012203, 2019.

\bibitem{HiaiJencova2024}
Fumio Hiai and Anna Jen{\v c}ov\'a.
\newblock $\alpha$-$z$-{R}\'enyi divergences in von {N}eumann algebras: data
  processing inequality, reversibility, and monotonicity properties in
  $\alpha$, $z$.
\newblock {\em Commun. Math. Phys.}, 405:271, 2024.

\bibitem{HiaiMosonyi2017}
Fumio Hiai and Mil\'an Mosonyi.
\newblock Different quantum $f$-divergences and the reversibility of quantum
  operations.
\newblock {\em Reviews in Mathematical Physics}, 29(7):1750023, 2017.

\bibitem{sc_vN}
Fumio Hiai and Mil\'an Mosonyi.
\newblock Quantum {R\'enyi} divergences and the strong converse exponent of
  state discrimination in operator algebras.
\newblock {\em Annales Henri Poincar\'e}, 24:1681--1724, 2023.
\newblock arXiv:2110.07320.

\bibitem{HircheTomamichel_integral}
Christoph Hirche and Marco Tomamichel.
\newblock Quantum {R}\'enyi and $f$-divergences from integral representations.
\newblock {\em Commun. Math. Phys.}, 405(208), 2024.
\newblock arXiv:2306.12343.

\bibitem{Jencova_NCLp}
A.~Jen{\v c}ov{\'a}.
\newblock {R}\'enyi relative entropies and noncommutative ${L}_p$-spaces.
\newblock {\em Ann. Henri Poincar\'e}, 19:2513--2542, 2018.
\newblock arXiv:1609.08462.

\bibitem{Jencova_NCLpII}
A.~Jen{\v c}ov{\'a}.
\newblock {R}\'enyi relative entropies and noncommutative ${L}_p$-spaces {II}.
\newblock {\em Annales Henri Poincar\'e}, 22:3235--3254, 2021.
\newblock arXiv:1707.00047.

\bibitem{Jencova2023}
Anna Jen{\v c}ov\'a.
\newblock Recoverability of quantum channels via hypothesis testing.
\newblock {\em Letters in Mathematical Physics}, 114(1), 2024.
\newblock arXiv: 2303.11707.

\bibitem{Kadison1968}
R.~V. Kadison.
\newblock Strong continuity of operator functions.
\newblock {\em Pacific J. Math.}, 26:121--129, 1968.

\bibitem{Kosaki_conti}
H.~Kosaki.
\newblock On the continuity of the map $\varphi\to|\varphi|$ from the predual
  of a ${W}^*$-algebra.
\newblock {\em J. Funct. Anal.}, 59:123--131, 1984.

\bibitem{Kossmann}
Gereon Kossmann, Ren\'e Schwonnek, Po-Chieh Liu, and Hao-Chung Cheng.
\newblock Device-independent quantum key distribution in the commuting operator
  framework.
\newblock arXiv:2607.03579.

\bibitem{LiuHircheCheng2025}
Po-Chieh Liu, Christoph Hirche, and Hao-Chung Cheng.
\newblock Layer cake representations for quantum divergences.
\newblock arXiv:2507.07065, 2025.

\bibitem{Matsumoto_newfdiv}
K.~Matsumoto.
\newblock A new quantum version of $f$-divergence.
\newblock In {\em Nagoya Winter Workshop 2015: Reality and Measurement in
  Algebraic Quantum Theory}, pages 229--273, 2018.

\bibitem{mosonyi2022geometric}
Mil\'an Mosonyi, Gergely Bunth, and P{\'e}ter Vrana.
\newblock Geometric relative entropies and barycentric {R\'enyi} divergences.
\newblock {\em Linear Algebra and its Applications}, 699:159--276, 2024.
\newblock arXiv:2207.14282.

\bibitem{HiaiMosonyi_testmeasured}
Mil\'an Mosonyi and Fumio Hiai.
\newblock Test-measured r\'enyi divergences.
\newblock {\em IEEE Transactions on Information Theory}, 69(2):1074--1092,
  February 2023.
\newblock arXiv:2201.05477.

\bibitem{MO}
Mil\'an Mosonyi and Tomohiro Ogawa.
\newblock Quantum hypothesis testing and the operational interpretation of the
  quantum {R\'enyi} relative entropies.
\newblock {\em Communications in Mathematical Physics}, 334(3):1617--1648,
  2015.

\bibitem{Renyi_new}
Martin {M\"uller}-Lennert, Fr\'ed\'eric Dupuis, Oleg Szehr, Serge Fehr, and
  Marco Tomamichel.
\newblock On quantum {R\'enyi} entropies: A new generalization and some
  properties.
\newblock {\em Journal of Mathematical Physics}, 54(12):122203, December 2013.

\bibitem{NH}
Hiroshi Nagaoka and Masahito Hayashi.
\newblock An information-spectrum approach to classical and quantum hypothesis
  testing for simple hypotheses.
\newblock {\em IEEE Transactions on Information Theory}, 53(2):534--549,
  February 2007.

\bibitem{Nelson1974}
E.~Nelson.
\newblock {N}otes on non-commutative integration.
\newblock {\em J. Funct. Anal.}, 15:103--116, 1974.

\bibitem{NSz}
Michael Nussbaum and Arleta~Szko\l a.
\newblock The {C}hernoff lower bound for symmetric quantum hypothesis testing.
\newblock {\em The Annals of Statistics}, 37:1040--1057, 2009.

\bibitem{FGR2026}
L.~Gao O.~Fawzi and M.~Rahaman.
\newblock Asymptotic equipartition theorems in von {N}eumann algebras.
\newblock {\em Ann. Henri Poincar\'e}, 27:95--141, 2026.

\bibitem{Ogata2011}
Yoshiko Ogata.
\newblock A generalization of powers--st\o rmer inequality.
\newblock {\em Lett. Math. Phys.}, 97:339--346, 2011.

\bibitem{P85}
D.~Petz.
\newblock Quasi-entropies for states of a von {N}eumann algebra.
\newblock {\em Publ. Res. Inst. Math. Sci.}, 21:781--800, 1985.

\bibitem{P86}
D\'enes Petz.
\newblock Quasi-entropies for finite quantum systems.
\newblock {\em Reports in Mathematical Physics}, 23:57--65, 1986.

\bibitem{Petz1986}
D\'enes Petz.
\newblock Sufficient subalgebras and the relative entropy of states of a von
  {N}eumann algebra.
\newblock {\em Commun. Math. Phys.}, 105:123--131, 1986.

\bibitem{Petz1988}
D\'enes Petz.
\newblock Sufficiency of channels over von {Neumann} algebras.
\newblock {\em Quart.~J.~Math.~Oxford Ser.~(2)}, 153:97--108, 1988.

\bibitem{daSilva2026}
G.~Lechner R.C.~da Silva, M.B.~Fr{\"o}b and L.~Sangaletti.
\newblock Integral representations of $f$-divergences for general von {N}eumann
  algebras.
\newblock 2026.
\newblock arXiv:2607.05195.

\bibitem{Renyi}
Alfr\'ed R\'enyi.
\newblock On measures of entropy and information.
\newblock In {\em Proc.~4th Berkeley Sympos.~Math.~Statist.~and Prob.},
  volume~I, pages 547--561. Univ. California Press, Berkeley, California, 1961.

\bibitem{Sason_Verdu_fdiv2016}
Igal Sason and Sergio Verdú.
\newblock $f$-divergence inequalities.
\newblock {\em IEEE Transactions on Information Theory}, 62(11):5973--6006,
  2016.

\bibitem{Stratila1981}
S.~Str{\u a}til{\u a}.
\newblock {\em {M}odular {T}heory in {O}perator {A}lgebras}.
\newblock Editura Academiei and Abacus Press, Tunbridge Wells, 1981.

\bibitem{takesaki1973}
M.~Takesaki.
\newblock Duality for crossed products and the structure of von {N}eumann
  algebras of type iii.
\newblock {\em Acta Math.}, 131:249–310, 1973.

\bibitem{Takesaki}
M.~Takesaki.
\newblock {\em Theory of operator algebras I}, volume 124 of {\em Encyclopaedia
  of Mathematical Sciences}.
\newblock Springer-Verlag, 2002.

\bibitem{Terp1981}
M.~Terp.
\newblock ${L}^p$ spaces associated with von {N}eumann algebras.
\newblock {N}otes, {C}openhagen {U}niversity, 1981.

\bibitem{Umegaki}
H.~Umegaki.
\newblock Conditional expectation in an operator algebraalgebra, {IV}: Entropy
  and information.
\newblock {\em Kodai Math.~Sem.~Rep.}, 14:59--85, 1962.

\bibitem{vLuWi2026}
L.~van Luijk and H.~Wilming.
\newblock Sufficiency and {P}etz recovery for positive maps.
\newblock 2026.
\newblock arXiv:2604.08380.

\bibitem{WWY}
Mark~M. Wilde, Andreas Winter, and Dong Yang.
\newblock Strong converse for the classical capacity of entanglement-breaking
  and {Hadamard} channels via a sandwiched {R\'enyi} relative entropy.
\newblock {\em Communications in Mathematical Physics}, 331(2):593--622,
  October 2014.

\end{thebibliography}

\end{document}